\documentclass{article}

\usepackage{PRIMEarxiv}

\usepackage[utf8]{inputenc}
\usepackage[T1]{fontenc}
\usepackage{multirow}

\usepackage{amsmath}
\usepackage{amssymb}
\usepackage{amsfonts}
\usepackage{amsthm}

\usepackage[round]{natbib}

\usepackage{hyperref}
\usepackage{placeins}
\usepackage{url}
\usepackage{mathtools}

\usepackage{booktabs}
\usepackage{nicefrac}
\usepackage{microtype}

\usepackage{graphicx}
\usepackage{subcaption}
\usepackage{caption}
\usepackage{float}

\usepackage{xcolor}
\usepackage[english]{babel}

\usepackage{algorithm}
\usepackage[noend]{algpseudocode}
\algrenewtext{EndFor}{\textbf{end}}

\graphicspath{{WSD/}{images/}{figures/}}

\usepackage{fancyhdr}
\theoremstyle{definition}
\newtheorem{definition}{Definition}[section]
\theoremstyle{plain}
\newtheorem{theorem}{Theorem}[section]
\newtheorem{lemma}[theorem]{Lemma}
\newtheorem{proposition}[theorem]{Proposition}
\newtheorem{corollary}{Corollary}
\theoremstyle{remark}
\newtheorem{remark}{Remark}

\newcommand{\R}{\mathbb{R}}

\newcommand{\bx}{\mathbf{x}}
\newcommand{\calC}{\mathcal{C}}
\newcommand{\bI}{\boldsymbol{I}}
\newcommand{\calB}{\mathcal{B}}
\newcommand{\bP}{\mathbf{P}}

\title{Relative Wasserstein Spatial Depth for Cluster Number Selection
in Distributional Data}

\author{
  Chenqin Xu \\
  Columbia University \\
  \texttt{cx2328@columbia.edu} \\
  \And
  Yisha Yao\thanks{Corresponding author.} \\
  Columbia University \\
  \texttt{yy3381@columbia.edu} \\
}
\begin{document}
\maketitle

\begin{abstract}
%Modern datasets are increasingly represented as probability distributions, motivating clustering methods that operate in Wasserstein space. While recent advances extend classical clustering to distribution-valued data, most existing approaches assume that the number of clusters is known a priori, an assumption that rarely holds in practice. In this work, we address the problem of cluster number selection and robust clustering in Wasserstein space. We propose a new criterion, Relative Wasserstein Spatial Depth (RWSD), which leverages Wasserstein geometry and depth-based centrality to provide a robust and interpretable measure of cluster separation tailored to distributional data.  In addition, we develop a Wasserstein K-median clustering algorithm that improves robustness to heavy-tailed variability and outliers.  Empirical results demonstrate that the proposed method consistently outperforms classical distance-based criteria, including the silhouette score and Davies–Bouldin index, particularly in challenging settings with outliers, overlapping clusters, and heterogeneous cluster structures.

Contemporary data in many scientific domains, such as images, media streams, biomedical omics, are naturally modeled as probability distributions in Wasserstein space instead of points in Euclidean space.  Consequently, clustering methods for distributions are in high demand, as they provide exploratory frameworks to group objects with similarity. 
Clustering distribution-valued data in Wasserstein space often requires choosing the number of clusters in advance. We propose Relative Wasserstein Spatial Depth (RWSD) as a selection criterion for optimal number of clusters.  It compares the depth of each distribution in its assigned cluster with its depth in competing clusters. We pair this criterion with Wasserstein \(K\)-means or a proposed \(K\)-medians algorithm. Under suitable assumptions, we establish consistency of the cluster centers and of the selected optimal number of clusters. We also give an error rate under two-stage sampling and characterize the robustness of RWSD under any fixed partition. In extensive simulations, RWSD selects the intended number of clusters under unequal separations, heterogeneous dispersion, heavy tails, and outliers, where silhouette score or the Davies–Bouldin Index often fails. Applications to MNIST images and flow cytometry data illustrate the practical value of RWSD. 
In some cases where a cluster is not geodesically convex and lacks a meaningful depth center, the RWSD can fail.
%The classical $K$-mean clustering, which groups the data points based on their relative distance, seems to have a natural extension to Wasserstein space via replacing the Euclidean distance by Wasserstein distance. 
%The curvature inherent to Wasserstein geometry and lack of linear structure complicate extensions of Euclidean $K$-mean methods, both in defining cluster centers and in selecting the number of clusters. Existing works on Wasserstein $K$-mean clustering assume that the number of clusters $K$ is known, which is rarely the case in many applications. In this paper, we develop a criterion to determine the optimal $K$, which is based on the recently proposed Wasserstein Spatial Depth (WSD). 
\end{abstract}

\section{Introduction}
Contemporary data are increasingly complex and multi-faceted.
In many applications, each data object is more naturally represented as a probability distribution rather than a vector in Euclidean space, such as distributions of gene expression levels across spatial locations, connectivity profiles across brain regions, high-resolution images, and media streams. 
Such representation captures the intrinsic variability and geometric structure of the data, while providing a low-dimensional embedding in the Wasserstein space \citep{panaretos2019statistical,bachoc2024wasserstein}.

As an important exploratory analysis technique, clustering generally groups objects based on their similarity. 
When data objects are distributions, classical clustering methods based on Euclidean geometry may fail to capture meaningful differences in shape, dispersion, or support.
Clustering techniques adapting to a space of probability distributions are in high demand, which respect the Wasserstein geometry.
%Let $\mathbf{P}\in\mathcal{P}(\mathcal{P}_2(\mathbb{R}^d))$ denote a \emph{distribution of distributions},  that is, a probability measure whose elements are themselves distributions in the Wasserstein space.  Our objective is to analyze the structure of $\mathbf{P}$ by identifying coherent clusters among its constituent distributions.

\paragraph{Related work}
Several efforts have been made to devise clustering techniques for distribution-valued data. 
Trimmed $K$-barycenters method~\citep{del2019robust} improves stability and robustness by trimming the most discrepant distributions. 
A modified Wasserstein distance proposed in~\citep{Verdinelli2019hybrid} is computationally more tractable and accelerates various clustering procedures.
\citep{10.1214/21-EJS1834} take a different approach: they represent
each distribution as a finite-dimensional vector and then apply
classical clustering methods.
The Wasserstein $K$-means clustering framework \citep{zhuang2022wasserstein} generalizes classical $K$-means to Wasserstein space through a distance-based formulation
\begin{equation*}%\label{eq:w2-kmeans}
\min_{ \{\calC_k\}_{k=1}^K} \Big\{
\sum_{k=1}^{K} \frac{1}{|\calC_k|}\sum_{\mu_i, \mu_j \in \calC_k}
\mathcal{W}_2^2(\mu_i, \mu_j), \ \cup_{k+1}^K \calC_k = [n] \Big\},
\end{equation*}
and proposes a semidefinite program relaxation to approximate the solution.
%This is because the Wasserstein barycenter function is neither geodesically convex nor continuous in Wasserstein space. 
Spectral clustering is also generalized to Wasserstein space by constructing an adjacency matrix of pairwise Wasserstein distances~\citep{wang2024spectral}.
Together, these works open a new direction for unsupervised pattern discovery in Wasserstein space.

\paragraph{Motivation}
All these clustering methods assume that the number of clusters $K$ is known \emph{a priori}, which is not realistic in many applications. Selecting a suitable $K$ remains a challenge in practice \citep{still2004many}.

For clustering in Euclidean space, several measurements (statistics) have been proposed to pick $K$ and assess how statistically meaningful a partition is. 
\emph{Silhouette Score} \citep{rousseeuw1987silhouettes} is among the most widely used, which evaluates how similar each point is to its own cluster compared to other clusters. 
The Wasserstein version of silhouette score can be defined as follows.
Given a particular clustering scheme $\{ \calC_1, \ldots, \calC_K \}$, a distribution $\mu_i$ assigned to $\calC_{k_i}$ has silhouette score being the difference between its within-cluster dissimilarity and minimum out-of-cluster dissimilarity,
\begin{equation}\label{eq:silhouette}
s_i = \frac{b_i - a_i}{\max(a_i, b_i)},
\end{equation}
where $a_i = \frac{1}{|\calC_{k_i}|-1}\sum_{\mu_j \in \calC_{k_i}} \mathcal{W}_2(\mu_i,\mu_j)$
is the average dissimilarity of $\mu_i$ to the other members of its own
cluster and $b_i = \min_{\ell \ne k_i} \frac{1}{|\calC_\ell|}\sum_{\mu_j \in \calC_\ell} \mathcal{W}_2(\mu_i,\mu_j)$
is its smallest average dissimilarity to any other cluster.
%The silhouette score of $\mu_i$ is then given by $s_i = (b_i - a_i)/\max(a_i,b_i)$, which satisfies $s_i \in [-1,1]$.
A larger $s_i$ indicates better separation of $\calC_{k_i}$ from other clusters. 
In practice, one computes the average silhouette score
$\bar{s}(K) = \tfrac{1}{n}\sum_{i=1}^n s_i(K)$ for different $K$'s, and selects $\widehat{K} = \arg\max_K \ \bar{s}(K)$.
The silhouette score is sensitive to overlapping or 
nonconvex clusters, as it relies on averages of pairwise distances \citep{liu2010understanding}.
Another widely used measurement is the \emph{Davies--Bouldin Index (DBI)} \citep{davies1979cluster}, which evaluates the average of ratios between within-cluster scattering and between-cluster separation.  The Wasserstein version of DBI is defined as 
\begin{equation}\label{eq:dbi}
\begin{aligned}
\mathrm{DBI}(K) &= \frac{1}{K}\sum_{k=1}^{K}\ \max_{\ell \ne k}\
  \frac{R_k + R_\ell}{\mathcal{W}_2(\bar{\mu}_k,\, \bar{\mu}_\ell)}, \\
R_k &= \frac{1}{|\calC_k|}\sum_{\mu_i \in \calC_k} \mathcal{W}_2(\mu_i,\, \bar{\mu}_k),
\end{aligned}
\end{equation}
where $R_k$ quantifies the within-cluster scattering of $\calC_k$ and $\bar{\mu}_k$ is the representative of $\calC_k$ (e.g. its Wasserstein barycenter). 
A smaller DBI indicates tighter and more separated clusters, and the best $K$ is the one that minimizes DBI.   
DBI is sensitive to outliers and blurred cluster boundaries \citep{liu2010understanding,zaki2014data}, and it can be unstable when the cluster scale changes \citep{halkidi2001clustering}.
In Euclidean space, \citet{jornsten2002robust} proposed the relative $L_1$-depth 
for cluster validation.

Although these measurements can be extended to Wasserstein space by replacing Euclidean distances with $\mathcal{W}_2$, they suffer from limited robustness (being sensitive to outliers) and may neglect the rich geometry of Wasserstein space. These limitations motivate the development of a robust and interpretable measure tailored to Wasserstein space.

\paragraph{Our contribution.}
We introduce a new criterion, \emph{Relative Wasserstein spatial depth}  (RWSD), for selecting the number of clusters $K$ when performing Wasserstein $K$-center clustering.   
We also develop a \emph{Wasserstein $K$-median clustering} algorithm that is robust to outliers and clusters with heavy tailed dispersion.

%\textcolor{blue}{We devise a robust unsupervised clustering procedure for probability distributions lying within Wasserstein space. Each cluster is represented by the Wasserstein median instead of the barycenter. }

\section{Methodology} 
Throughout the paper, let $\mathcal{P}_2(\mathbb{R}^d)$ be the collection of probability distributions on $\mathbb{R}^d$ with finite second moment, among which the subset of absolutely continuous distributions is denoted as
$\mathcal{P}_2^{a.c}(\mathbb{R}^d)$. Let
$\mathcal{P}(\mathcal{P}_2(\mathbb{R}^d))$ be the collection of distributions of distributions in $\mathcal{P}_2(\mathbb{R}^d)$. 

\subsection{Wasserstein spatial depth}\label{sec:intro_to_WSD}
Classical statistical depth can be viewed as a multivariate analogue of the classical quantiles or ranks for univariate data. It provides a robust measure of centrality and outlyingness for multivariate data or functional observations \citep{zuo2000general}. Unlike distance-based measures, depth functions are affine-invariant and resistant to outliers, making them particularly suitable for analyzing heterogeneous data structures.
A new notion of depth defined on Wasserstein space is recently proposed in \citep{bachoc2024wasserstein}, named as {\it Wasserstein spatial depth} (WSD). 
The WSD of a distribution $Q \in \mathcal{P}_2^{a.c}(\mathbb{R}^d)$ with respect to a population of distributions $\mathbf{P} \in \mathcal{P}(\mathcal{P}_2(\mathbb{R}^d))$ is defined as
\begin{definition}[Wasserstein spatial depth  \citep{bachoc2024wasserstein}]\label{def:WSD}
\begin{equation*}
    \mathrm{WSD}(Q;\mathbf{P}) := 1 - \left( \int \left\| 
\mathbb{E}_{P \sim \mathbf{P}} \!\left[ 
\frac{x - T_{Q,P}(x)}{\mathcal{W}_2(Q,P)} \right] 
\right\|^2 dQ(x) \right)^{1/2}.
\end{equation*}
When $\mathcal{W}_2(P,Q)=0$, we set $\frac{x - T_{Q,P}(x)}{\mathcal{W}_2(P,Q)}=0$ by default.
\end{definition}
It provides a way to rank and order distributions, enabling the development of rank-based clustering techniques and inferential tools \citep{yaowasserstein} for Wasserstein space.  WSD is shown to preserve critical properties of conventional statistical depths, notably, ranging within $[0,1]$, transformation and geodesic invariance, vanishing at infinity, reaching a maximum at the geometric median, and continuity. 
Additionally, WSD has a straightforward plug-in estimator based on sampled empirical distributions, which is consistent and asymptotically normal.

In practice, we only have access to empirical distributions
$\hat\mu_{1,m}, \ldots, \hat\mu_{n,m}$ obtained by two-stage sampling: first
draw $\mu_1, \ldots, \mu_n \stackrel{iid}{\sim} \mathbf{P}$, and then
draw $m$ data points $X_{i,1}, \ldots, X_{i,m}$ from each $\mu_i$.
For $\ell \neq i$, let $T_{i\ell} := T_{\hat\mu_{i,m},\, \hat\mu_{\ell,m}}$
be the empirical ``optimal transport map'' from $\hat\mu_{i,m}$ to
$\hat\mu_{\ell,m}$ (the barycentric projection of the exact discrete
optimal plan; implementation details are deferred to
Appendix~\ref{app:wsd-impl}), write
$W_{i\ell} := \mathcal{W}_2(\hat\mu_{i,m}, \hat\mu_{\ell,m})$, and let
$\bar{u}_{i,j} := \frac{1}{n}\sum_{\ell \neq i}
\big(X_{i,j} - T_{i\ell}(X_{i,j})\big)/W_{i\ell}$
be the averaged normalized transport displacement of the $j$-th atom
of $\hat\mu_{i,m}$. The empirical WSD is then computed as
\begin{equation}\label{eq:empirical-WSD}
\mathrm{WSD}(\hat\mu_{i,m};\, \hat{\mathbf{P}}_{n,m})
= 1 - \Big(\frac{1}{m}\sum_{j=1}^{m} \|\bar{u}_{i,j}\|^2\Big)^{1/2}.
\end{equation}

\subsection{Relative Wasserstein spatial depth}\label{sec:propose_RWSD}
Building on the concept of WSD, we propose \emph{Relative Wasserstein Spatial Depth} (RWSD) as a robust index for selecting an optimal number of clusters $K$. 
The RWSD of a particular distribution $Q$ is defined as 
\begin{equation}\label{eq:relative-depth}
\mathrm{RWSD}(Q) = \mathrm{WSD}(Q;\mathcal{C}_{k_Q})-
    {\displaystyle \max_{j\neq k_Q}\mathrm{WSD}(Q;\mathcal{C}_j)},
\end{equation}
where $k_Q$ indicates which cluster $Q$ belongs to.  
RWSD compares within- and between-cluster depths, emphasizing centrality and separation rather than distance or variance. 
A larger $\mathrm{RWSD}(Q)$ indicates that $Q$ is considerably deeper within its own cluster than in others, implying a confident assignment.

Let $Q_1, Q_2, \cdots, Q_n$ be the set of distributions to be clustered. 
The optimal number of clusters $K^\ast$ can be determined by maximizing the average RWSD
\begin{equation}
\widehat K = \arg\max_K
\frac{1}{n}\sum_{i=1}^{n}\mathrm{RWSD}_K(Q_i),
\label{eq:optK-ReD}
\end{equation}
where $\mathrm{RWSD}_K(Q_i)$ denotes the relative WSD of $Q_i$ under the clustering solution with $K$ clusters.
This criterion favors partitions that are both internally cohesive and externally well separated in terms of Wasserstein geometry.
Compared to classical distance-based indices (silhouette score and DBI), RWSD respects Wasserstein geometry and is invariant under common transformations (\textit{e.g.}, scaling). Additionally, it is robust to unbalanced dispersions or separations, and overlapping clusters (shown in Section~\ref{sec:simulations}). The empirical version of RWSD is computed by simply plugging in the empirical WSD. 

To implement RWSD, we run Algorithm \ref{alg:selectK}.
\begin{algorithm}[H]
\caption{Select Optimal Number of Clusters $K$ by RWSD}
\label{alg:selectK}
\begin{algorithmic}[1]
  \State \textbf{Input:} a collection of probability $\{\mu_1,\dots,\mu_n\}\subset\mathcal P_2(\R^d)$, a candidate set $\mathcal{K}$ containing possible numbers of clusters
  \For{each $K\in\mathcal{K}$}
    \State Run Wasserstein $K$-centers to partition $\{\mu_i\}_{i=1}^n$ into $\calC_1, \calC_2, \ldots, \calC_K$. 
    \For{$i = 1$ to $n$}
      \State $\mathrm{RWSD}_K(\mu_i) \gets \mathrm{WSD}(\mu_i; \calC_{k_i}) - {\displaystyle \max_{j\neq k(\mu_i)}\mathrm{WSD}(\mu_i;\calC_j)}$
    \EndFor
    \State $\overline{\mathrm{RWSD}}(K) \gets \dfrac{1}{n}\sum_{i=1}^n \mathrm{RWSD}_K(\mu_i)$
  \EndFor
  \State $\widehat K \gets \arg\max_{K\in\mathcal{K}}\;\overline{\mathrm{RWSD}}(K)$
    \State \textbf{Output:} the optimal cluster count $\widehat K$
\end{algorithmic}
\end{algorithm}
In the above procedure, the ``Wasserstein $K$-centers step'' can be either Wasserstein $K$-means or Wasserstein $K$-medians. The details are deferred to Appendix~\ref{app:methods}. The complexity of Algorithm~\ref{alg:selectK} is described in Appendix~\ref{app:complexity}.

\section{Theoretical Guarantee}\label{sec:theory}
%The empirical $\mathrm{RWSD}_K$ is a function of the empirical $K$-partition and the empirical WSD. Hence, the  deviation of empirical RWSD from its population version comes from the randomness in the empirical Wasserstein $K$-means clustering and empirical WSD. Establishing the consistency of empirical RWSD therefore requires two ingredients: (a) consistency of the empirical WSD, which is established in \citet[Theorem~6.3]{bachoc2024wasserstein}; (b) consistency of Wasserstein $K$-means clustering, \textit{i.e.}, convergence of the empirical cluster centers to the well-defined population analogies as $n\rightarrow\infty$, a strong consistency in the sense of \citet{pollard1981strong}. Theorem~\ref{thm:W2-kmeans} establishes the consistency of Wasserstein $K$-means.  Theorem~\ref{thm:rwsd-abstract} shows that the empirical RWSD is a consistent estimator of its population counterpart if the clustering is consistent.  
The empirical $\mathrm{RWSD}_K$ depends on both the estimated $K$-partition and the empirical WSD. We first establish strong consistency of the cluster centers for Wasserstein $K$-means (Theorem~\ref{thm:W2-kmeans}) and $K$-medians (Theorem~\ref{thm:W1-kmedians}). Combining consistent clustering with the consistency of empirical WSD established by \citet[Theorem~6.3]{bachoc2024wasserstein}, we show that empirical RWSD converges to its population counterpart (Theorem~\ref{thm:rwsd-abstract}). Under a unique population maximizer and a positive separation margin, this yields consistency of the selected number of clusters. We then quantify the non-asymptotic error due to sampling $n$ distributions from the population $\mathbf{P}$ and approximating each distribution by $m$ data points. We also characterize the robustness of RWSD is via bounding its influence function. 

Let $(X,d)$ be a Polish metric space and let $\mathcal{P}_2(X)$ denote the space of 
Borel probability measures on $X$ with finite second moment, endowed with the $2$-Wasserstein distance.
%\begin{equation}
%\mathcal{W}_2(\mu,\nu) := \left(\inf_{\pi\in\Pi(\mu,\nu)} \int d(x,y)^2\, d\pi(x,y) \right)^{1/2},
%\end{equation}
%where $\Pi(\mu,\nu)$ denotes the set of all couplings between $\mu$ and $\nu$.
We observe i.i.d.\ random objects taking values in $\mathcal{P}_2(X)$
\[
\mu_1,\mu_2,\ldots, \mu_n \overset{iid}{\sim} \mathbf{P}, \quad\mathbf{P}\in \mathcal{P}(\mathcal{P}_2(X))
\]
with the empirical version of $\mathbf{P}$ denoted as $\mathbf{P}_n := \frac{1}{n}\sum_{i=1}^n \delta_{\mu_i}$. 

Given a finite set $A\subset\mathcal{P}_2(X)$ and a distribution of distributions 
$\mathbf{P}\in \mathcal{P}(\mathcal{P}_2(X))$, consider the function defined in \citep{pollard1981strong}
\begin{equation}\label{eq:Phi-def}
\Phi(A,\mathbf{P})
   := \int \min_{\nu\in A} \mathcal{W}_2^2(\mu,\nu)\, \mathbf{P}(d\mu).
\end{equation}
For $K \ge 1$, let $m_K(\mathbf{P}) := \inf_{|A|\le K}\Phi(A,\mathbf{P})$
denote the minimal value of $\Phi(\,\cdot\,,\mathbf{P})$ over sets of
cardinality at most $K$, and let
$ A_K^\star \in \arg\min_{|A|\le K}\Phi(A,\mathbf{P})$ be a population
minimizer. Their empirical versions $m_K(\mathbf{P}_n)$ and
$\widehat A_{n,K} \in \arg\min_{|A|\le K}\Phi(A,\mathbf{P}_n)$ are
defined in the same way with $\mathbf{P}_n$ in place of $\mathbf{P}$.

Throughout this section, we assume the base space $(X,d)$ is compact, and hence, the Wasserstein space $\mathcal{P}_2(X)$ is itself compact under $\mathcal{W}_2$. It further implies that the above
population and empirical minimizers exist.

\begin{theorem}[Consistency of Wasserstein $k$-means]\label{thm:W2-kmeans}
Assume:
\begin{enumerate}
%\item[(i)] $(X,d)$ is a compact metric space. Hence $\mathcal{P}_2(X)$ is compact under the $\mathcal{\mathcal{W}_2}$ metric.

\item[(i)]
For some (and therefore every) $\mu_0\in\mathcal{P}_2(X)$,
\begin{equation*}
    \int_{\mathcal{P}_2(X)} \mathcal{W}_2(\mu,\mu_0)^2\,\bP(d\mu) < \infty.
\end{equation*}

\item[(ii)]
For each $k=1,\ldots,K$, the population minimizer $A^\star_k$ satisfying $\Phi(A^\star_k,\bP) = m_k(\bP)$ is unique up to permutation.

\item[(iii)]
The population $k$-means values are strictly separated across consecutive model sizes:
\begin{equation*}
    m_{k-1}(\bP) > m_k(\bP),
    \qquad k=2,\ldots,K .
\end{equation*}
\end{enumerate}
Then the empirical $k$-means centers are strongly consistent:
\[
\widehat A_{n,k} \longrightarrow A^\star_k
\quad\text{a.s.\ in the Hausdorff metric},
\]
and the empirical optimal values satisfy
\[
\Phi(\widehat A_{n,k}, \bP_n)\to m_k(\bP)
\qquad\text{a.s.}
\]
\end{theorem}

The argument of Theorem~\ref{thm:W2-kmeans} can be naturally extended to establish the consistency of Wasserstein $K$-medians, whose objective is
\[
\Phi_1(A, \mathbf{P}) := \int \min_{\nu \in A} \mathcal{W}_2(\mu, \nu)\,d\mathbf{P}(\mu).
\]
Let $m^{(1)}_k(\mathbf{P}) := \inf_{|A| \le k} \Phi_1(A, \mathbf{P})$,
and let $A^{(1)}_k$ and $A^{(1)}_{n,k}$ be population and empirical
minimizers of $\Phi_1$ over sets of cardinality at most $k$.

\begin{theorem}[Consistency of Wasserstein $K$-medians] \label{thm:W1-kmedians}
Suppose the Wasserstein median of any cluster under consideration is unique.
Under the same set of assumptions in Theorem~\ref{thm:W2-kmeans} except replacing $\Phi(A, \mathbf{P})$ by $\Phi_1(A, \mathbf{P})$, 
%Assume that, for each $k \le K$, the population $K$-medians configuration $A^{(1)}_k$ is unique up to permutation. Further assuming $m^{(1)}_{k-1}(\mathbf{P}) > m^{(1)}_k(\mathbf{P})$ for $k = 2, \ldots, K$,
it holds that $A^{(1)}_{n,k} \to A^{(1)}_k$ almost surely in the Hausdorff metric, and
$\Phi_1(A^{(1)}_{n,k}, \mathbf{P}_n) \to m^{(1)}_k(\mathbf{P})$ almost surely. 
%Consequently, Theorem~\ref{thm:rwsd-abstract} and Corollary~\ref{cor:Kstar-consistency} hold for the $K$-medians pipeline under the same uniqueness assumption.
\end{theorem}

\begin{remark}\label{rem:kmedians}
In general, the uniqueness of the Wasserstein median remains open. We list some special cases in Appendix~\ref{app:kmedians-proof}, where the Wasserstein median is indeed unique. 
%The uniqueness of the population $K$-medians configuration is an identification assumption. Unlike its $K$-means counterpart, which is supported by the uniqueness of Wasserstein barycenters \citep{agueh2011barycenters}, it has no general sufficient condition: uniqueness of a Wasserstein median is open in general and fails for two-member clusters. We verify it at the level of single clusters for location families and for nondegenerate Gaussians with commuting covariances, where the problem reduces to Euclidean $K$-medians (Appendix~\ref{app:kmedians-proof}); outside such families its validity is open.
\end{remark}

We next show that whenever the Wasserstein $K$-center is consistent in the sense of Theorem~\ref{thm:W2-kmeans} for $K$-mean and Theorem~\ref{thm:W1-kmedians} for $K$-median, the corresponding empirical RWSD is a consistent estimator of its population counterpart. %This abstracts the role of the clustering step and applies to any clustering procedure satisfying the consistency hypothesis below.
\begin{theorem}[Consistency of empirical RWSD under consistent clustering]
\label{thm:rwsd-abstract}
Let $\mathbf{P} \in \mathcal{P}(\mathcal{P}_2(\mathbb{R}^d))$ be atomless, and suppose 
the Wasserstein $K$-center clustering method produces empirical cluster centers 
$\{\widehat{A}_{n,k}\}_{k=1}^{K}$ that converge almost surely in the Hausdorff metric induced by $\mathcal{W}_2$ to a set of unique population centers $\{A_k^\star\}_{k=1}^{K}$. 
Assume further that
%the population partition $\{\mathcal{C}_k\}_{k=1}^{K}$ induced by $\{A_k^\star\}_{k=1}^{K}$ is well-separated, in the sense that 
$\mathbf{P}$-almost every $\mu$ has a unique nearest center in $A_k^\star$ for every $k=1, \ldots, K$. 
%for the partition $\{\hat{\mathcal{C}}_k^{(n,m)}\}_{k=1}^{K}$ induced by $\{\hat{A}_{n,k}\}_{k=1}^{K}$ and the corresponding 
Then the empirical RWSD satisfies
\[
\mathrm{RWSD}_K(\hat\mu_{i,m}) \;\xrightarrow{P}\; \mathrm{RWSD}_K(\mu_i)
\quad \text{as } n, m \to \infty,
\]
for $\mathbf{P}$-almost every $\mu_i$. 
\end{theorem}

Combining Theorem~\ref{thm:W2-kmeans} and Theorem~\ref{thm:rwsd-abstract} immediately yields the consistency of empirical RWSD induced by Wasserstein $K$-means clustering.
% \textcolor{red}{Corollary 1 repeats similar result as Theorem 3.2?}

% \begin{corollary}[Consistency of empirical RWSD under Wasserstein $K$-means]
% \label{cor:rwsd-consistency}
% Under the conditions of Theorem~\ref{thm:W2-kmeans} and Theorem~\ref{thm:rwsd-abstract}, 
%suppose in addition that $\mathbf{P}$ is atomless and that the population partition induced by the unique $K$-means centers is well-separated in the sense of Theorem~\ref{thm:rwsd-abstract}. 
% the empirical RWSD satisfies
% \[
% \mathrm{RWSD}_K(\hat\mu_{i,m}) \;\xrightarrow{P}\; \mathrm{RWSD}_K(\mu_i)
% \quad \text{as } n, m \to \infty,
% \]
% for $\mathbf{P}$-almost every $\mu_i$.
% \end{corollary}
Let $\mathcal{K}$ be a finite set of candidate numbers of clusters with
$\min\mathcal{K} \ge 2$. For each $K \in \mathcal{K}$, the selection
procedure clusters the $n$ objects into $K$ groups, evaluates
the empirical relative depth of every object, and records the average
$\widehat{R}_{n,m}(K) := \frac{1}{n}\sum_{i=1}^{n}
\mathrm{RWSD}_K(\hat\mu_{i,m})$; the selected number of clusters is
$\widehat{K}_{n,m} := \arg\max_{K \in \mathcal{K}} \widehat{R}_{n,m}(K)$.
The population counterpart of the criterion is
$R(K) := \mathbb{E}_{\mu \sim \mathbf{P}}\big[\mathrm{RWSD}_K(\mu)\big]$, and  define
\[
\ K^\star := \arg\max_{K \in \mathcal{K}} R(K)
\]
as the population optimal number of clusters. Since $\mathcal{K}$ is
finite, applying Theorem~\ref{thm:rwsd-abstract} to every $K\in \mathcal K$, together with a margin condition on $R$, we establish the consistency of the selected $K$.

\begin{corollary}[Consistency of the selected number of clusters] \label{cor:Kstar-consistency}
Suppose that $K^\star$ is the unique maximizer of $R$ over $\mathcal{K}$, with margin
\begin{equation}\label{eq:margin}
\delta \;:=\; R(K^\star) - \max_{K^ \in \mathcal{K} \setminus \{K^\star\}} R(K)
\;>\; 0 .
\end{equation}
Under the conditions of Theorem~\ref{thm:W2-kmeans} (Theorem~\ref{thm:W1-kmedians}) and Theorem~\ref{thm:rwsd-abstract}, it holds
\[
\mathbb{P}\big(\widehat{K}_{n,m} = K^\star\big) \;\longrightarrow\; 1 
\]
as $n, m \to \infty$.
%so that the Wasserstein $K$-means (medians) satisfy the hypothesis of Theorem~\ref{thm:rwsd-abstract}; that $\mathbf{P}$ is atomless; and that the population partition induced at every level $K \in \mathcal{K}$ is well-separated in the sense of Theorem~\ref{thm:rwsd-abstract}.
\end{corollary}

Next we quantify the error rate of the empirical RWSD. Under the two-stage sampling scheme, error comes from two sources: sampling
finitely many distributions from the population, and approximating each sampled distribution by finitely many data points. The following proposition bounds both errors: the first attains a parametric rate $n^{-1/2}$, and the second equals the optimal rate of estimating OT map from $m$ points.
%Under two-stage sampling, the empirical RWSD has two sources of error: each cluster contains finitely many distributions, and each distribution is observed through finitely many points. The next result bounds them separately; the first attains the parametric rate $n^{-1/2}$, and the second contributes the rate of estimating OT maps from $m$ points.

\begin{proposition}[Convergence rate under a fixed partition]
\label{prop:rwsd-rate}
Consider a fixed population partition
$\mathcal{C}^1, \ldots, \mathcal{C}^K$ with $K \ge 2$, 
$\pi_k := \mathbf{P}(\mathcal{C}^k) > 0$ and
$\mathbf{P}^{(k)} := \mathbf{P}(\,\cdot \mid \mathcal{C}^k)$, and any fixed distribution $\mu \in \mathcal{P}_2^{a.c.}(\mathbb{R}^d)$. Suppose that
$\mu$ and each $\mathbf{P}^{(k)}$ satisfy Assumption~6.5 of
\citet{bachoc2024wasserstein} %(compact, strongly convex supports with $C^{1,1}$ boundaries and densities bounded above and away from zero),
and that
$M := \max_{k} \mathbb{E}_{P' \sim \mathbf{P}^{(k)}}
[\mathcal{W}_2^{-1}(\mu, P')] < \infty$.
Let $\mathcal{C}^k_{n,m}$ be the $k$-th cluster of empirical measures $\hat\mu_{i,m}$ sampled from $\mu_i \in \mathcal{C}^k$. Let $\widehat{\mathrm{RWSD}}_K(\mu)$ and
$\mathrm{RWSD}_K(\mu)$ be the relative depths of $\mu$ computed with $\{\mathcal{C}^k_{n,m}\}_{k=1}^K$ and with
$\{\mathcal{C}^k\}_{k=1}^K$ as in \eqref{eq:relative-depth}. 
%with the cluster of $\mu$ fixed and $\mathrm{WSD}(\mu; \emptyset) := 0$, and 
Denote $\Delta_K(\mu) = \big|\widehat{\mathrm{RWSD}}_K(\mu) - \mathrm{RWSD}_K(\mu)\big|$. Then
\[
\mathbb{E}\,|\Delta_K(\mu)| \le 4KM\sqrt{\alpha(d,m)}
+ \frac{C}{\sqrt{n}} \sum_{k=1}^{K} \pi_k^{-1/2},
\]
where $C$ is an absolute constant, and $\alpha(d,m)$ has the form
\begin{equation*}\label{Rates-OT}
  \alpha(d,m):=C' \cdot \begin{cases}
   \frac{1}{m} & {\rm if }
    \ d=1,\\
  \frac{\log(m)}{m} & {\rm if }
    \ d=2,\\
   m^{-\frac{2}{d}} & {\rm if }
    \ d>2.
  \end{cases} 
\end{equation*}
%squared-$L^2$ rate of optimal-transport map estimation \citep[Eq.~(6.8)]{bachoc2024wasserstein}, equals $C'm^{-1}$, $C'(\log m)/m$, or $C'm^{-2/d}$ for $d = 1$, $d = 2$, or $d > 2$,
%with $C'$ depending only on $d$, the supports, and the constants of Assumption~6.5.
\end{proposition}

Finally, we characterize the robustness of RWSD via bounding its influence function.  
Since each object contributes to RWSD only through unit-norm transport directions, RWSD has bounded influence function. 
Consider any fixed $K$-partition $\mathcal{C}^1, \ldots, \mathcal{C}^K$ ($K\geq 2$) of $n$ probability measures $\nu_1, \ldots, \nu_n$, without empty clusters. 
Let $\hat\pi_{\min}$ be the smallest fraction across $K$ clusters, and define $\overline{R}(K) := n^{-1}\sum_{i=1}^{n}\mathrm{RWSD}_K(\nu_i)$. 

\begin{proposition}[Bounded influence]\label{prop:influence}
Replace any $s$ elements among $\nu_1, \ldots, \nu_n$ by $s$ arbitrary probability measures belonging to $\mathcal{P}_2(\mathbb{R}^d)$ without altering the partition, and let $\overline{R}{}'$ be the new RWSD after replacement. Then
$|\overline{R}{}'(K) - \overline{R}(K)| \le 6s/(n\hat\pi_{\min})$ for
any $K \in \mathcal{K}$. Consequently, if the maximizer
$\widehat{K} \coloneqq \arg\max_{K\in\mathcal K} \overline{R}(K) $ 
exceeds every other $K\in \mathcal K$ by a margin $\delta > 0$ and $s < \delta\,\hat\pi_{\min}\,n/12$, then
$\widehat{K}$ also maximizes $\overline{R}{}'$.
\end{proposition}

By contrast, however large $n$ is, a single replacement can lower the average silhouette score by at least $1/2$ and can raise DBI by an arbitrarily large amount (see Lemma~\ref{lem:contrast} in Appendix~\ref{app:contrast}). Therefore, RWSD has better robustness property than silhouette score and DBI.

\section{Numerical Simulations}\label{sec:simulations}
Given the sample set of empirical distributions, we implement Wasserstein
$K$-center clustering with three selection criteria, the silhouette score,
the DBI, and RWSD, and record their selected optimal $K$ values. We choose
silhouette score and DBI for comparison because these two criteria are the
most widely used and computationally efficient; additional discussion on
other criteria is provided in Appendix~\ref{app:baseline-choice}. In all
the simulations, we use a two-stage sampling scheme: step 1: draw
distributions $\mu_1^k, \mu_2^k, \ldots, \mu_{n_k}^k$ from the $k$-th
Cluster; step 2: draw $m$ data points
$\bx_1^{i,k}, \bx_2^{i,k}, \ldots, \bx_m^{i,k}$ from each
$\mu_i^k,\ i=1, \cdots, n_k, \ k=1, \ldots, K$. Throughout this section, we
choose $n_k = 200,\ m=50$ for all $k\in \mathcal K$ with
$\mathcal K = \{ 2, \ldots, 10\}$, and each reported quantity in the tables
below represents the average of $20$ independent replications with $\pm$
one standard deviation. Although we run simulations for every
$K\in \mathcal K$, but due to space limitation,
Tables~\ref{tab:loc-imbalance}--\ref{tab:triangle-noise} only retain the
necessary information, and omit results for some $K$.

\paragraph{Clusters with unbalanced separations}
Consider three clusters with unbalanced separations,
\[
\text{Cluster } k = \Big\{ N(\boldsymbol{c}_k, \sigma^2 \bI_{2\times 2}),\ \boldsymbol{c}_k \sim \text{Unif}[\mathcal{B}_k] \Big\},
\]
where $k = 1, 2, 3$, $\mathcal{B}_1 = [-1.8,-1.7] \times [0,0.5]$,
$\mathcal{B}_2 = [-0.3,-0.2] \times [0,0.5]$,
$\mathcal{B}_3 = [20,21] \times [0,0.5]$, $\sigma = 0.12$. As shown in
Figure~\ref{fig:loc-imbalance}, Clusters 1 and 2 lie closer together while
Cluster 3 is isolated. Silhouette and DBI both select $K=2$, merging the two
nearby clusters because the distant Cluster 3 dominates all pairwise
distances. RWSD correctly identifies $K=3$, as depth-based centrality
remains sensitive to local structure.

\begin{figure}[!htb]
  \centering
  \begin{minipage}[c]{0.50\textwidth}
    \centering
    \includegraphics[width=\linewidth, height =0.5\linewidth]{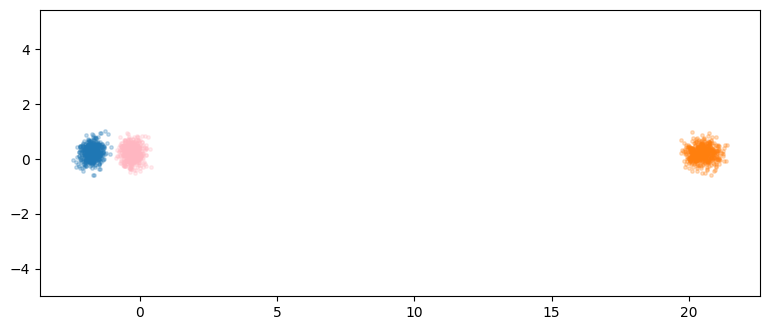}
    \caption{\small Clusters with unbalanced separations.}
    \label{fig:loc-imbalance}
  \end{minipage}\hfill
  \begin{minipage}[c]{0.46\textwidth}
    \centering\small
    \captionof{table}{ Criterion profiles in the unbalanced-separation setting (mean $\pm$ sd over $20$ replications); boldface marks each criterion's optimum. The true $K^\star=3$.}
    \label{tab:loc-imbalance}
    \setlength{\tabcolsep}{4pt}
    \begin{tabular}{lccc}
      \toprule
      $K$ & Silhouette & RWSD & DBI \\
      \midrule
      2 & \textbf{.967 ($\pm$.001)} & .398 ($\pm$.001) & \textbf{.049 ($\pm$.001)} \\
      3 & .913 ($\pm$.002) & \textbf{.402 ($\pm$.001)} & .127 ($\pm$.004) \\
      4 & .764 ($\pm$.029) & .375 ($\pm$.003) & .438 ($\pm$.065) \\
      5 & .673 ($\pm$.040) & .355 ($\pm$.006) & .617 ($\pm$.043) \\
      6 & .589 ($\pm$.029) & .334 ($\pm$.005) & .709 ($\pm$.038) \\
      \midrule
      $\widehat K$ & 2 & 3 & 2 \\
      \bottomrule
    \end{tabular}
  \end{minipage}
\end{figure}
\FloatBarrier

\paragraph{Clusters with imbalanced dispersion}
Consider two highly heterogeneous clusters of distributions,
\begin{align*}
    &\text{Cluster 1} = \Big\{ N(\boldsymbol{c}_1, \bI_{2\times 2}), \ \boldsymbol{c}_1 \sim \text{Unif}\big[ \calB_{100}(\boldsymbol{0}) \big] \Big\}, \\
    &\text{Cluster 2} = \Big\{ N(\boldsymbol{c}_2, \bI_{2 \times 2}), \ \boldsymbol{c}_2 \sim \text{Unif}\big[ \calB_{0.1}(\boldsymbol{a}) \big] \Big\},
\end{align*}
where $\calB_{0.1}(\boldsymbol{a})$ represents a ball centered at
$\boldsymbol{a}$ with radius $0.1$, and $\boldsymbol{a}=(120, 0)^\top$. As
shown in Figure~\ref{fig:unbalanced_dispersion}, the distributions in
Cluster 1 (blue) are substantially more dispersed than those in Cluster 2
(orange), and Cluster 2 is almost part of Cluster 1, which makes it very
difficult to separate two clusters. The optimal number of clusters under
different selection criteria, including RWSD (Algorithm~\ref{alg:selectK}),
silhouette score, and DBI are summarized in
Table~\ref{tab:sparse-vs-compact}. Distance-based selection criteria tend
to pick a larger $K$ because they rely on distance-based dispersions. By
contrast, RWSD measures each distribution's relative centrality within its
assigned cluster, and remain stable under extremely heterogeneous clusters.

\begin{figure}[!htb]
  \centering
  \begin{minipage}[c]{0.38\textwidth}
    \centering
    \includegraphics[width=\linewidth]{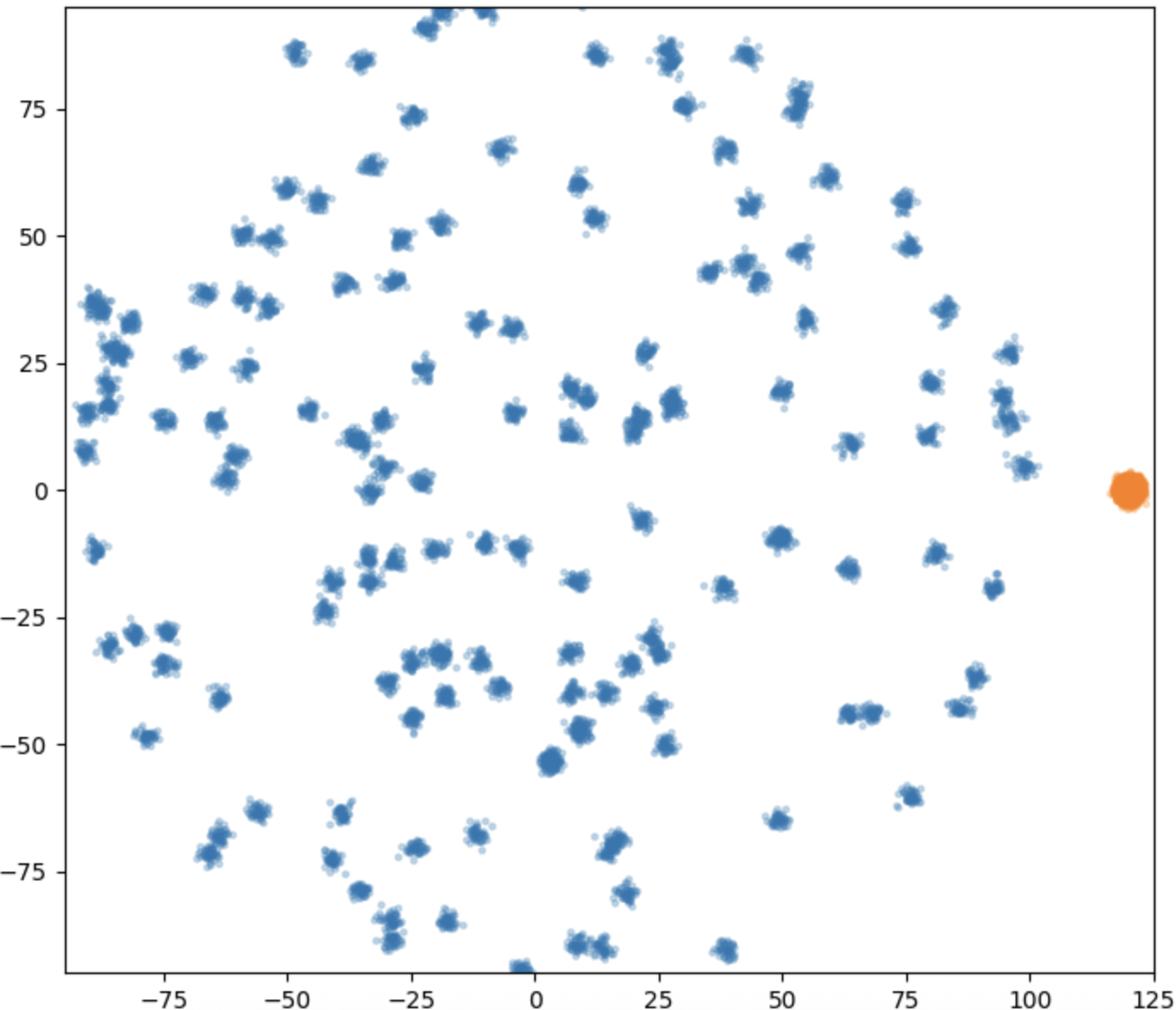}
    \caption{\small Two clusters of unbalanced dispersions.}
    \label{fig:unbalanced_dispersion}
  \end{minipage}\hfill
  \begin{minipage}[c]{0.56\textwidth}
    \centering\small
    \captionof{table}{\small Criterion profiles in the unbalanced-dispersion setting (mean $\pm$ sd over $20$ replications); boldface marks each criterion's optimum. The true $K^\star=2$.}
    \label{tab:sparse-vs-compact}
    \setlength{\tabcolsep}{4pt}
    \begin{tabular}{lccc}
      \toprule
      $K$ & Silhouette & RWSD & DBI \\
      \midrule
      2 & .656 ($\pm$.012) & \textbf{.358 ($\pm$.005)} & \textbf{.602 ($\pm$.031)} \\
      3 & .654 ($\pm$.010) & .354 ($\pm$.014) & .814 ($\pm$.049) \\
      4 & .669 ($\pm$.010) & .351 ($\pm$.017) & .727 ($\pm$.033) \\
      5 & .664 ($\pm$.010) & .342 ($\pm$.018) & .761 ($\pm$.065) \\
      6 & .667 ($\pm$.009) & .343 ($\pm$.015) & .768 ($\pm$.033) \\
      7 & .671 ($\pm$.008) & .337 ($\pm$.013) & .757 ($\pm$.038) \\
      8 & .672 ($\pm$.010) & .337 ($\pm$.010) & .767 ($\pm$.053) \\
      9 & .674 ($\pm$.010) & .336 ($\pm$.010) & .763 ($\pm$.059) \\
      10 & \textbf{.676 ($\pm$.011)} & .336 ($\pm$.012) & .762 ($\pm$.043) \\
      \midrule
      $\widehat K$ & 10 & 2 & 2 \\
      \bottomrule
    \end{tabular}
  \end{minipage}
\end{figure}
\FloatBarrier

\paragraph{Clusters with heterogeneous scales}
Consider three clusters, indexed by $k = 1, 2, 3$, of centered Gaussian
distributions with very different variances,
\[
\text{Cluster } k = \Big\{ N(\boldsymbol{0}, \sigma_k^2 \bI_{2\times 2}),\
\sigma_k \sim \mathrm{Tri}[a_k, b_k, c_k] \Big\},
\]
where $\mathrm{Tri}[a, b, c]$ denotes the triangular distribution on
$[a, c]$ with mode $b$, and $(a_1, b_1, c_1) = (0.0, 0.2, 0.4)$,
$(a_2, b_2, c_2) = (1.0, 1.2, 1.4)$, $(a_3, b_3, c_3) = (5.0, 5.2, 5.4)$.
The data clouds from the three clusters are marked as dark green, green,
and light green, respectively, in Figure~\ref{fig:scale-imbalance}.
Silhouette selects $K=2$ because the smallest-variance group forms a
compact core dominating dispersion comparisons. DBI improves slightly but
still struggles under heterogeneous within-cluster spread. RWSD recovers
$K=3$, as it reflects how deep each distribution is within its cluster
compared to other clusters.

\begin{figure}[!htb]
  \centering
  \begin{minipage}[c]{0.38\textwidth}
    \centering
    \includegraphics[width=\linewidth]{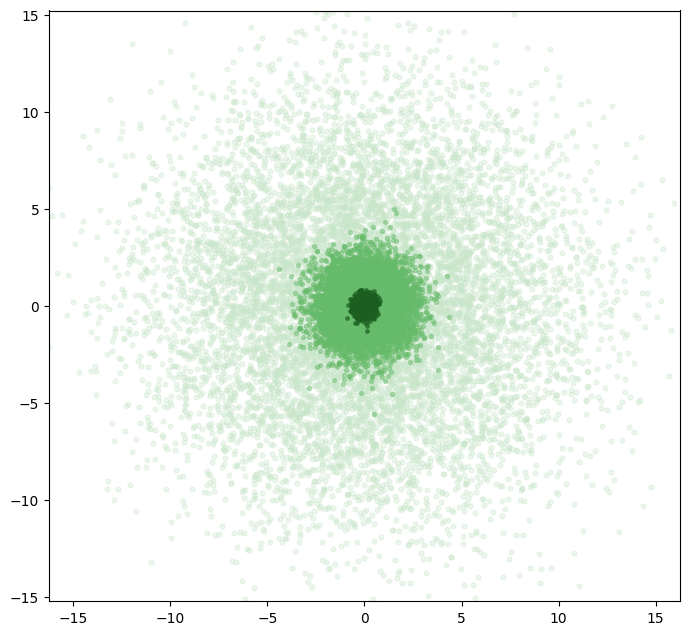}
    \caption{\small Clusters with heterogeneous scales.}
    \label{fig:scale-imbalance}
  \end{minipage}\hfill
  \begin{minipage}[c]{0.56\textwidth}
    \centering\small
    \captionof{table}{Criterion profiles in the heterogeneous-scale setting  (mean $\pm$ sd over $20$ replications); boldface marks each  criterion's optimum. The true $K^\star=3$.}
    \label{tab:scale-imbalance}
    \setlength{\tabcolsep}{4pt}
    \begin{tabular}{lccc}
      \toprule
      $K$ & Silhouette & RWSD & DBI \\
      \midrule
      2 & \textbf{.660 ($\pm$.003)} & .360 ($\pm$.002) & .569 ($\pm$.005) \\
      3 & .511 ($\pm$.003) & \textbf{.365 ($\pm$.002)} & \textbf{.553 ($\pm$.004)} \\
      4 & .310 ($\pm$.028) & .186 ($\pm$.045) & 1.732 ($\pm$.114) \\
      5 & .219 ($\pm$.077) & .113 ($\pm$.048) & 2.153 ($\pm$.219) \\
      6 & .173 ($\pm$.066) & .082 ($\pm$.042) & 2.345 ($\pm$.183) \\
      7 & .171 ($\pm$.065) & .079 ($\pm$.042) & 2.408 ($\pm$.148) \\
      8 & .160 ($\pm$.057) & .071 ($\pm$.037) & 2.421 ($\pm$.138) \\
      9 & .154 ($\pm$.045) & .070 ($\pm$.035) & 2.381 ($\pm$.152) \\
      10 & .152 ($\pm$.045) & .068 ($\pm$.034) & 2.393 ($\pm$.134) \\
      \midrule
      $\widehat K$ & 2 & 3 & 3 \\
      \bottomrule
    \end{tabular}
  \end{minipage}
\end{figure}
\FloatBarrier

\paragraph{Clusters with heavy-tailed dispersion}
Consider three clusters of distributions on $\mathbb{R}^2$, with
$\sigma=0.2$ and $k=1,2,3$,
\begin{equation*}
    \text{Cluster } k = \Big\{ N(\boldsymbol{c}_k, \sigma^2 \bI),\ \boldsymbol{c}_k = \boldsymbol{a}_k + (r\cos\theta, r\sin\theta) \Big\}
\end{equation*}
with $\boldsymbol{a}_1 = (-1, 0), \boldsymbol{a}_2 = (1, 0),
\boldsymbol{a}_3 = (0, \sqrt{3})$,
$r \sim \text{Pareto}(r_{\min} = 0.35, \alpha = 0.45)$ truncated at
$r_{\max} = 5.0$, and $\theta \sim \text{Unif}[0, 2\pi)$. Each cluster has
a core and heavy-tailed peripheral data clouds
(Figure~\ref{fig:three_families_heavy_tail}). Silhouette and DBI both split
one cluster into smaller sub-clusters because heavy-tailed peripheral
clouds inflate within-cluster dispersions. RWSD picks the correct $K$,
thanks to its robustness to heavy tails (Table~\ref{tab:pareto-validation}).

\begin{figure}[!htb]
  \centering
  \begin{minipage}[c]{0.38\textwidth}
    \centering
    \includegraphics[width=\linewidth]{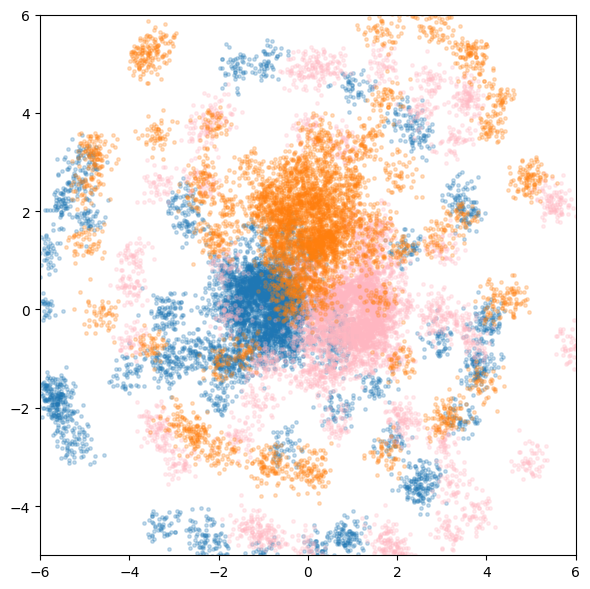}
    \caption{\small Clusters with heavy-tailed dispersion.}
    \label{fig:three_families_heavy_tail}
  \end{minipage}\hfill
  \begin{minipage}[c]{0.56\textwidth}
    \centering\small
    \captionof{table}{Criterion profiles in the heavy-tailed setting (mean $\pm$ sd over $20$ replications); boldface marks each criterion's optimum. The true $K^\star=3$.}
    \label{tab:pareto-validation}
    \setlength{\tabcolsep}{4pt}
    \begin{tabular}{lccc}
      \toprule
      $K$ & Silhouette & RWSD & DBI \\
      \midrule
      2  & .287 ($\pm$.011) & .278 ($\pm$.004) & 1.667 ($\pm$.037) \\
      3  & .322 ($\pm$.005) & \textbf{.307 ($\pm$.002)} & 1.282 ($\pm$.029) \\
      4  & .296 ($\pm$.021) & .289 ($\pm$.011) & 1.178 ($\pm$.030) \\
      5  & .306 ($\pm$.016) & .284 ($\pm$.006) & 1.037 ($\pm$.039) \\
      6  & .331 ($\pm$.005) & .284 ($\pm$.005) & .941 ($\pm$.027) \\
      7  & .369 ($\pm$.006) & .289 ($\pm$.006) & .839 ($\pm$.007) \\
      8  & \textbf{.376 ($\pm$.035)} & .293 ($\pm$.009) & \textbf{.822 ($\pm$.095)} \\
      9  & .373 ($\pm$.030) & .295 ($\pm$.007) & .837 ($\pm$.058) \\
      10 & .367 ($\pm$.025) & .293 ($\pm$.009) & .850 ($\pm$.033) \\
      \midrule
      $\widehat K$ & 8 & 3 & 8 \\
      \bottomrule
    \end{tabular}
  \end{minipage}
\end{figure}
\FloatBarrier

\paragraph{Clusters contaminated with outliers}
Consider three clusters
\[
\text{Cluster } k = \Big\{ N(\boldsymbol{c}_k, 0.08^2 \bI_{2\times 2}),\ \boldsymbol{c}_k \sim N(\boldsymbol{a}_k, 0.6^2 \bI_{2\times 2}) \Big\}
\]
for $k=1, 2, 3$ with $\boldsymbol{a}_1 = (-0.8, 0),
\boldsymbol{a}_2 = (0.8, 0), \boldsymbol{a}_3 = (0, 0.8\sqrt{3})$. Every
cluster is contaminated with $40\%$ outliers, where each outlier
distribution is \textit{iid} drawn from
\begin{equation*}
    \Big\{ N(\boldsymbol{c}, 0.08^2 \bI_{2\times 2}),\ \boldsymbol{c} \sim \text{Unif}[\mathcal{R}] \Big\},
\end{equation*}
where $\mathcal{R} = [-1.6,\ 1.6] \times [-0.8,\ 0.8 + 0.8\sqrt{3}]$ is a
rectangle. Shown in Figure~\ref{fig:sim-three-disks} and
Table~\ref{tab:triangle-noise}, distance-based criteria are largely
affected by the outliers and split the true components into smaller
artificial sub-clusters, while depth-based RWSD is robust to overlapping
and outliers.

\begin{figure}[!htb]
  \centering
  \begin{minipage}[c]{0.38\textwidth}
    \centering
    \includegraphics[width=0.9\linewidth]{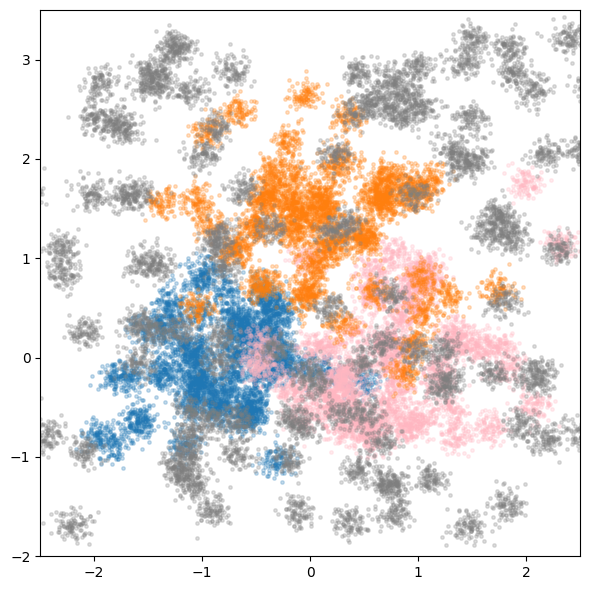}
    \caption{\small Clusters (colored) contaminated with outliers (gray).}
    \label{fig:sim-three-disks}
  \end{minipage}\hfill
  \begin{minipage}[c]{0.56\textwidth}
    \centering\small
    \captionof{table}{Criterion profiles in the outlier-contaminated setting (mean $\pm$ sd over $20$ replications); boldface marks each criterion's optimum. The true $K^\star=3$.}
    \label{tab:triangle-noise}
    \setlength{\tabcolsep}{4pt}
    \begin{tabular}{lccc}
      \toprule
      $K$ & Silhouette & RWSD & DBI \\
      \midrule
      2  & .318 ($\pm$.018) & .284 ($\pm$.007) & 1.36 ($\pm$.082) \\
      3  & \textbf{.357 ($\pm$.018)} & \textbf{.308 ($\pm$.009)} & 1.03 ($\pm$.048) \\
      4  & .328 ($\pm$.030) & .297 ($\pm$.006) & 1.06 ($\pm$.085) \\
      5  & .331 ($\pm$.023) & .295 ($\pm$.010) & .997 ($\pm$.086) \\
      6  & .328 ($\pm$.022) & .292 ($\pm$.008) & .962 ($\pm$.058) \\
      7  & .328 ($\pm$.020) & .290 ($\pm$.008) & .942 ($\pm$.059) \\
      8  & .326 ($\pm$.025) & .289 ($\pm$.009) & .914 ($\pm$.064) \\
      9  & .328 ($\pm$.018) & .288 ($\pm$.006) & .903 ($\pm$.056) \\
      10 & .330 ($\pm$.020) & .287 ($\pm$.007) & \textbf{.893 ($\pm$.053)} \\
      \midrule
      $\widehat K$ & 3 & 3 & 10 \\
      \bottomrule
    \end{tabular}
  \end{minipage}
\end{figure}
\FloatBarrier

Among the above settings, the first three witness similar performances by
Wasserstein $K$-means and $K$-medians, while Wasserstein $K$-medians beats
$K$-means in the last two cases. Across all five settings, RWSD selects the
true $K$, whereas the silhouette score and DBI fail in at least three of
the settings.

Two further groups of experiments are shown in
Appendix~\ref{app:regular-sims} to complement the above results. The first
group (Appendix~\ref{app:regular}) gives some regular settings where the
clusters are balanced, well separated, and hence, all three criteria
recover the true $K$. The second group lists a few settings where RWSD
fails. In these settings, the clusters are not geodesically convex, and
RWSD cannot efficiently distinguish within- and between-cluster geometry.
The success of RWSD depends on WSD's ability to correctly detect a
center-outward ordering. But WSD fails to do so in the case of ``hollow
clusters'' (Appendix~\ref{app:failure}).

\section{Real Data Analysis}\label{sec:real-data}
We apply RWSD to two real datasets, the MNIST handwritten digit dataset and a flow cytometry dataset from an acute myeloid leukemia cohort, to showcase its practical utility. Extra details are deferred to Appendix~\ref{app:experiment-details}.

\subsection{MNIST digits}
%We first apply our framework to the MNIST handwritten digit dataset.
Each grayscale MNIST handwritten digit image of size $28\times28$ is treated as a discrete distribution over pixel coordinates with intensities normalized to sum to one, $\mu_i \in \mathcal{P}_2([0,1]^2)$. 
We evaluate RWSD under different ground-truth numbers of clusters. The first setting (true $K = 3$) includes digits labeled as ``1'', ``2'', and ``9'', with 30 images for each digit to form a total of $N = 90$ normalized empirical distributions. The second setting (true $K = 5$) includes digits ``1'', ``2'', ``5'', ``8'', and ``9'', with 30 images for each digit to form $N = 150$ empirical distributions.
%The Wasserstein $K$-median clustering algorithm was applied using barycenter solver with entropic regularization.
For each $K \in \{2, \ldots, 9\}$, we computed the silhouette score, RWSD, and DBI to pick the optimal $K$.

\begin{table}[t]
\centering
\caption{\small MNIST digits: criterion profiles under the two ground-truth
  settings; boldface marks each criterion's optimum.}
  \vspace{3mm}
\begin{tabular}{l ccc ccc}
\toprule
& \multicolumn{3}{c}{$K^\star = 3$} & \multicolumn{3}{c}{$K^\star = 5$} \\
\cmidrule(lr){2-4} \cmidrule(lr){5-7}
$K$ & Silhouette & RWSD & DBI & Silhouette & RWSD & DBI \\
\midrule
2 & 0.2429 & 0.1453 & 1.8272 & 0.1990 & 0.1395 & 1.8927 \\
3 & 0.2794 & \textbf{0.1730} & 1.5550 & 0.2050 & 0.1344 & 1.7009 \\
4 & 0.2754 & 0.1712 & \textbf{1.4019} & 0.2097 & 0.1397 & \textbf{1.6903} \\
5 & 0.2541 & 0.1533 & 1.4831 & \textbf{0.2164} & \textbf{0.1398} & 1.7157 \\
6 & \textbf{0.2872} & 0.1576 & 1.4330 & 0.1904 & 0.1204 & 1.9217 \\
7 & 0.2601 & 0.1450 & 1.4536 & 0.1867 & 0.1176 & 1.9033 \\
8 & 0.2447 & 0.1428 & 1.4703 & 0.1723 & 0.1173 & 1.7966 \\
9 & 0.2176 & 0.1247 & 1.6573 & 0.1855 & 0.1073 & 1.9403 \\
\midrule
\textbf{$\widehat K$} & 6 & 3 & 4 & 5 & 5 & 4 \\
\bottomrule
\end{tabular}
\end{table}

\subsection{Acute Myeloid Leukemia (AML) flow cytometry cohort}\label{subsec:aml-cytof}
We apply RWSD to the BLAST110 flow cytometry dataset \citep{mocking2024computational}, comprising bone marrow samples from $20$ healthy individuals and $30$ acute myeloid leukemia (AML) patients. Each individual is represented by a distribution on $\R^8$ of eight protein markers, estimated from $m=1000$ cells. The cluster of AML patients are expected to be separable from  healthy individuals based on their abnormal distributions of protein markers on hematopoietic cells in bone marrow.   
As shown in Table~\ref{tab:aml-validation}, both RWSD and silhouette score select the two well-separated groups, consistent with previous findings \citep{mocking2024computational}. Preprocessing details are described in Appendix~\ref{app:aml-preprocessing}. 

%We also apply RWSD to a flow cytometry dataset, BLAST110 \citep{mocking2024computational}, which collects bone marrow samples from $20$ healthy individuals and $30$ acute myeloid leukemia (AML) patients. Flow cytometry quantifies fluorescence-labeled protein markers at the single-cell level and is a central tool for characterizing hematopoietic sub-populations and leukemic phenotypes \citep{saeys2016computational}. The goal is to identify AML patients based on their abnormal distributions of protein markers on hematopoietic cells in bone marrow.    The healthy population and AML patients are treated as two clusters, the flow cytometry profile of $d=8$ biomarkers of each individual is a distribution (on $\R^8$) from either cluster, and the $m=1000$ cells (profiles) from each individual are $m=1000$ data points drawn from the corresponding distribution.  Details of cytometry preprocessing are deferred to Appendix~\ref{app:aml-preprocessing}.
%As shown in Table~\ref{tab:aml-validation}, both silhouette score and RWSD work well on this dataset as the two populations (healthy \textit{vs.} AML patients) have quite separated flow cytometry profiles, consistent with the existing reports \citep{mocking2024computational}. 

\begin{table}[H]
\centering
\caption{RWSD and silhouette score pick the correct $K$.}
\label{tab:aml-validation}
\vspace{3mm}
\begin{tabular}{lccc}
\toprule
$K$ & Silhouette & RWSD & DBI \\
\midrule
2 & \textbf{0.2156} & \textbf{0.1604} & 1.2391 \\
3 & 0.2074 & 0.1435 & 1.2483 \\
4 & 0.1450 & 0.0812 & 1.1107 \\
5 & 0.1373 & 0.0604 & 0.9324 \\
6 & 0.1549 & 0.0586 & \textbf{0.8751} \\
\midrule
\textbf{$\widehat K$} & 2 & 2 & 6 \\
\bottomrule
\end{tabular}
\end{table}
\section{Discussion and Limitations}
We propose a new criterion, \textit{Relative Wasserstein Spatial Depth} (RWSD), to select the optimal number of clusters in Wasserstein $K$-center clustering. %It is built upon a recently developed notion of depth for probability distributions, Wasserstein spatial depth.  
RWSD favors partitions that are both internally cohesive and externally well separated in terms of centrality and Wasserstein geometry. 
Demonstrated by extensive simulations and real data analysis, RWSD is more robust than distance-based criteria under outliers, heavy dispersions, overlapped clusters, and unbalanced clusters. This advantage is largely attributed to its awareness of the nonlinear  Wasserstein geometry. 

On the theoretical side, we established a Pollard-type strong consistency result for Wasserstein $K$-centers, which implies the consistency of RWSD.
On the computational aspect, we devise a Wasserstein $K$-medians clustering algorithm, which is robust to ``irregular'' clusters. Computational efficiency is largely improved by restricting the support size of Wasserstein barycenters, entropic regularized OT, and iteratively reweighted least squares.

\paragraph{Limitations} Our theoretical results are established in  scenarios where the clusters are geodesically convex and the WSD correctly induces a center--outward ordering. In such benign cases, RWSD effectively distinguishes within- and between-cluster geometry.  
It remains open to explore scenarios where the clusters are not geodesically convex, and global depth may mask valleys in the Wasserstein landscape leading to cluster merging. We emphasize using the proposed method alongside domain expertise, validation, and careful interpretation.
A potential remedy would be \emph{local} Wasserstein spatial depth, where the centrality of a distribution is evaluated within its small neighborhood.

\bibliographystyle{plainnat}
\bibliography{WSD_Clustering.bib}

\clearpage
\appendix

\section{Wasserstein K-center clustering}\label{app:methods}
We describe two Wasserstein $K$-center clustering methods used in this work, namely Wasserstein $K$-means and Wasserstein $K$-medians. These two methods differ in choices of the centers of the clusters: the former uses Wasserstein barycenters, whereas the latter uses Wasserstein Fréchet medians. 

\subsection{Wasserstein K-means clustering}\label{app:kmeans}
Given a collection of probability measures $\{\mu_i\}_{i=1}^n \subset \mathcal{P}_2(\mathbb{R}^d)$, Wasserstein K-means seeks a $K$-group partition by minimizing the within-cluster variation measured in the $2$-Wasserstein metric:
\begin{equation}
\min_{\{\mathcal C_k,\nu_k\}_{k=1}^K}
\sum_{k=1}^K \sum_{i\in\mathcal C_k} \mathcal{W}_2^2(\mu_i,\nu_k),
\label{eq:wkm-objective}
\end{equation}
where $\mathcal C_1,\ldots,\mathcal C_K$ form a partition of $\{1,\ldots,n\}$, and each $\nu_k \in \mathcal{P}_2(\mathbb{R}^d)$ is the Wasserstein barycenter of the $k$th cluster.
%In contrast to Euclidean $K$-means, the cluster representatives are no longer arithmetic means, but Wasserstein barycenters.

Starting from an initial set of centroids
$\nu_1^{(0)},\ldots,\nu_K^{(0)}$,
the algorithm alternates between an assignment step and a barycenter update step. At iteration $t$, each distribution $\mu_i$ is assigned to its nearest current centroid:
\begin{equation}
\mathcal C_k^{(t)} = \left\{ i\in\{1,\ldots,n\}: \mathcal{W}_2(\mu_i,\nu_k^{(t)}) \le \mathcal{W}_2(\mu_i,\nu_j^{(t)}), \ \forall j=1,\ldots,K
\right\},
\label{eq:wkm-assignment}
\end{equation}
and the centroid of each cluster is then updated as the corresponding Wasserstein barycenter:
\begin{equation}
\nu_k^{(t+1)} = \arg\min_{\nu\in\mathcal P_2(\mathbb R^d)}
\frac{1}{|\mathcal C_k^{(t)}|}
\sum_{i\in\mathcal C_k^{(t)}} \mathcal{W}_2^2(\mu_i,\nu),
\quad k=1,\ldots,K.
\label{eq:wkm-barycenter-update}
\end{equation}
%In our setting, each subproblem in \eqref{eq:wkm-barycenter-update} is computed using the semi-free support Wasserstein barycenter procedure introduced in the previous subsection. 
The full procedure is summarized in Algorithm~\ref{alg:wasserstein-kmeans-barycenter}.

\begin{algorithm}[h]
\caption{Barycenter-based Wasserstein $K$-means clustering (\citep{zhuang2022wasserstein})}
\label{alg:wasserstein-kmeans-barycenter}
\begin{algorithmic}[1]
\State \textbf{Input:} a collection of probability $\{\mu_1,\dots,\mu_n\}\subset\mathcal P_2(\R^d)$, number of clusters $K$, tolerance $\varepsilon$
\State \textbf{Output:} cluster labels $\{l_i\}_{i=1}^n$ and Wasserstein centroids $\{\nu_k\}_{k=1}^K$
\State Initialize $\nu_1^{(0)},\ldots,\nu_K^{(0)} \in \mathcal P_2(\mathbb R^d)$
\State $t \gets 0$
\Repeat
    \For{$i=1,\ldots,n$}
        \State $l_i^{(t)} \gets \arg\min_{1\le k\le K} \mathcal{W}_2(\mu_i,\nu_k^{(t)})$
    \EndFor
    \For{$k=1,\ldots,K$}
        \State $\mathcal C_k^{(t)} \gets \{i:l_i^{(t)}=k\}$
        \State $\nu_k^{(t+1)} \gets \arg\min_{\nu\in\mathcal P_2(\mathbb R^d)}
        \frac{1}{|\mathcal C_k^{(t)}|}\sum_{i\in\mathcal C_k^{(t)}} \mathcal{W}_2^2(\mu_i,\nu)$
    \EndFor
    %\State Compute $\{\nu_k^{(t+1)}\}_{k=1}^K$ by the semi-free support barycenter algorithm in the previous subsection
    \State $\Delta^{(t)} \gets \max_{1\le k\le K} \mathcal{W}_2(\nu_k^{(t+1)},\nu_k^{(t)})$
    \State $t \gets t+1$
\Until{$\Delta^{(t-1)}\le \varepsilon$}
\State \textbf{return} $\{l_i^{(t-1)}\}_{i=1}^n$ and $\{\nu_k^{(t)}\}_{k=1}^K$
\end{algorithmic}
\end{algorithm}
%Algorithm~\ref{alg:wasserstein-kmeans-barycenter} is the Wasserstein analogue of the classical Lloyd iteration, with Euclidean centroids replaced by Wasserstein barycenters. Although the objective in \eqref{eq:wkm-objective} is nonconvex, the alternating assignment and update structure is simple to implement and naturally compatible with the semi-free support barycenter solver described above.
The algorithm can be viewed as an alternating minimization scheme in which the cluster labels are updated by nearest-barycenter assignment, and the barycenters are obtained by solving regularized semi-free support OT problems (detailed below). 

%\subsection{Compute Wasserstein Barycenter with Semi-Free Support}\label{sec:Wasserstein_barycenter}
The Wasserstein barycenter $\nu^\ast$ of multiple distributions $\{ \mu_i \}_{i=1}^n \in\mathcal{P}_2(\mathbb{R}^d)$ is defined as
\begin{equation}\label{Wasserstein_median}
\nu^\ast=\arg\min_{\nu\in\mathcal{P}_2(\mathbb{R}^d)} \frac1n\sum_{i=1}^n \mathcal{W}_2^2(\nu,\mu_i).
\end{equation}
In practice, the $\mu_i$'s are empirical distributions and written as $\mu_i=\sum_{j=1}^{m_i} b_{ij}\,\delta_{y_{ij}}$. 
When the barycenter is allowed to have completely free support, the  problem \eqref{Wasserstein_median} admits a linear programming (LP) formulation over a polytope that couples the
barycenter with each input distribution. Expressing the barycenter as $\nu=\sum_{p=1}^{K} a_p\,\delta_{x_p}$ with unknown supports
$X=\{x_p\}_{p=1}^{K}$ and weights $\boldsymbol{a}\in\Delta^{K}$, one seeks transport plans
$T^{(i)}\in\mathbb{R}_+^{K\times m_i}$ satisfying $T^{(i)}\mathbf{1}=\boldsymbol{a}$ and $(T^{(i)})^\top\mathbf{1}=\boldsymbol{b}_i$
and minimize the total cost $\sum_i \langle T^{(i)}, C(X,Y_i)\rangle$.
A key structural fact is the sparsity of optimal barycenters: there exists an optimal
discrete barycenter whose support size is bounded by
\[
K^\ast \;\le\; \sum_{i=1}^n m_i \;-\; n \;+\; 1,
\]
see the LP characterization and sparsity theorem of \emph{Anderes, Borgwardt, and Miller} (2016, Theorems~1–2)\citep{anderes2016discrete}.
Consequently, even though an optimal barycenter is sparse, the ambient LP may be very large because
the candidate support set can grow combinatorially with $\{m_i\}$, which leads to substantial computational cost.

When computing Wasserstein barycenters, we adopt the approach of \citep{cuturi2014fast}, which approximates the barycenter by a discrete probability measure whose support size is moderate. The algorithm alternates between (i) updating the weights of each support point and
(ii) updating the support of the barycenter.
Within each iteration, entropy-regularized OT is used to improve efficiency, which is a strictly convex problem. 
For discrete measures
$\mu=\sum_p a_p\delta_{x_p}$ and $\nu=\sum_q b_q\delta_{y_q}$ with cost
$M_{pq}=\|x_p-y_q\|^2$, entropy regularized OT can be formulated as
\begin{equation*}
%p_\varepsilon(a,b;M)=
\min_{T\in \Pi(\mu,\nu)} \langle T,M\rangle - \varepsilon H(T),
\qquad
H(T)=-\sum_{p,q} T_{pq}\log T_{pq},
\end{equation*}
whose dual involves potentials $(f,g)$ and can be computed efficiently by Sinkhorn matrix scaling.

\subsection{Wasserstein $K$-medians clustering}\label{app:kmedian}
Given a set of distributions $\{\mu_i\}_{i=1}^n \subset \mathcal{P}_2(\R^d)$,
Wasserstein $K$-medians clustering seeks a partition
$\{\calC_k\}_{k=1}^K$ and cluster representatives that minimize the
total within-cluster Wasserstein distance:
\begin{equation}\label{eq:W2Kmedian-objective}
\min_{\{\calC_k, \nu_k\}_{k=1}^K}
\sum_{k=1}^{K} \sum_{i\in \calC_k} \mathcal{W}_2(\mu_i, \nu_k).
\end{equation}
For fixed clusters, each optimal $\nu_k$ is a Wasserstein Fréchet
median of $\calC_k$, the median counterpart of the barycenter used by
$K$-means; like barycenters \citep{le2017existence}, Wasserstein
medians exist and are consistent \citep{you2025wasserstein}. The
problem is nonconvex, already because of the discrete assignment of
distributions to clusters, so we approximate its solution with
Algorithm~\ref{alg:Wkmed-pam-warm}.

\begin{algorithm}[h]
  \caption{Wasserstein $K$-median Clustering with PAM Warm Start}\label{alg:Wkmed-pam-warm}
  \begin{algorithmic}[1]
    \State \textbf{Input:}
      a collection of probability $\{\mu_1,\dots,\mu_n\}\subset\mathcal P_2(\R^d)$;
      number of clusters $K$;
      tolerance $\varepsilon$
    \State \textbf{Output:} cluster labels $\{l_i\}_{i=1}^n$ and Wasserstein medians $\{m_k\}_{k=1}^K$
    \State \textbf{PAM warm start:}
      \Statex \hspace{1.2em}Compute pairwise distances $D_{ij}=\mathcal{W}_2(\mu_i,\mu_j)$
      \Statex \hspace{1.2em}Run PAM on $D$ to obtain initial medoids $\mu_{n_1},\dots,\mu_{n_K}$ and initial labels $\{l_i^{(0)}\}_{i=1}^n$
      \Statex \hspace{1.2em}Set $m_k^{(0)} \gets \mu_{n_k}$ for $k=1,\dots,K$
    \State $t \gets 0$
    \Repeat
      \For{$i=1,\ldots,n$}
        \State $l_i^{(t)} \gets \arg\min_{1\le k\le K} \mathcal{W}_2(\mu_i,m_k^{(t)})$
      \EndFor
      \For{$k=1,\ldots,K$}
        \State $\mathcal C_k^{(t)} \gets \{i:l_i^{(t)}=k\}$
        \State $m_k^{(t+1)} \gets \arg\min_{\nu\in\mathcal{P}_2(X)} \sum_{i\in\mathcal C_k^{(t)}} \mathcal{W}_2(\mu_i,\nu)$
      \EndFor
      \State $\Delta^{(t)} \gets \max_{1\le k\le K} \mathcal{W}_2(m_k^{(t+1)},m_k^{(t)})$
      \State $t \gets t+1$
    \Until{$\Delta^{(t-1)} \le \varepsilon$}
    \State \textbf{return} $\{l_i^{(t-1)}\}_{i=1}^n$ and $\{m_k^{(t)}\}_{k=1}^K$
  \end{algorithmic}
\end{algorithm}
\medskip

We elaborate two key steps in the above algorithm: initialization via PAM and the computation of Fréchet median in Wasserstein space. 

%\noindent\textbf{PAM-based warm start}
To enhance stability and accelerate convergence, we use Partitioning Around Medoids (PAM) \citep{kaufman1987finding} as a warm start.  
Formally, given empirical measures $\{\mu_i\}_{i=1}^n \subset \mathcal{P}_2(X)$,
we first compute the pairwise Wasserstein distance matrix
\begin{equation}
D_{ij} = \mathcal{W}_2(\mu_i,\mu_j), \quad i,j=1,\ldots,n.
\label{eq:pam-warm-matrix}
\end{equation}
The PAM heuristic is then applied to identify $K$ representative medoids 
that minimize the total within-cluster dispersion:
\begin{equation}
\min_{\{\calC_k,m_k\}}
\sum_{k=1}^{K}\sum_{i\in \calC_k} D_{i,m_k}.
\label{eq:pam-warm-objective}
\end{equation}
The resulting medoids $\{m_k^{(0)}\}_{k=1}^K$ and initial cluster assignments 
$\{l_i^{(0)}\}_{i=1}^n$ provide a good initialization
for the subsequent Wasserstein $K$-median optimization.

%\noindent\textbf{Fréchet Median in Wasserstein Space} 
The Fréchet median of a set of probability measures 
$\{\mu_i\}_{i=1}^N \subset \mathcal{P}_2(\mathbb{R}^d)$ 
is defined as the probability measure that minimizes the sum of its Wasserstein distances to each element in that set:
\begin{equation}
\mu^{\ast} = \arg\min_{\nu \in \mathcal{P}_2(\mathbb{R}^d)}
\sum_{i=1}^{N} \mathcal{W}_2(\nu, \mu_i).
\end{equation}
%where $\mathcal{P}_2(\mathbb{R}^d)$ denotes the space of probability measures with finite second moment.

To compute $\mu^{\ast}$, we adopt an \emph{iteratively reweighted least squares} (IRLS) scheme \citep{you2025wasserstein}
defined on the Wasserstein space, which reformulates the median estimation problem as a sequence of weighted Wasserstein barycenter computations.
At each iteration $t$, the algorithm updates the weights according to the current distances between the candidate median $\nu^{(t)}$ and the input measures:
\begin{equation}
w_n^{(t)} = \frac{1}{\mathcal{W}_2(\nu^{(t)},\mu_n)}, \qquad
\tilde w_n^{(t)} = 
\frac{w_n^{(t)}}{\sum_{k=1}^N w_k^{(t)}},
\label{eq:irls-weight-update}
\end{equation}
Given these normalized weights $\tilde w_n^{(t)}$,
the next iterate $\nu^{(t+1)}$ is obtained by solving the corresponding weighted barycenter problem:
\begin{equation}
\nu^{(t+1)} =
\arg\min_{\nu \in \mathcal{P}_2(\mathbb{R}^d)}
\sum_{n=1}^{N} \tilde w_n^{(t)}\, \mathcal{W}_2^2(\nu, \mu_n).
\label{eq:weighted-barycenter}
\end{equation}
Such updates replace the non smooth median objective with a locally
quadratic surrogate, yielding a smooth majorization–minimization step 
in Wasserstein space.

In practice, each weighted barycenter subproblem
\eqref{eq:weighted-barycenter} can be efficiently solved using 
the semi-free support barycenter method introduced in the previous subsection.
By fixing a moderate number of support points and jointly optimizing their locations and weights under entropic regularization, we obtain a numerically stable and computationally tractable approximation of the Wasserstein median. The resulting estimator can thus be viewed as a semi-free support Wasserstein median, whose computation is summarized in Algorithm~\ref{alg:Wmed-IRLS} \citep{you2025wasserstein}.

\begin{algorithm}[h]
\caption{Wasserstein median computation by IRLS.}
\label{alg:Wmed-IRLS}
\begin{algorithmic}[1]
\State \textbf{Input:} a collection of probability $\{\mu_1,\dots,\mu_n\}\subset\mathcal P_2(\R^d)$
\State \textbf{Output:} $\displaystyle \nu^\ast=\arg\min_{\nu\in\mathcal{P}_2(\mathbb{R}^d)}\sum_{n=1}^N \mathcal{W}_2(\nu,\mu_n)$.
\Statex
\State \textbf{Initialize}
$\nu^{(0)}\in\mathcal{P}_2(\mathbb{R}^d)$.
\State \Repeat
      \State $w_n^{(t)} = 1 / \mathcal{W}_2(\nu^{(t)},\mu_n)$
      \State $\tilde{w}_n^{(t)} = w_n^{(t)} / \sum_{k=1}^N w_k^{(t)}$
      \State $\displaystyle
             \nu^{(t+1)} = \arg\min_{\nu\in\mathcal{P}_2(\mathbb{R}^d)}
             \sum_{n=1}^N \tilde{w}_n^{(t)}\,\mathcal{W}_2^2(\nu,\mu_n)$
      \If{$\nu^{(t+1)}$ equals one of the $\mu_n$'s \textbf{or} $\nu^{(t+1)}$ is close to $\nu^{(t)}$}
         \State \textbf{return} $\nu^{(t+1)}$ and stop
      \EndIf
\Until{convergence}
\end{algorithmic}
\end{algorithm}

% \medskip
% \noindent\textbf{Computing Wasserstein Fréchet median via IRLS }
% \begin{equation}
% m_k^{(t+1)} = \arg\min_{\nu\in\mathcal{P}_2(X)}
% \sum_{i=1}^{n} p_{i,k}^{(t)} \mathcal{W}_2(\nu,\mu_i),
% \label{eq:cluster-median-update}
% \end{equation}
% where $p_{i,k}^{(t)}=\mathbb{P}[l_i^{(t)}=k]$.
% Each median update in \eqref{eq:cluster-median-update}
% is computed using the semi-free support IRLS scheme 
% introduced in the previous section \cite{you2025wasserstein}.

\subsection{Implementation of empirical WSD}\label{app:wsd-impl}
For each pair $(\hat\mu_{q,m},\, \hat\mu_{i,m})$ of empirical measures with uniform weights $1/m$, 
we solve the discrete OT problem exactly via linear programming using the 
\texttt{ot.emd} routine in the Python Optimal Transport library \citep{flamary2021pot}, 
with squared Euclidean ground cost. This yields an OT plan 
$\pi_{q,i} \in \mathbb{R}^{m \times m}_{+}$ and the squared Wasserstein distance 
$\mathcal{W}_2^2(\hat\mu_{q,m},\, \hat\mu_{i,m}) = \sum_{j,k} \pi_{q,i}(j,k)\, \|X_{q,j} - X_{i,k}\|^2$. 
We then form the barycentric projection
\[
T_{\hat\mu_{q,m},\, \hat\mu_{i,m}}(X_{q,j}) := m \sum_{k=1}^{m} X_{i,k}\, \pi_{q,i}(j,k),
\]
which serves as the discrete analogue of the OT map in Definition~\ref{def:WSD}. 
When $\mathcal{W}_2(\hat\mu_{q,m},\, \hat\mu_{i,m}) < \tau$ for a small tolerance $\tau$, 
we set the corresponding direction vector to zero, following the convention in Definition~\ref{def:WSD}.
\section{Main proofs}\label{app:proof}

\subsection{Proof of Theorem~\ref{thm:W2-kmeans}}
\paragraph{Notation and preliminaries.}
Throughout this appendix, for a metric space $(Y,\rho)$, we write
$B_\rho(y,r):=\{y'\in Y:\rho(y',y)<r\}$ for the open ball of radius $r$
around $y$. The \emph{Hausdorff distance} between two nonempty subsets
$A,B\subseteq Y$ is
\[
    H(A,B):=\max\Bigl\{\sup_{a\in A}\inf_{b\in B}\rho(a,b),\;
                       \sup_{b\in B}\inf_{a\in A}\rho(a,b)\Bigr\}.
\]
The distance $H(A,B)$ is small when every point of $A$ is close to some
point of $B$ and every point of $B$ is close to some point of $A$; we use
$H$ to compare finite sets of centers. The \emph{support}
$\mathrm{supp}(\mathbf{P})$ of a Borel probability measure $\mathbf{P}$ on
$Y$ is the set of points $y\in Y$ with $\mathbf{P}(B_\rho(y,r))>0$ for
every $r>0$. When $Y$ is Polish, $\mathrm{supp}(\mathbf{P})$ is nonempty
and $\mathbf{P}(\mathrm{supp}(\mathbf{P}))=1$
\citep{billingsley1999convergence}.

For a finite set $A\subset\mathcal{P}_2(X)$ and
$\mathbf{P}\in\mathcal{P}(\mathcal{P}_2(X))$, recall the population
$K$-means objective $\Phi$ and the optimal value $m_K$
\begin{equation*}
\begin{aligned}
    &\Phi(A,\mathbf{P})=\int g_A(\mu)\,\mathbf{P}(d\mu),\qquad
    g_A(\mu):=\min_{\nu\in A}\mathcal{W}_2^2(\mu,\nu),\\
    &m_K(\mathbf{P})=\inf\{\Phi(A,\mathbf{P}):|A|\le K\}.
\end{aligned}
\end{equation*}
Let $\mathbf{P}_n:=\tfrac{1}{n}\sum_{i=1}^{n}\delta_{\mu_i}$ be the
empirical measure of the distributions $\mu_1,\dots,\mu_n$ drawn
independently from $\mathbf{P}$. We compare center sets in the space
\[
    \mathcal{A}_K:=\{A\subset\mathcal{P}_2(X):1\le|A|\le K\},
\]
equipped with the Hausdorff distance $H$ induced by $\mathcal{W}_2$, and we
write $D:=\mathrm{diam}(\mathcal{P}_2(X))$. To lighten notation, we write
$\bar A(K):=A^\star_K$ for the population minimizer and
$A_n(K):=\widehat{A}_{n,K}$ for the empirical minimizer in
Theorem~\ref{thm:W2-kmeans}. Strong consistency then means
$A_n(K)\to\bar A(K)$ almost surely in $H$.

\begin{remark}[Compactness of $(\mathcal{P}_2(X),\mathcal{W}_2)$]\label{rem:p2-compact}
The proof uses the compactness of $(\mathcal{P}_2(X),\mathcal{W}_2)$.
Because $\mathcal{P}_2(X)$ is an infinite-dimensional space of measures,
the compactness is not obvious, and we derive the compactness from three
classical facts, (a)--(c). Part~(d) then gives the compactness of
$(\mathcal{A}_K,H)$.

\smallskip\noindent\emph{(a) Weak compactness of $\mathcal{P}(X)$.}
Since $X$ is compact and $\mu(X)=1$ for every $\mu\in\mathcal{P}(X)$, the
family $\mathcal{P}(X)$ is uniformly tight. Prokhorov's theorem
\citep{prokhorov1956convergence,billingsley1999convergence} then shows
that $\mathcal{P}(X)$ with the weak topology is compact and Polish.

\smallskip\noindent\emph{(b) $\mathcal{P}_2(X)=\mathcal{P}(X)$.}
Fix $x_0\in X$. Every $\mu\in\mathcal{P}(X)$ satisfies
\[
    \int_X d(x,x_0)^2\,\mu(dx)\;\le\;\mathrm{diam}(X)^2\;<\;\infty,
\]
so every $\mu\in\mathcal{P}(X)$ has a finite second moment, and
$\mathcal{P}_2(X)=\mathcal{P}(X)$.

\smallskip\noindent\emph{(c) $\mathcal{W}_2$ metrizes weak convergence.}
By \citet[Theorem~6.9]{villani2009optimal} and
\citet[Proposition~7.1.5]{ambrosio2008gradient}, on a compact base space,
$\mathcal{W}_2(\mu_n,\mu)\to0$ if and only if $\mu_n\rightharpoonup\mu$.
Hence $\mathcal{W}_2$ induces the weak topology of (a).

\smallskip\noindent
By (a)--(c), $(\mathcal{P}_2(X),\mathcal{W}_2)$ is a compact Polish space.
The diameter $D$ of $(\mathcal{P}_2(X),\mathcal{W}_2)$ is finite, because
for all $\mu,\nu\in\mathcal{P}_2(X)$,
\[
    \mathcal{W}_2^2(\mu,\nu)
    \;=\;\inf_{\pi\in\Pi(\mu,\nu)}\int d(x,y)^2\,\pi(dx,dy)
    \;\le\;\mathrm{diam}(X)^2 .
\]

\smallskip\noindent\emph{(d) Hausdorff compactness of $\mathcal{A}_K$.}
By the Blaschke selection theorem
\citep[Theorem~7.3.8]{burago2001course}, for every compact metric space
$(Y,\rho)$, the set
\[
    \mathrm{Cl}(Y):=\{F\subseteq Y:\,F\text{ nonempty and closed}\}
\]
is compact under the Hausdorff distance induced by $\rho$. We apply the
theorem with $Y=\mathcal{P}_2(X)$. Finite sets are closed, so
$\mathcal{A}_K=\{F\in\mathrm{Cl}(\mathcal{P}_2(X)):|F|\le K\}$. The set
$\mathcal{A}_K$ is closed in $\mathrm{Cl}(\mathcal{P}_2(X))$: if
$F_n\to F$ in $H$ and $F$ contained $K+1$ points
$\nu_1,\dots,\nu_{K+1}$ with pairwise distances at least $3\varepsilon>0$,
then for large $n$ every $\nu_l$ would have a point of $F_n$ within
distance $\varepsilon$, and the $K+1$ nearby points of $F_n$ would be
distinct, so $|F_n|>K$. Hence $(\mathcal{A}_K,H)$ is compact.

\smallskip\noindent
The proof uses compactness in two ways. The compactness of
$(\mathcal{P}_2(X),\mathcal{W}_2)$ gives the bound
\[
    0\;\le\;g_A(\mu)\;\le\;D^2
    \qquad\text{for all }A\in\mathcal{A}_K\text{ and }\mu\in\mathcal{P}_2(X),
\]
so every average over $\mathbf{P}_n$ below is an average of bounded
i.i.d.\ random variables, to which the strong law of large numbers
(SLLN) applies. The compactness of $(\mathcal{A}_K,H)$ gives the finite
nets in Lemma~\ref{lem:uslln} and the convergent subsequences in Step~3.
\end{remark}

The proof adapts the strategy of \citet{pollard1981strong} from Euclidean
space to Wasserstein space and has three steps. Step~1 fixes an anchor
$\mu_0$ in the support of $\mathbf{P}$ and shows that at least one center
of $A_n(K)$ stays within a fixed distance of $\mu_0$. Step~2 extends the
bound to all centers of $A_n(K)$ by level separation,
assumption~(iii): if one center were far from $\mu_0$, removing the far
center would barely raise the empirical cost, although $K-1$ centers
cost strictly more than $K$ centers in the population. Step~3 uses the
compactness of $(\mathcal{A}_K,H)$ and the uniform convergence of the
empirical objective to show that every limit point of $A_n(K)$ equals
the unique population minimizer $\bar A(K)$. We first prove two lemmas
used in all three steps.

\begin{lemma}[Lipschitz continuity of $\Phi$]\label{lem:lip-phi}
For any $A,B\in\mathcal{A}_K$ and any probability measure $\mathbf{Q}$ on
$\mathcal{P}_2(X)$,
\[
    |\Phi(A,\mathbf{Q})-\Phi(B,\mathbf{Q})|\;\le\;\bigl(2D+H(A,B)\bigr)\,H(A,B).
\]
In particular, $A\mapsto\Phi(A,\mathbf{Q})$ is continuous on
$(\mathcal{A}_K,H)$.
\end{lemma}

The bound does not depend on $\mathbf{Q}$, so we can apply
Lemma~\ref{lem:lip-phi} both with $\mathbf{Q}=\mathbf{P}$ and with
$\mathbf{Q}=\mathbf{P}_n$.

\begin{proof}
Let $\delta:=H(A,B)$. By the definition of $H$, and since $B$ is finite,
every $a\in A$ has some $b(a)\in B$ with
$\mathcal{W}_2(a,b(a))\le\delta$. Fix $\mu\in\mathcal{P}_2(X)$, and let
$a_\mu\in A$ be a center nearest to $\mu$, so that
$g_A(\mu)=\mathcal{W}_2^2(\mu,a_\mu)$. The triangle inequality gives
\[
    \mathcal{W}_2(\mu,b(a_\mu))
    \;\le\; \mathcal{W}_2(\mu,a_\mu)+\mathcal{W}_2(a_\mu,b(a_\mu))
    \;\le\; \sqrt{g_A(\mu)}+\delta.
\]
Squaring both sides and using $\sqrt{g_A(\mu)}\le D$ gives
\[
    g_B(\mu)\;\le\;\mathcal{W}_2^2(\mu,b(a_\mu))
    \;\le\; g_A(\mu)+2\delta\sqrt{g_A(\mu)}+\delta^2
    \;\le\; g_A(\mu)+(2D+\delta)\delta.
\]
Exchanging the roles of $A$ and $B$ gives
$g_A(\mu)\le g_B(\mu)+(2D+\delta)\delta$, so
$|g_A(\mu)-g_B(\mu)|\le(2D+\delta)\delta$ for every $\mu$. Integrating
against $\mathbf{Q}$ proves the lemma.
\end{proof}

\begin{lemma}[Uniform SLLN over $\mathcal{A}_K$]\label{lem:uslln}
\[
    \sup_{A\in\mathcal{A}_K}\bigl|\Phi(A,\mathbf{P}_n)-\Phi(A,\mathbf{P})\bigr|
    \longrightarrow 0 \qquad \text{a.s.}
\]
In particular, $m_j(\mathbf{P}_n)\to m_j(\mathbf{P})$ a.s.\ for every
$j\le K$.
\end{lemma}

For a fixed center set $A$, $\Phi(A,\mathbf{P}_n)$ is an average of $n$
i.i.d.\ random variables bounded by $D^2$, so
$\Phi(A,\mathbf{P}_n)\to\Phi(A,\mathbf{P})$ a.s.\ by the SLLN. The
empirical minimizer $A_n(K)$ is random, however, so the proof of
Theorem~\ref{thm:W2-kmeans} needs the convergence to hold uniformly over
$\mathcal{A}_K$. Lemma~\ref{lem:lip-phi} and a finite net of
$(\mathcal{A}_K,H)$ reduce the supremum over $\mathcal{A}_K$ to a maximum
over finitely many center sets, to which the SLLN applies.

\begin{proof}
Fix $\varepsilon>0$ and choose $\delta\in(0,1]$ with
$(2D+1)\delta<\varepsilon/3$. Since $(\mathcal{A}_K,H)$ is compact
(Remark~\ref{rem:p2-compact}), there is a finite set
$\mathcal{A}_{K,\delta}\subset\mathcal{A}_K$ such that every
$A\in\mathcal{A}_K$ has some $A^\delta\in\mathcal{A}_{K,\delta}$ with
$H(A,A^\delta)\le\delta$. Lemma~\ref{lem:lip-phi} then gives, for every
$A\in\mathcal{A}_K$ and every probability measure $\mathbf{Q}$ on
$\mathcal{P}_2(X)$,
\begin{equation}\label{eq:lemma2-net}
    \bigl|\Phi(A,\mathbf{Q})-\Phi(A^\delta,\mathbf{Q})\bigr|
    \;\le\;(2D+\delta)\delta
    \;\le\;(2D+1)\delta
    \;<\;\varepsilon/3.
\end{equation}
The set $\mathcal{A}_{K,\delta}$ is finite and every $g_{A^\delta}$ is
bounded by $D^2$, so the SLLN \citep[Theorem~22.1]{billingsley1995probability},
applied to each $A^\delta\in\mathcal{A}_{K,\delta}$, gives
\[
    \max_{A^\delta\in\mathcal{A}_{K,\delta}}
        \bigl|\Phi(A^\delta,\mathbf{P}_n)-\Phi(A^\delta,\mathbf{P})\bigr|
    \;\xrightarrow{a.s.}\;0.
\]
For every $A\in\mathcal{A}_K$, the triangle inequality, as in the
empirical-process argument of \citet{ganssler1979empirical}, gives
\begin{equation*}
    \begin{aligned}
    |\Phi(A,\mathbf{P}_n)-\Phi(A,\mathbf{P})|
    \;\le\; &|\Phi(A,\mathbf{P}_n)-\Phi(A^\delta,\mathbf{P}_n)|
        + |\Phi(A^\delta,\mathbf{P}_n)-\Phi(A^\delta,\mathbf{P})|\\
        &+ |\Phi(A^\delta,\mathbf{P})-\Phi(A,\mathbf{P})|.
        \end{aligned}
\end{equation*}
By \eqref{eq:lemma2-net} with $\mathbf{Q}=\mathbf{P}_n$ and with
$\mathbf{Q}=\mathbf{P}$, the first and third terms are each below
$\varepsilon/3$. The middle term is at most the maximum in the previous
display, which tends to $0$ a.s. Taking the supremum over
$A\in\mathcal{A}_K$ and then the $\limsup$ over $n$ gives
\[
    \limsup_{n\to\infty}\,\sup_{A\in\mathcal{A}_K}
    |\Phi(A,\mathbf{P}_n)-\Phi(A,\mathbf{P})|\;\le\;2\varepsilon/3\;<\;\varepsilon
    \qquad\text{a.s.}
\]
Applying the bound with $\varepsilon=1/k$ for $k=1,2,\dots$ and
intersecting the countably many probability-one events proves the
uniform convergence. Finally, $\mathcal{A}_j\subseteq\mathcal{A}_K$ for
$j\le K$, so
\[
    |m_j(\mathbf{P}_n)-m_j(\mathbf{P})|
    \;\le\;\sup_{A\in\mathcal{A}_j}|\Phi(A,\mathbf{P}_n)-\Phi(A,\mathbf{P})|
    \;\le\;\sup_{A\in\mathcal{A}_K}|\Phi(A,\mathbf{P}_n)-\Phi(A,\mathbf{P})|,
\]
and the right-hand side tends to $0$ a.s.
\end{proof}

We now prove Theorem~\ref{thm:W2-kmeans}.

\paragraph{Step 1: At least one empirical center lies in a fixed ball.}
Step~1 fixes an anchor $\mu_0$ in the support of $\mathbf{P}$ and shows
that $A_n(K)$ keeps a center near $\mu_0$. If every center of $A_n(K)$
were far from $\mu_0$, the distributions near $\mu_0$ would be far from
every center, and the empirical cost $\Phi(A_n(K),\mathbf{P}_n)$ would
exceed the cost of the single center $\{\mu_0\}$, which contradicts the
optimality of $A_n(K)$.

Since $\mathcal{P}_2(X)$ is Polish (Remark~\ref{rem:p2-compact}),
$\mathrm{supp}(\mathbf{P})$ is nonempty; choose
$\mu_0\in\mathrm{supp}(\mathbf{P})$. By the definition of the support,
\begin{equation}\label{eq:positive-mass}
    p_r:=\mathbf{P}\bigl(B_{\mathcal{W}_2}(\mu_0,r)\bigr)>0
    \qquad\text{for every }r>0.
\end{equation}
Fix any $r>0$; the value of $r$ plays no role. Let
\[
    c_0:=\Phi(\{\mu_0\},\mathbf{P})=\int \mathcal{W}_2^2(\mu,\mu_0)\,\mathbf{P}(d\mu)\;\le\;D^2.
\]
Fix a radius $M>r$ and suppose that, for infinitely many $n$,
\begin{equation}\label{eq:contradiction-step1}
    A_n(K)\cap B_{\mathcal{W}_2}(\mu_0,M)=\varnothing.
\end{equation}
For every $n$ that satisfies \eqref{eq:contradiction-step1}, every
center $\nu\in A_n(K)$ satisfies $\mathcal{W}_2(\nu,\mu_0)\ge M$. For
every $\mu\in B_{\mathcal{W}_2}(\mu_0,r)$, the triangle inequality gives
\[
    \mathcal{W}_2(\mu,\nu)\;\ge\;\mathcal{W}_2(\nu,\mu_0)-\mathcal{W}_2(\mu,\mu_0)\;\ge\;M-r,
\]
so $g_{A_n(K)}(\mu)\ge(M-r)^2$. Restricting the integral that defines
$\Phi$ to $B_{\mathcal{W}_2}(\mu_0,r)$ gives
\begin{equation}\label{eq:lower-bound-step1}
    \begin{aligned}
    \Phi(A_n(K),\mathbf{P}_n)
    \;&=\; \int g_{A_n(K)}(\mu)\,\mathbf{P}_n(d\mu)\\
    &\ge\; \int_{B_{\mathcal{W}_2}(\mu_0,r)} (M-r)^2\,\mathbf{P}_n(d\mu)
    \;=\; (M-r)^2\,\mathbf{P}_n\bigl(B_{\mathcal{W}_2}(\mu_0,r)\bigr).
    \end{aligned}
\end{equation}
The indicator of $B_{\mathcal{W}_2}(\mu_0,r)$ is bounded, so the SLLN
gives $\mathbf{P}_n(B_{\mathcal{W}_2}(\mu_0,r))\to p_r$ a.s. Taking the
$\liminf$ of \eqref{eq:lower-bound-step1} along the $n$ that satisfy
\eqref{eq:contradiction-step1} gives
\begin{equation}\label{eq:liminf-step1}
    \liminf_{n}\Phi(A_n(K),\mathbf{P}_n)\;\ge\;(M-r)^2 p_r.
\end{equation}
For an upper bound, the single center $\{\mu_0\}$ is admissible because
$1\le K$, so the optimality of $A_n(K)$ gives
\[
    \Phi(A_n(K),\mathbf{P}_n)
    \;\le\;\Phi(\{\mu_0\},\mathbf{P}_n)
    \;=\;\int \mathcal{W}_2^2(\mu,\mu_0)\,\mathbf{P}_n(d\mu).
\]
The integrand is bounded by $D^2$, so the SLLN gives
$\Phi(\{\mu_0\},\mathbf{P}_n)\to c_0$ a.s., and
\begin{equation}\label{eq:limsup-step1}
    \limsup_{n\to\infty}\Phi(A_n(K),\mathbf{P}_n)\;\le\;c_0.
\end{equation}
Together, \eqref{eq:liminf-step1} and \eqref{eq:limsup-step1} give
$(M-r)^2p_r\le c_0$. The constant $c_0$ does not depend on $M$, while
$(M-r)^2p_r\to\infty$ as $M\to\infty$. Choose $M_1>r$ with
$(M_1-r)^2p_r>c_0$. Then, almost surely, \eqref{eq:contradiction-step1}
with $M=M_1$ holds for only finitely many $n$, so
\[
    A_n(K)\cap B_{\mathcal{W}_2}(\mu_0,M_1)\neq\varnothing
    \qquad\text{for all large } n .
\]
The same conclusion holds for every radius larger than $M_1$.

\paragraph{Step 2: All empirical centers lie in a larger ball.}
Step~1 bounds the distance from $\mu_0$ to one center of $A_n(K)$.
Step~2 bounds the distance from $\mu_0$ to every center of $A_n(K)$. If
$K=1$, Step~1 already bounds the only center, so let $K\ge2$. Step~2
uses level separation, assumption~(iii):
\[
    \Delta_K:=m_{K-1}(\mathbf{P})-m_K(\mathbf{P})>0.
\]
If a center of $A_n(K)$ were far from $\mu_0$, the far center would not
be the nearest center of any distribution near $\mu_0$, so removing the
far center would barely raise the empirical cost. The remaining $K-1$
centers would then nearly attain $m_K(\mathbf{P})$, which contradicts
$\Delta_K>0$.

For $M>0$, let
\[
    \tau(M):=\tfrac{9}{4}\int_{\{\mathcal{W}_2(\mu,\mu_0)>2M\}}
    \mathcal{W}_2^2(\mu,\mu_0)\,\mathbf{P}(d\mu).
\]
The integrand is bounded by $D^2$ and tends to $0$ pointwise as
$M\to\infty$, so dominated convergence gives $\tau(M)\to0$. Enlarging
$M_1$ from Step~1 if necessary, we assume $\tau(M_1)<\Delta_K/2$.
Suppose, to reach a contradiction, that along a subsequence (still
indexed by $n$) some center $\nu_n\in A_n(K)$ satisfies
\begin{equation}\label{eq:nu-far}
    \mathcal{W}_2(\nu_n,\mu_0)>5M_1.
\end{equation}
By Step~1, for all large $n$, $A_n(K)$ also contains a center $\eta_n$
with $\mathcal{W}_2(\eta_n,\mu_0)<M_1$. Remove the far center:
\[
    A_n':=A_n(K)\setminus\{\nu_n\}.
\]
The set $A_n'$ contains $\eta_n$ and has at most $K-1$ centers.
Removing a center can only raise the cost, so
$g_{A_n'}(\mu)-g_{A_n(K)}(\mu)\ge0$ for every $\mu$. We bound the
increase separately for $\mu$ near $\mu_0$ and for $\mu$ far from
$\mu_0$.

The factor $5$ in \eqref{eq:nu-far} makes $\eta_n$ closer than $\nu_n$
to every $\mu$ near $\mu_0$. For every $\mu$ with
$\mathcal{W}_2(\mu,\mu_0)\le2M_1$, the triangle inequality gives
\[
    \mathcal{W}_2(\mu,\nu_n)\;\ge\;\mathcal{W}_2(\nu_n,\mu_0)-\mathcal{W}_2(\mu,\mu_0)\;>\;5M_1-2M_1=3M_1,
\]
\[
    \mathcal{W}_2(\mu,\eta_n)\;\le\;\mathcal{W}_2(\mu,\mu_0)+\mathcal{W}_2(\mu_0,\eta_n)\;\le\;2M_1+M_1=3M_1.
\]
Hence $\nu_n$ is not a nearest center of $\mu$, and removing $\nu_n$
leaves the cost of $\mu$ unchanged:
\begin{equation}\label{eq:inner-zero}
    g_{A_n'}(\mu)\;=\;g_{A_n(K)}(\mu)
    \qquad\text{for every }\mu\text{ with }\mathcal{W}_2(\mu,\mu_0)\le 2M_1.
\end{equation}
For $\mu$ with $\mathcal{W}_2(\mu,\mu_0)>2M_1$, the center
$\eta_n\in A_n'$ gives
\[
    g_{A_n'}(\mu)-g_{A_n(K)}(\mu)
    \;\le\;\mathcal{W}_2^2(\mu,\eta_n)
    \;\le\;\bigl(\mathcal{W}_2(\mu,\mu_0)+M_1\bigr)^2
    \;<\;\tfrac{9}{4}\,\mathcal{W}_2^2(\mu,\mu_0),
\]
where the first inequality uses $g_{A_n(K)}(\mu)\ge0$ and
$\eta_n\in A_n'$, the second uses the triangle inequality, and the third
uses $M_1<\mathcal{W}_2(\mu,\mu_0)/2$. Together with
\eqref{eq:inner-zero},
\begin{equation}\label{eq:increment-bound}
    0\;\le\;g_{A_n'}(\mu)-g_{A_n(K)}(\mu)\;\le\;
    \begin{cases}
        0, & \mathcal{W}_2(\mu,\mu_0)\le 2M_1,\\[2pt]
        \tfrac{9}{4}\,\mathcal{W}_2^2(\mu,\mu_0), & \mathcal{W}_2(\mu,\mu_0)>2M_1.
    \end{cases}
\end{equation}

Integrating \eqref{eq:increment-bound} against $\mathbf{P}_n$ gives an
upper bound on the increase of the empirical cost:
\begin{equation}\label{eq:tail-bound}
    \Phi(A_n',\mathbf{P}_n)-\Phi(A_n(K),\mathbf{P}_n)
    \;\le\;\tfrac{9}{4}\!\int_{\{\mathcal{W}_2(\mu,\mu_0)>2M_1\}}\!\!
        \mathcal{W}_2^2(\mu,\mu_0)\,\mathbf{P}_n(d\mu).
\end{equation}
The integrand on the right is bounded by $D^2$, so the SLLN gives,
almost surely,
\begin{equation}\label{eq:tail-limit}
    \tfrac{9}{4}\!\int_{\{\mathcal{W}_2(\mu,\mu_0)>2M_1\}}\!\!
        \mathcal{W}_2^2(\mu,\mu_0)\,\mathbf{P}_n(d\mu)
    \;\xrightarrow{n\to\infty}\;\tau(M_1)\;<\;\Delta_K/2.
\end{equation}
For a lower bound, $A_n'$ has at most $K-1$ centers and $A_n(K)$ is an
empirical minimizer, so
\[
    \Phi(A_n',\mathbf{P}_n)\;\ge\;m_{K-1}(\mathbf{P}_n),
    \qquad
    \Phi(A_n(K),\mathbf{P}_n)\;=\;m_K(\mathbf{P}_n).
\]
By Lemma~\ref{lem:uslln}, $m_{K-1}(\mathbf{P}_n)\to m_{K-1}(\mathbf{P})$
and $m_K(\mathbf{P}_n)\to m_K(\mathbf{P})$ a.s., so
\begin{equation}\label{eq:lower-gap}
    \liminf_{n\to\infty}\bigl[\Phi(A_n',\mathbf{P}_n)-\Phi(A_n(K),\mathbf{P}_n)\bigr]
    \;\ge\; m_{K-1}(\mathbf{P})-m_K(\mathbf{P})
    \;=\;\Delta_K .
\end{equation}
Along the subsequence, \eqref{eq:tail-bound} and \eqref{eq:tail-limit}
bound the increase $\Phi(A_n',\mathbf{P}_n)-\Phi(A_n(K),\mathbf{P}_n)$
from above by a quantity that tends to $\tau(M_1)<\Delta_K/2$, while
\eqref{eq:lower-gap} bounds the same increase from below by a quantity
whose $\liminf$ is at least $\Delta_K$. The two bounds are incompatible,
so no subsequence satisfies \eqref{eq:nu-far}. Hence, almost surely, for
all large $n$, every center of $A_n(K)$ lies in the closed ball
\[
    \overline{B}:=\{\nu\in\mathcal{P}_2(X):\mathcal{W}_2(\nu,\mu_0)\le5M_1\}.
\]

\paragraph{Step 3: Identification of the limit.}
Step~3 shows that every limit point $\widetilde{A}$ of $A_n(K)$ equals
$\bar A(K)$. Lemma~\ref{lem:uslln} replaces $\mathbf{P}_n$ by
$\mathbf{P}$ in the cost of $A_n(K)$, Lemma~\ref{lem:lip-phi} passes from
$A_n(K)$ to $\widetilde{A}$, and together the two lemmas give
$\Phi(\widetilde{A},\mathbf{P})=m_K(\mathbf{P})$. Uniqueness,
assumption~(ii), then gives $\widetilde{A}=\bar A(K)$.

Fix an outcome in the probability-one event on which
Lemma~\ref{lem:uslln} and Steps~1--2 hold; the rest of the argument is
deterministic. The closed ball $\overline{B}$ is a closed subset of the
compact space $(\mathcal{P}_2(X),\mathcal{W}_2)$, so $\overline{B}$ is
compact, and by the Blaschke selection theorem, as in
Remark~\ref{rem:p2-compact}(d), the set
\[
    \mathcal{A}_K^{\overline{B}}:=\{A\subset\overline{B}:1\le|A|\le K\}
\]
is compact under $H$. By Step~2, $A_n(K)\in\mathcal{A}_K^{\overline{B}}$
for all large $n$. Hence every subsequence of $(A_n(K))$ has a further
subsequence $(A_{n_j}(K))$ with $A_{n_j}(K)\to\widetilde{A}$ in $H$ for
some $\widetilde{A}\in\mathcal{A}_K^{\overline{B}}$.

By Lemma~\ref{lem:lip-phi}, $\Phi(\cdot,\mathbf{P})$ is continuous in
$H$, so
\begin{equation}\label{eq:step3-cont}
    \Phi(A_{n_j}(K),\mathbf{P})\;\longrightarrow\;\Phi(\widetilde{A},\mathbf{P}).
\end{equation}
By Lemma~\ref{lem:uslln},
$\sup_{A\in\mathcal{A}_K}|\Phi(A,\mathbf{P}_{n_j})-\Phi(A,\mathbf{P})|\to0$;
taking $A=A_{n_j}(K)$ gives
\begin{equation}\label{eq:step3-uniform}
    \bigl|\Phi(A_{n_j}(K),\mathbf{P}_{n_j})-\Phi(A_{n_j}(K),\mathbf{P})\bigr|
    \;\longrightarrow\;0.
\end{equation}
Since $\Phi(A_{n_j}(K),\mathbf{P}_{n_j})=m_K(\mathbf{P}_{n_j})$ and
$m_K(\mathbf{P}_{n_j})\to m_K(\mathbf{P})$ by Lemma~\ref{lem:uslln},
\eqref{eq:step3-cont} and \eqref{eq:step3-uniform} give
\[
    \Phi(\widetilde{A},\mathbf{P})
    \;=\;\lim_{j}\Phi(A_{n_j}(K),\mathbf{P})
    \;=\;\lim_{j}\Phi(A_{n_j}(K),\mathbf{P}_{n_j})
    \;=\;\lim_{j} m_K(\mathbf{P}_{n_j})
    \;=\;m_K(\mathbf{P}).
\]
Hence $\widetilde{A}$ minimizes $\Phi(\cdot,\mathbf{P})$ over
$\mathcal{A}_K$, and the uniqueness of the population minimizer,
assumption~(ii), gives $\widetilde{A}=\bar A(K)$.

Every subsequence of $(A_n(K))$ therefore has a further subsequence that
converges to $\bar A(K)$. If $A_n(K)$ did not converge to $\bar A(K)$,
some $\varepsilon>0$ and some subsequence would satisfy
$H(A_{n_j}(K),\bar A(K))\ge\varepsilon$ for all $j$, and no further
subsequence of that subsequence could converge to $\bar A(K)$. Hence
$A_n(K)\to\bar A(K)$ in $H$. The outcome was an arbitrary element of a
probability-one event, so
\[
    A_n(K)\;\longrightarrow\;\bar A(K)
    \qquad\text{a.s.\ in the Hausdorff metric.}
\]
Lemma~\ref{lem:uslln} also gives
$\Phi(A_n(K),\mathbf{P}_n)=m_K(\mathbf{P}_n)\to m_K(\mathbf{P})$ a.s.
The argument uses assumptions~(ii) and~(iii) only at level $K$. Both
assumptions hold at every level $k\le K$, so the same argument with $k$
in place of $K$ proves the theorem for every $k=1,\dots,K$.\hfill$\square$
\subsection{Proof of Theorem~\ref{thm:W1-kmedians} and the uniqueness assumption}
\label{app:kmedians-proof}

\paragraph{On the uniqueness assumption.}
In our proof, we assume the population Wasserstein mean and Wasserstein median are unique. 
The uniqueness of the Wasserstein barycenter has been well-studied. For example, it is unique in the case of finitely many measures with at least one of them being absolutely continuous \citep{agueh2011barycenters}. The uniqueness of Wasserstein median remains largely open. We
are not aware of any general condition that guarantees the uniqueness of a Wasserstein median.
%The convergence guarantees for the reweighted algorithm of \citet{you2025wasserstein} do not help here, since they concern the iterates from a given initialization rather than the minimizer itself. 
Here we verify the uniqueness of Wasserstein median in two simple cases. (i) In a location family, the Wasserstein median reduces to the Euclidean median of the family of centers, which is unique unless the centers lie on a line
\citep{kuhn1973note, 10.1214/aos/1176350511}. 
(ii) For a family of nondegenerate Gaussians with diagonal covariance matrices, the $\mathcal{W}_2$ distance between any two such Gaussian distributions coincides with the Euclidean distance between the two centers plus the Euclidean distance between the two sets of diagonals
\citep{gelbrich1990formula, bhatia2019bures}. The Wasserstein median of such a Gaussian family is unique. In these two cases, the Wasserstein median is unique, and hence the reweighted update is exactly Weiszfeld's iteration \citep{kuhn1973note, beck2015weiszfeld}, whose convergence is guaranteed. 

\paragraph{Proof of Theorem~\ref{thm:W1-kmedians}.}
Write $g^{(1)}_A(\mu) := \min_{\nu \in A} \mathcal{W}_2(\mu, \nu)$ and
$D := \mathrm{diam}(\mathcal{P}_2(X))$. Existence of population and
empirical minimizers follows as for $\Phi$: for fixed $\mathbf{Q}$,
$A \mapsto \Phi_1(A, \mathbf{Q})$ is continuous in the Hausdorff
metric (by the Lipschitz bound below) on the compact set of center
sets of cardinality at most $k$.

\emph{Level separation.} Assumption~(iii) with $\Phi_1$ in place of
$\Phi$ is in fact implied by the uniqueness assumption. Suppose
$m^{(1)}_{k-1}(\mathbf{P}) = m^{(1)}_k(\mathbf{P})$ for some $k \ge 2$.
A $(k-1)$-optimal set $A$ is then also $k$-optimal, and so is
$A \cup \{\nu\}$ for any $\nu \notin A$, since adding a center cannot
increase $\Phi_1$. These are two distinct $k$-optimal sets,
contradicting the uniqueness of the population $k$-medians
configuration $A^{(1)}_k$. Hence
$m^{(1)}_{k-1}(\mathbf{P}) > m^{(1)}_k(\mathbf{P})$.

\emph{Lipschitz bound.} For center sets $A, B$ and every $\mu$,
$|g^{(1)}_A(\mu) - g^{(1)}_B(\mu)| \le H(A, B)$: given $\nu \in A$
attaining $g^{(1)}_A(\mu)$, pick $\nu' \in B$ with
$\mathcal{W}_2(\nu, \nu') \le H(A, B)$, so
$g^{(1)}_B(\mu) \le \mathcal{W}_2(\mu, \nu') \le g^{(1)}_A(\mu) + H(A, B)$,
and symmetrically. Consequently
$|\Phi_1(A, \mathbf{Q}) - \Phi_1(B, \mathbf{Q})| \le H(A, B)$ for every
$\mathbf{Q}$, which replaces Lemma~\ref{lem:lip-phi} with the constant
$1$ in place of $2D + H(A, B)$.

\emph{Uniform law of large numbers.} The proof of
Lemma~\ref{lem:uslln} used only the Lipschitz bound, the compactness of
the family of center sets, and the bounded envelope $D^2$ of $g_A$;
with the bound above and the envelope $D$ of $g^{(1)}_A$, the same
finite-net argument gives
$\sup_A |\Phi_1(A, \mathbf{P}_n) - \Phi_1(A, \mathbf{P})| \to 0$ almost
surely and $m^{(1)}_k(\mathbf{P}_n) \to m^{(1)}_k(\mathbf{P})$ almost
surely.

\emph{Steps 1--3.} Step~1 of the proof of Theorem~\ref{thm:W2-kmeans}
goes through with the lower bound $(M - r)^2 p_r$ replaced by
$(M - r)\, p_r$ and the singleton bound $c_0$ replaced by
$\int \mathcal{W}_2(\mu, \mu_0)\, \mathbf{P}(d\mu) \le D$. Step~2 uses
only the triangle inequality, the uniform law of large numbers, and
the level separation established above. Step~3 uses Hausdorff
compactness, the Lipschitz bound, the uniform law of large numbers,
and the uniqueness of $A^{(1)}_k$, exactly as before. The consequences
for Theorem~\ref{thm:rwsd-abstract} and
Corollary~\ref{cor:Kstar-consistency} follow because those results
only require almost-sure Hausdorff consistency of the centers used by
the pipeline. \qed

We now verify the uniqueness of single-cluster medians in the two
families mentioned above. For $\mu \in \mathcal{P}_2(\mathbb{R}^d)$,
write $m_\mu$ for its mean and $\bar\mu$ for its centered version, the
law of $Z - m_\mu$ under $Z \sim \mu$.

\begin{lemma}[Mean--dispersion decomposition]\label{lem:mean-disp}
For $\mu, \nu \in \mathcal{P}_2(\mathbb{R}^d)$,
$\mathcal{W}_2^2(\mu, \nu) = \|m_\mu - m_\nu\|^2 + \mathcal{W}_2^2(\bar\mu, \bar\nu)$.
\end{lemma}

\begin{proof}
For any coupling $(Z, Z')$ of $\mu$ and $\nu$,
$\mathbb{E}\|Z - Z'\|^2 = \|m_\mu - m_\nu\|^2 + \mathbb{E}\|(Z - m_\mu) - (Z' - m_\nu)\|^2$,
because the cross term $2\langle m_\mu - m_\nu, \mathbb{E}[(Z - m_\mu) - (Z' - m_\nu)]\rangle$
vanishes. Translation maps couplings of $(\mu, \nu)$ bijectively onto
couplings of $(\bar\mu, \bar\nu)$, so taking the infimum over
couplings on both sides gives the identity.
\end{proof}

\begin{proposition}[Location families]\label{prop:location-median}
Let $\rho \in \mathcal{P}_2(\mathbb{R}^d)$ be centered and let a
cluster have law $\mathbf{Q}$ on its members, supported on the
location family $\{\rho_c := (x \mapsto x + c)_{\#}\rho : c \in \mathbb{R}^d\}$;
write $Q$ for the law of the location $c$. Then every minimizer of
$F(\nu) := \int \mathcal{W}_2(\nu, \gamma)\, \mathbf{Q}(d\gamma)$ over
$\mathcal{P}_2(\mathbb{R}^d)$ is of the form $\rho_{u}$ with $u$ a
geometric median of $Q$, that is, a minimizer of
$u \mapsto \int \|u - c\|\, Q(dc)$. In particular the minimizer is
unique whenever $Q$ is not supported on a line.
\end{proposition}

\begin{proof}
For $\nu \in \mathcal{P}_2(\mathbb{R}^d)$ with mean $u$ set
$\delta := \mathcal{W}_2(\bar\nu, \rho)$. By Lemma~\ref{lem:mean-disp},
$\mathcal{W}_2(\nu, \rho_c) = \sqrt{\|u - c\|^2 + \delta^2}$, so
$F(\nu) = \int \sqrt{\|u - c\|^2 + \delta^2}\, Q(dc) \ge \int \|u - c\|\, Q(dc) = F(\rho_u)$,
with equality if and only if $\delta = 0$, that is, $\nu = \rho_u$.
Hence every minimizer of $F$ lies in the family and its location
minimizes $u \mapsto \int \|u - c\|\, Q(dc)$. The geometric median of
a law on $\mathbb{R}^d$ is unique unless the law is supported on a
line \citep{10.1214/aos/1176350511}; for a finitely supported $Q$
this is the classical statement for Fermat's problem
\citep{kuhn1973note}.
\end{proof}

For a cluster of Gaussians with commuting covariances we may rotate
to a common eigenbasis, so it suffices to treat the family
$\mathcal{G} := \{N(m, \mathrm{diag}(s_1^2, \ldots, s_d^2)) : m \in \mathbb{R}^d,\ s \in [0, \infty)^d\}$.
For $\gamma = N(m, \mathrm{diag}(s^2))$ and
$\gamma' = N(m', \mathrm{diag}(s'^2))$ in $\mathcal{G}$,
$\mathcal{W}_2^2(\gamma, \gamma') = \|m - m'\|^2 + \|s - s'\|^2$
\citep{gelbrich1990formula, bhatia2019bures}, so
$\theta(\gamma) := (m, s) \in \mathbb{R}^{2d}$ is an isometry of
$\mathcal{G}$ onto the convex set $\mathbb{R}^d \times [0, \infty)^d$.
For $\nu \in \mathcal{P}_2(\mathbb{R}^d)$ with marginals $\nu_j$ of mean
$m_j$ and variance $v_j$, define the retraction
$G\nu := N(m, \mathrm{diag}(v_1, \ldots, v_d)) \in \mathcal{G}$.

\begin{lemma}[Gaussian retraction]\label{lem:retraction}
For every $\nu \in \mathcal{P}_2(\mathbb{R}^d)$ and every
$\gamma \in \mathcal{G}$, $\mathcal{W}_2(G\nu, \gamma) \le \mathcal{W}_2(\nu, \gamma)$.
If moreover $\gamma$ is nondegenerate, that is, all its standard
deviations $s_j$ are positive, then equality holds if and only if
$\nu = G\nu$.
\end{lemma}

\begin{proof}
Let $\gamma = N(\mu, \mathrm{diag}(s^2))$, with one-dimensional
marginals $\gamma_j = N(\mu_j, s_j^2)$. For any coupling $(Z, Z')$ of
$\nu$ and $\gamma$, $\mathbb{E}\|Z - Z'\|^2 = \sum_j \mathbb{E}(Z_j - Z'_j)^2 \ge \sum_j \mathcal{W}_2^2(\nu_j, \gamma_j)$,
since $(Z_j, Z'_j)$ is a coupling of $(\nu_j, \gamma_j)$. For each
$j$, Lemma~\ref{lem:mean-disp} and the triangle inequality with the
point mass at zero give
$\mathcal{W}_2^2(\nu_j, \gamma_j) = (m_j - \mu_j)^2 + \mathcal{W}_2^2(\bar\nu_j, \bar\gamma_j) \ge (m_j - \mu_j)^2 + (\sqrt{v_j} - s_j)^2$,
because $\mathcal{W}_2(\bar\nu_j, \delta_0) = \sqrt{v_j}$ and
$\mathcal{W}_2(\bar\gamma_j, \delta_0) = s_j$. Summing over $j$ and
using the formula for $\mathcal{W}_2$ on $\mathcal{G}$,
$\mathcal{W}_2^2(\nu, \gamma) \ge \sum_j [(m_j - \mu_j)^2 + (\sqrt{v_j} - s_j)^2] = \mathcal{W}_2^2(G\nu, \gamma)$.

Now let all $s_j$ be positive and suppose equality holds; taking an
optimal coupling of $(\nu, \gamma)$ in the display above, every
inequality in the chain is an equality, and since each of the two
chains is a sum of termwise inequalities, equality holds term by
term. Equality in the second chain reads
$\mathcal{W}_2(\bar\nu_j, \bar\gamma_j) = |\mathcal{W}_2(\bar\nu_j, \delta_0) - \mathcal{W}_2(\bar\gamma_j, \delta_0)|$.
The one-dimensional Wasserstein space is isometric, through quantile
functions, to a convex subset of $L^2(0, 1)$, so its geodesics are
unique straight segments, and equality in the reverse triangle
inequality with $\bar\gamma_j \ne \delta_0$ forces the quantile
function of $\bar\nu_j$ to be a nonnegative multiple of that of
$\bar\gamma_j$; hence $\bar\nu_j$ is a centered Gaussian, possibly
degenerate, and $\nu_j$ is Gaussian for every $j$. Equality in the
first chain means that the optimal coupling of $(\nu, \gamma)$ is
simultaneously optimal for every pair of marginals; since $\gamma_j$
is absolutely continuous, the optimal coupling of $(\nu_j, \gamma_j)$
is unique and given by the monotone map $Z_j = T_j(Z'_j)$, and since
the coordinates of $Z' \sim \gamma$ are independent, so are the
coordinates of $Z$. Thus $\nu$ is the product of its Gaussian
marginals, that is, $\nu = G\nu$. The converse is trivial.
\end{proof}

\begin{proposition}[Commuting Gaussian families]\label{prop:gauss-median}
Let a cluster have law $\mathbf{Q}$ supported on the nondegenerate
members of $\mathcal{G}$, and write $Q := \theta_{\#}\mathbf{Q}$ for the
law of its parameters in $\mathbb{R}^d \times (0, \infty)^d$. Then every minimizer of
$F(\nu) := \int \mathcal{W}_2(\nu, \gamma)\, \mathbf{Q}(d\gamma)$ over
$\mathcal{P}_2(\mathbb{R}^d)$ lies in $\mathcal{G}$, and its parameter
is a geometric median of $Q$; the minimizer is unique whenever $Q$ is
not supported on a line.
\end{proposition}

\begin{proof}
By Lemma~\ref{lem:retraction}, $F(G\nu) \le F(\nu)$ for every $\nu$,
with equality only if $\nu = G\nu$; hence every minimizer lies in
$\mathcal{G}$. On $\mathcal{G}$, $F(\gamma) = \int \|\theta(\gamma) - \theta(\gamma')\|\, \mathbf{Q}(d\gamma') = \int \|\theta(\gamma) - q\|\, Q(dq)$,
so the minimizers are the elements of $\mathcal{G}$ whose parameters
are geometric medians of $Q$. A geometric median of $Q$ has all
coordinates $s_j$ positive: at a parameter with $s_j = 0$ the
$j$-th component of the gradient of $q \mapsto \int \|q - q'\|\, Q(dq')$
is $-\int (q'_j / \|q - q'\|)\, Q(dq') < 0$, so the objective
strictly decreases in the direction of increasing $s_j$. It is
unique unless $Q$ is supported on a line
\citep{10.1214/aos/1176350511}.
\end{proof}

\paragraph{From single clusters to configurations.}
In either family, let $\mathbf{P}$ be supported on the family (on its
nondegenerate members in the Gaussian case) and let $Q$ be the law of
the parameters. For a center set $A$ with
parameters $\theta(A)$ after retraction, the two propositions give
$\Phi_1(A, \mathbf{P}) \ge \Phi_1(GA, \mathbf{P}) = \int \min_{a \in \theta(GA)} \|q - a\|\, Q(dq)$,
so $m^{(1)}_k(\mathbf{P})$ equals the Euclidean $K$-medians value of
$Q$ at level $k$. Under level separation every center of a minimizer
attracts positive $\mathbf{P}$-mass, and a center outside the family
could then be replaced by its retraction, which is strictly closer to
every member it attracts; hence every minimizer lies in the family.
The uniqueness assumption of Theorem~\ref{thm:W1-kmedians} for
$\mathbf{P}$ is therefore the uniqueness of the Euclidean $K$-medians
configuration of $Q$, the identification assumption of
\citet{pollard1981strong}, whose consistency theory covers the
criterion $\phi(t) = t$. In the location family the iteratively
reweighted update of Algorithm~\ref{alg:Wmed-IRLS} reduces to
Weiszfeld's iteration on the locations, which converges to the
geometric median from any initialization outside the data
\citep{kuhn1973note, beck2015weiszfeld}.

\subsection{Proof of Theorem~\ref{thm:rwsd-abstract}}
\label{app:rwsd-abstract-proof}

We prove the consistency of empirical RWSD under any Wasserstein $K$-center clustering
method that produces strongly consistent cluster centers. The argument combines three
ingredients: (i) the assumed Hausdorff convergence of empirical cluster centers;
(ii) consistency of the empirical WSD plug-in estimator under two-stage sampling
\citep[Theorem~6.4]{bachoc2024wasserstein}; and (iii) the continuous mapping theorem.

\paragraph{Step 1: Stability of cluster assignments.}
Let $\{A_k^\star\}_{k=1}^K$ denote the unique population centers and
$\{\hat{A}_{n,k}\}_{k=1}^K$ their empirical counterparts, with
$\hat{A}_{n,k} \to A_k^\star$ almost surely in the Hausdorff metric induced by $\mathcal{W}_2$
by hypothesis. For any $\mu \in \mathcal{P}_2(\mathbb{R}^d)$ that has a unique
nearest population center $A_{k(\mu)}^\star$, the continuity of
$A \mapsto \mathcal{W}_2(\mu, A)$ on the space of finite center sets under Hausdorff distance
implies that, almost surely for $n$ large enough, $\mu$ is assigned to the empirical
cluster $\hat{\mathcal{C}}_{k(\mu)}^{(n)}$ corresponding to $\hat{A}_{n, k(\mu)}$.
By the assumption of Theorem~\ref{thm:rwsd-abstract} that
$\mathbf{P}$-almost every $\mu$ has a unique nearest population center,
this holds for $\mathbf{P}$-almost every $\mu$.

\paragraph{Step 2: Convergence of empirical cluster distributions.}
Let $\hat{\mathbf{P}}_{n,m}^{(k)} := \frac{1}{|\hat{\mathcal{C}}_k^{(n)}|}
\sum_{\hat\mu_{i,m} \in \hat{\mathcal{C}}_k^{(n)}} \delta_{\hat\mu_{i,m}}$
be the empirical distribution over distributions in cluster $k$, and
$\mathbf{P}^{(k)}$ its population counterpart restricted to $\mathcal{C}_k$.
By Step 1 and \citet[Lemma~6.3]{bachoc2024wasserstein},
$\hat{\mathbf{P}}_{n,m}^{(k)} \to \mathbf{P}^{(k)}$ in $d_{\mathrm{BL}(2)}$ in probability
as $n, m \to \infty$, for each $k = 1, \ldots, K$.

\paragraph{Step 3: Pointwise consistency of empirical WSD on each cluster.}
Fix $q$ and let $k_q := k(\mu_q)$ denote the cluster index of $\mu_q$.
By \citet[Theorem~6.4]{bachoc2024wasserstein} applied to the cluster-restricted
distribution-of-distributions $\mathbf{P}^{(k)}$, for each $k = 1, \ldots, K$,
\[
\widehat{\mathrm{WSD}}\!\left(\hat\mu_{q,m};\, \hat{\mathbf{P}}_{n,m}^{(k)}\right)
\;\xrightarrow{P}\; \mathrm{WSD}\!\left(\mu_q;\, \mathbf{P}^{(k)}\right)
\quad \text{as } n, m \to \infty,
\]
provided $\mathbf{P}^{(k)}$ is atomless, which is inherited from the atomlessness
of $\mathbf{P}$.

\paragraph{Step 4: Continuous mapping.}
By Steps 1--3 and the continuous mapping theorem applied to the function
$(a, b_1, \ldots, b_{K-1}) \mapsto a - \max_j b_j$,
\[
\widehat{\mathrm{RWSD}}_K(\hat\mu_{q,m})
= \widehat{\mathrm{WSD}}\!\left(\hat\mu_{q,m};\, \hat{\mathbf{P}}_{n,m}^{(k_q)}\right)
- \max_{j \neq k_q} \widehat{\mathrm{WSD}}\!\left(\hat\mu_{q,m};\, \hat{\mathbf{P}}_{n,m}^{(j)}\right)
\xrightarrow{P} \mathrm{RWSD}_K(\mu_q),
\]
which completes the proof. \qed

\subsection{Proof of Corollary~\ref{cor:Kstar-consistency}}
\label{app:Kstar-proof}

We first record the per-object consistency on which the argument
builds. Fix $K \in \mathcal{K}$. By Theorem~\ref{thm:W2-kmeans}, the
empirical $K$-means centers $\widehat{A}_{n,K}$ converge almost surely in
the Hausdorff metric to $A^\star_K$, which is unique up to permutation
by assumption~(ii). Since $\mathbf{P}$ is atomless and
$\mathbf{P}$-almost every $\mu$ has a unique nearest center in
$A^\star_K$, the hypotheses of Theorem~\ref{thm:rwsd-abstract} are
satisfied with $\hat{A}_{n,k}$ taken to be the empirical $K$-means
centers, and therefore
\begin{equation}\label{eq:perobject-K}
\mathrm{RWSD}_K(\hat\mu_{i,m}) \;\xrightarrow{P}\; \mathrm{RWSD}_K(\mu_i)
\qquad (n, m \to \infty)
\end{equation}
for $\mathbf{P}$-almost every $\mu_i$. Because $\mathcal{K}$ is finite,
\eqref{eq:perobject-K} holds simultaneously for all $K \in \mathcal{K}$.
The same argument applies to the $K$-medians pipeline, with
Theorem~\ref{thm:W1-kmedians} in place of Theorem~\ref{thm:W2-kmeans}.

The remainder of the proof has two steps. Step~1 upgrades the
per-object consistency~\eqref{eq:perobject-K} to the averaged criterion
$\widehat{R}_{n,m}(K)$, uniformly over the finite candidate set;
Step~2 is a deterministic margin argument identifying the maximizer
$K^\star$.
\paragraph{Step 1: Uniform consistency of the criterion.}
We claim that
\begin{equation}\label{eq:uniform-criterion}
    \max_{K\in\mathcal{K}}
    \bigl|\widehat{R}_{n,m}(K)-R(K)\bigr|
    \;\xrightarrow{P}\;0
    \qquad\text{as } n,m\to\infty .
\end{equation}
Since $\mathcal{K}$ is finite, it suffices to fix $K\in\mathcal{K}$ and
decompose
\[
    \widehat{R}_{n,m}(K)-R(K)
    =\underbrace{\frac1n\sum_{i=1}^n
    \bigl[\mathrm{RWSD}_K(\hat\mu_{i,m})-\mathrm{RWSD}_K(\mu_i)\bigr]}
    _{(\mathrm{I})}
    \;+\;\underbrace{\frac1n\sum_{i=1}^n\mathrm{RWSD}_K(\mu_i)-R(K)}
    _{(\mathrm{II})},
\]
where $\mathrm{RWSD}_K(\mu_i)$ denotes the population quantity, defined
through the population partition induced by $A^\star_K$.

\emph{Term $(\mathrm{II})$.}
The variables $\mathrm{RWSD}_K(\mu_1),\ldots,\mathrm{RWSD}_K(\mu_n)$ are
i.i.d.\ and bounded, $|\mathrm{RWSD}_K|\le1$, since
$\mathrm{WSD}\in[0,1]$; the strong law of large numbers gives
$(\mathrm{II})\to0$ almost surely.

\emph{Term $(\mathrm{I})$.}
Set
$Z_i:=\bigl|\mathrm{RWSD}_K(\hat\mu_{i,m})-\mathrm{RWSD}_K(\mu_i)\bigr|
\le2$.
The data triples $(\mu_i,X_{i,1},\ldots,X_{i,m})$, $i=1,\ldots,n$, are
i.i.d.\ across $i$, and the clustering-and-assignment pipeline is a
permutation-equivariant function of the sample (the $K$-centers step
returns a global minimizer, with ties broken by a fixed measurable rule
that does not depend on the ordering of the observations); hence the
vector $(Z_1,\ldots,Z_n)$ is exchangeable, and
\[
    \mathbb{E}\bigl|(\mathrm{I})\bigr|
    \;\le\;\mathbb{E}\Bigl[\frac1n\sum_{i=1}^nZ_i\Bigr]
    \;=\;\mathbb{E}[Z_1].
\]
We emphasize that only equality of the expectations
$\mathbb{E}[Z_i]=\mathbb{E}[Z_1]$ is used here; no independence between
$Z_i$ and the data-dependent partition is required.
By \eqref{eq:perobject-K}, $Z_1\to0$ in probability
conditionally on $\mathbf{P}$-almost every realization of $\mu_1$;
integrating over $\mu_1$ by dominated convergence gives the
unconditional convergence
$\mathbb{P}(Z_1>\varepsilon)\to0$ for every $\varepsilon>0$.
Since $0\le Z_1\le2$,
\[
    \mathbb{E}[Z_1]
    \;\le\;\varepsilon+2\,\mathbb{P}(Z_1>\varepsilon)
    \;\longrightarrow\;\varepsilon
    \qquad\text{for every }\varepsilon>0,
\]
so $\mathbb{E}[Z_1]\to0$, and Markov's inequality yields
$(\mathrm{I})\xrightarrow{P}0$.
Combining the two terms and taking the maximum over the finite set
$\mathcal{K}$ proves \eqref{eq:uniform-criterion}.

\paragraph{Step 2: Margin reduction.}
Consider the event
\[
    E_{n,m}:=\Bigl\{\max_{K\in\mathcal{K}}
    \bigl|\widehat{R}_{n,m}(K)-R(K)\bigr|<\delta/2\Bigr\},
\]
whose probability tends to one by \eqref{eq:uniform-criterion} and the
margin condition \eqref{eq:margin}.
On $E_{n,m}$, for every $K\in\mathcal{K}$ with $K\neq K^\star$,
\[
    \widehat{R}_{n,m}(K)
    \;<\;R(K)+\frac\delta2
    \;\le\;R(K^\star)-\delta+\frac\delta2
    \;=\;R(K^\star)-\frac\delta2
    \;<\;\widehat{R}_{n,m}(K^\star),
\]
so the maximizer of $\widehat{R}_{n,m}$ over $\mathcal{K}$ is
unique and equals $K^\star$, i.e.\ $\widehat K_{n,m}=K^\star$.
Hence
$\mathbb{P}\bigl(\widehat K_{n,m}=K^\star\bigr)
\ge\mathbb{P}(E_{n,m})\to1$.
\qed

\subsection{Proof of Proposition~\ref{prop:rwsd-rate}}
\label{app:rate-proof}

We bound the error of the empirical RWSD in three steps. Let $\mathcal{C}_{k_\mu}$ be the population cluster that contains $\mu$.
For each $k$, let $n_k\sim\mathrm{Binomial}(n,\pi_k)$ be the number of
sampled distributions $\mu_i$ in $\mathcal{C}_k$, and let
$\mathcal{C}^k_n$ be the latent cluster formed by the distributions
$\mu_i\in\mathcal{C}_k$ themselves, before each $\mu_i$ is replaced by
$\hat\mu_{i,m}$. If $n_k=0$, both $\mathcal{C}^k_{n,m}$ and
$\mathcal{C}^k_n$ are empty, and we set $\mathrm{WSD}(\mu;\emptyset):=0$.

\paragraph{Step 1: Per-cluster depth errors.}
Since $|\max_jb_j-\max_jb_j'|\le\max_j|b_j-b_j'|$, all real numbers
$a,a',b_j,b_j'$ satisfy
\[
\big|(a-\max_jb_j)-(a'-\max_jb_j')\big|
\;\le\;2\max\big\{|a-a'|,\,\max_j|b_j-b_j'|\big\}.
\]
We apply the inequality with $a$ and $b_j$ the depths of $\mu$ in
$\mathcal{C}_{k_\mu}$ and $\mathcal{C}_j$, $j\ne k_\mu$, and with $a'$
and $b_j'$ the depths of $\mu$ in $\mathcal{C}^{k_\mu}_{n,m}$ and
$\mathcal{C}^j_{n,m}$. By the definition of RWSD,
\begin{equation}\label{eq:lip-reduction}
\big|\widehat{\mathrm{RWSD}}_K(\mu)-\mathrm{RWSD}_K(\mu)\big|
\;\le\;
2\max_{1\le k\le K}
\big|\mathrm{WSD}\big(\mu;\mathcal{C}^k_{n,m}\big)
-\mathrm{WSD}\big(\mu;\mathcal{C}_k\big)\big| .
\end{equation}

\paragraph{Step 2: Two sources of error.}
The latent cluster $\mathcal{C}^k_n$ separates the error from observing
each distribution through $m$ points (second stage) from the error from
sampling finitely many distributions (first stage):
\begin{equation}\label{eq:two-sources}
\big|\mathrm{WSD}\big(\mu;\mathcal{C}^k_{n,m}\big)
-\mathrm{WSD}\big(\mu;\mathcal{C}_k\big)\big|
\;\le\;
\big|\mathrm{WSD}\big(\mu;\mathcal{C}^k_{n,m}\big)
-\mathrm{WSD}\big(\mu;\mathcal{C}^k_{n}\big)\big|
+\big|\mathrm{WSD}\big(\mu;\mathcal{C}^k_{n}\big)
-\mathrm{WSD}\big(\mu;\mathcal{C}_k\big)\big| .
\end{equation}

\paragraph{Step 3a: Second-stage error.}
Conditionally on $n_k$, the distributions in $\mathcal{C}^k_n$ are
i.i.d.\ draws from $\mathbf{P}^{(k)}$, each observed through $m$ i.i.d.\
points, and by assumption, $\mu$ and $\mathbf{P}^{(k)}$ satisfy
Assumption~6.5 of \citet{bachoc2024wasserstein}. The proof of their
Lemma~6.6 averages per-distribution bounds. By their Lemma~F.1, the
$i$-th distribution contributes at most
$2\|T_{\mu,\hat\mu_{i,m}}-T_{\mu,\mu_i}\|_{L^2(\mu)}/\mathcal{W}_2(\mu,\mu_i)$,
and their display~(6.8) with Jensen's inequality bounds the conditional
expectation of the contribution given $\mu_i$ by
$2\sqrt{\alpha(d,m)}/\mathcal{W}_2(\mu,\mu_i)$. The same argument
applies to $\mathcal{C}^k_n$ for every $n_k\ge1$ and gives
\begin{equation}\label{eq:m-term}
\mathbb{E}\Big[
\big|\mathrm{WSD}\big(\mu;\mathcal{C}^k_{n,m}\big)
-\mathrm{WSD}\big(\mu;\mathcal{C}^k_{n}\big)\big|
\,\Big|\,n_k\Big]
\;\le\;
2\sqrt{\alpha(d,m)}\;
\mathbb{E}_{P'\sim\mathbf{P}^{(k)}}
\big[\mathcal{W}_2^{-1}(\mu,P')\big]
\;\le\;2\sqrt{\alpha(d,m)}\,M ;
\end{equation}
for $n_k=0$, the left-hand side is zero. The constant $C'$ in
$\alpha(d,m)$ comes from display~(6.8) of \citet{bachoc2024wasserstein}
and depends only on $d$, the supports, and the density bounds in their
Assumption~6.5.

\paragraph{Step 3b: First-stage error.}
For each $i$ with $\mu_i\in\mathcal{C}_k$, let
$\xi_i:=(I-T_{\mu,\mu_i})/\|I-T_{\mu,\mu_i}\|_{L^2(\mu)}$. The unit
vector $\xi_i\in L^2(\mu)$ is well defined because $M<\infty$ forces
$\mathcal{W}_2(\mu,\mu_i)>0$ almost surely. Conditionally on $n_k$, the
$\xi_i$ are i.i.d.\ with mean
\[
S_k:=\mathbb{E}_{P'\sim\mathbf{P}^{(k)}}
\Big[\frac{I-T_{\mu,P'}}{\|I-T_{\mu,P'}\|_{L^2(\mu)}}\Big],
\]
and with $S_{k,n}:=n_k^{-1}\sum_{i:\,\mu_i\in\mathcal{C}_k}\xi_i$, the
two depths are $\mathrm{WSD}(\mu;\mathcal{C}_k)=1-\|S_k\|_{L^2(\mu)}$ and
$\mathrm{WSD}(\mu;\mathcal{C}^k_n)=1-\|S_{k,n}\|_{L^2(\mu)}$. The
triangle inequality, Jensen's inequality, and
$\mathbb{E}\big[\|S_{k,n}-S_k\|_{L^2(\mu)}^2\,\big|\,n_k\big]
=(1-\|S_k\|_{L^2(\mu)}^2)/n_k$ give
\begin{equation}\label{eq:n-term}
\mathbb{E}\Big[
\big|\mathrm{WSD}\big(\mu;\mathcal{C}^k_{n}\big)
-\mathrm{WSD}\big(\mu;\mathcal{C}_k\big)\big|
\,\Big|\,n_k\Big]
\;\le\;
\mathbb{E}\big[\|S_{k,n}-S_k\|_{L^2(\mu)}\,\big|\,n_k\big]
\;\le\;n_k^{-1/2}.
\end{equation}
Both depths lie in $[0,1]$, so the left-hand side of \eqref{eq:n-term}
is also at most one, which covers $n_k=0$. Since
$n_k\sim\mathrm{Binomial}(n,\pi_k)$, the Chernoff bound
$\mathbb{P}(n_k\le\pi_kn/2)\le e^{-\pi_kn/8}$ gives
\begin{equation}\label{eq:nk-moment}
\mathbb{E}\big[\min\{n_k^{-1/2},1\}\big]
\;\le\;\sqrt{2}\,(\pi_kn)^{-1/2}+e^{-\pi_kn/8}
\;\le\;C_0\,(\pi_kn)^{-1/2},
\end{equation}
where $\min\{n_k^{-1/2},1\}:=1$ when $n_k=0$, and
$C_0:=\sqrt2+\sup_{x>0}\sqrt{x}\,e^{-x/8}$ is an absolute constant.

\paragraph{Conclusion.}
Since a maximum of nonnegative terms is at most their sum, equation
\eqref{eq:lip-reduction} and \eqref{eq:two-sources} are bound
$\mathbb{E}\,\Delta_K(\mu)$ by twice the sum over $k$ of the expected
second-stage and first-stage errors. By \eqref{eq:m-term}, the expected
second-stage error is at most $2\sqrt{\alpha(d,m)}\,M$, and by
\eqref{eq:n-term} and \eqref{eq:nk-moment}, the expected first-stage
error is at most $C_0\,(\pi_kn)^{-1/2}$. Therefore
\[
\mathbb{E}\,\Delta_K(\mu)
\;\le\;
2\sum_{k=1}^{K}\Big[2\sqrt{\alpha(d,m)}\,M+C_0\,(\pi_kn)^{-1/2}\Big]
\;=\;
4K\sqrt{\alpha(d,m)}\,M
+\frac{2C_0}{\sqrt{n}}\sum_{k=1}^{K}\pi_k^{-1/2},
\]
which is the bound of Proposition~\ref{prop:rwsd-rate} with $C:=2C_0$.
In particular, when $K$, the cluster probabilities $\pi_k$, and $M$ are
fixed, the error of the empirical RWSD converges to zero at the rate
\[
\mathbb{E}\,\Delta_K(\mu)
=
\begin{cases}
O\big(m^{-1/2}+n^{-1/2}\big) & \text{if } d=1,\\[2pt]
O\big(\sqrt{(\log m)/m}+n^{-1/2}\big) & \text{if } d=2,\\[2pt]
O\big(m^{-1/d}+n^{-1/2}\big) & \text{if } d>2.
\end{cases}
\]
\qed
\subsection{Proof of Proposition~\ref{prop:influence}}

In this subsection and the next, we write every depth through a kernel
$\xi(Q,P)\in L^2(Q)$ with $\|\xi(Q,P)\|_{L^2(Q)}\le1$:
\[
\mathrm{WSD}(Q;C):=1-\Big\|\,|C|^{-1}\sum_{P\in C}\xi(Q,P)\Big\|_{L^2(Q)} .
\]
The population depth uses $\xi(Q,P):=(I-T_{Q,P})/\mathcal{W}_2(Q,P)$,
which has norm one for $P\ne Q$. The empirical depth of
Section~\ref{sec:propose_RWSD} replaces $T_{Q,P}$ by the barycentric
projection of the optimal plan, which gives norm at most one by
Jensen's inequality. In both cases, $\xi(Q,Q):=0$.

Fix $K$. Let $I_l$ be the index set of the $l$-th cluster $C_l$, let
$n_l:=|I_l|$, and let $k(i)$ be the index of the cluster that contains
$\nu_i$. Let $J\subset\{1,\dots,n\}$, $|J|=s$, be the indices of the
replaced measures, let $s_l:=|J\cap I_l|$, and write $\nu_j'$ for the
measure at index $j$ after the replacement, so that $\nu_j'=\nu_j$ for
$j\notin J$. The criterion is $\overline{R}(K)=n^{-1}\sum_{i=1}^{n}t_i$
with
\[
t_i:=a_i-\max_{l\ne k(i)}b_{i,l},\qquad
a_i:=\mathrm{WSD}(\nu_i;C_{k(i)}),\qquad
b_{i,l}:=\mathrm{WSD}(\nu_i;C_l),
\]
and $t_i'$, $a_i'$, and $b_{i,l}'$ denote the same quantities after the
replacement. Since $\|\xi\|\le1$, every depth lies in $[0,1]$, and every
$t_i$ lies in $[-1,1]$.

\emph{Replaced measures.} For $i\in J$, both $t_i$ and $t_i'$ lie in
$[-1,1]$, so $|t_i'-t_i|\le2$. The $s$ replaced measures therefore
contribute at most $2s/n$ to $|\overline{R}{}'(K)-\overline{R}(K)|$.

\emph{Unreplaced measures.} Fix $i\notin J$. The measure $\nu_i$ itself
is unchanged, so the replacement changes only the cluster averages
\[
S_l(\nu_i):=\frac{1}{n_l}\sum_{j\in I_l}\xi(\nu_i,\nu_j),
\qquad l=1,\dots,K,
\]
which determine $a_i=1-\|S_{k(i)}(\nu_i)\|_{L^2(\nu_i)}$ and
$b_{i,l}=1-\|S_l(\nu_i)\|_{L^2(\nu_i)}$. The replacement changes at most
$s_l$ summands of $S_l(\nu_i)$, each by at most
$\|\xi(\nu_i,\nu_j')-\xi(\nu_i,\nu_j)\|_{L^2(\nu_i)}\le2$, so
$\|S_l'(\nu_i)-S_l(\nu_i)\|_{L^2(\nu_i)}\le2s_l/n_l$. Since
$t\mapsto1-\|t\|$ is $1$-Lipschitz, $|b_{i,l}'-b_{i,l}|\le2s_l/n_l$ and
$|a_i'-a_i|\le2s_{k(i)}/n_{k(i)}$. Together with
$|\max_lb_l-\max_lb_l'|\le\max_l|b_l-b_l'|$, the two bounds give
\[
|t_i'-t_i|
\;\le\;|a_i'-a_i|+\max_{l\ne k(i)}|b_{i,l}'-b_{i,l}|
\;\le\;\frac{2s_{k(i)}}{n_{k(i)}}+\max_{l}\frac{2s_l}{n_l}
\;\le\;4\max_{l}\frac{s_l}{n_l}.
\]

Averaging $|t_i'-t_i|$ over the replaced and the unreplaced measures
gives the first claim of Proposition~\ref{prop:influence}:
\[
|\overline{R}{}'(K)-\overline{R}(K)|
\;\le\;\frac{2s}{n}+4\max_l\frac{s_l}{n_l}
\;\le\;\frac{6s}{n\hat\pi_{\min}},
\]
where the last step uses $s_l\le s$, $n_l\ge n\hat\pi_{\min}$, and
$\hat\pi_{\min}\le1$.

For the second claim, apply the first claim at every $K\in\mathcal{K}$
with the partition for that $K$; $\hat\pi_{\min}$ is the smallest
cluster fraction over the partitions for all $K\in\mathcal{K}$. If
$s<\hat\delta\,\hat\pi_{\min}\,n/12$, then
$6s/(n\hat\pi_{\min})<\hat\delta/2$, so every $\overline{R}(K)$ changes
by less than $\hat\delta/2$. For every $K\ne\widehat{K}$, the margin
gives $\overline{R}(K)\le\overline{R}(\widehat{K})-\hat\delta$, so
\[
\overline{R}{}'(K)
\;<\;\overline{R}(K)+\hat\delta/2
\;\le\;\overline{R}(\widehat{K})-\hat\delta/2
\;<\;\overline{R}{}'(\widehat{K}).
\]
Hence $\widehat{K}$ is the unique maximizer of $\overline{R}{}'$. \qed

Appendix~\ref{app:robust-proofs} gives a population version of
Proposition~\ref{prop:influence}.

\subsection{Population version of Proposition~\ref{prop:influence}}
\label{app:robust-proofs}

\begin{corollary}[Breakdown of the selected $K$, population version]
\label{cor:breakdown}
Under the assumptions of Corollary~\ref{cor:Kstar-consistency}, let
$\pi^{(K)}_l$ denote the mass of the $l$-th population cluster at level
$K$, assume that all of them are positive, and set
$\pi_{\min} := \min_{K \in \mathcal{K}} \min_{l \le K} \pi^{(K)}_l$.
Then, with probability tending to one as $n, m \to \infty$, the
selection remains $K^\star$ under every contamination, in the sense of
Proposition~\ref{prop:influence}, of fewer than
$\delta\,\pi_{\min}\,n / 32$ objects.
\end{corollary}

\begin{proof}
Let $\overline{R}$ be the criterion of the uncontaminated sample with
the clusters produced by Algorithm~\ref{alg:selectK}, so that
$\overline{R} = \widehat{R}_{n,m}$, and consider the event
$F_n := \{\sup_{K \in \mathcal{K}} |\overline{R}(K) - R(K)| < \delta / 8\}
\cap \{\hat\pi_{\min} \ge \pi_{\min} / 2\}$. On $F_n$ the empirical
margin of $K^\star$ satisfies
$\hat\delta \ge \delta - 2(\delta / 8) = 3\delta / 4$ and
$\widehat{K} = K^\star$, so by Proposition~\ref{prop:influence} any
contamination of size
$s < (3\delta / 4)(\pi_{\min} / 2)\, n / 12 = \delta\,\pi_{\min}\, n / 32$
leaves the selection at $K^\star$; thus $F_n$ is contained in the
event of the statement, and it remains to show that
$\mathbb{P}(F_n) \to 1$. The first event has probability tending to
one by the uniform consistency \eqref{eq:uniform-criterion} of the
empirical criterion. For the second, fix $K$ and $l$: the empirical
cluster fraction dominates the fraction of objects whose
\emph{population} cluster is $l$ minus the fraction of misassigned
objects; the former tends to $\pi^{(K)}_l$ almost surely by the strong
law of large numbers, and the latter vanishes in probability because,
under the hypotheses of Theorem~\ref{thm:rwsd-abstract}, each object is
assigned to its population cluster with probability tending to one and
the objects are exchangeable. Since $\mathcal{K}$ is finite, the
intersection over $K$ and $l$ retains probability tending to one.
\end{proof}

\subsection{Sensitivity of the silhouette score and DBI}
\label{app:contrast}

\begin{lemma}\label{lem:contrast}
Let $q \ge 4$, $n = 2q$, and hold two clusters of size $q$ fixed. There
are configurations in which replacing a single distribution lowers the
average silhouette score by at least $1/2$, and configurations in which
it raises DBI by an arbitrarily large amount, whether the DBI centers
are held fixed or recomputed as barycenters or medians.
\end{lemma}

\begin{proof}
Realize point configurations as Dirac measures, for which
$\mathcal{W}_2$ between two Diracs is the distance between their
locations; the construction is thus one-dimensional. Cluster $A$
holds $q$ objects at $0$ and cluster $B$ holds $q$ objects at $1$,
and one member of $A$ is replaced by a Dirac at location $L > 0$.

\emph{Silhouette score.} Initially every object has within-distance
$a(i) = 0$ and between-distance $b(i) = 1$, so every silhouette value
is $s(i) = (b - a)/\max(a, b) = 1$ and the average is $1$. Take
$L \ge 2q$. For each of the $q - 1$ remaining members of $A$, the
average distance to the other members of its own cluster becomes
$a(i) = L/(q - 1) > 2 > 1 = b(i)$, hence
$s(i) = (b - a)/a = 1/a - 1 < -1/2$, a decrease of more than $3/2$.
The silhouettes of the $q$ members of $B$ are unchanged, since
$a(i) = 0$ still gives $s(i) = 1$, and the replaced object's
silhouette cannot increase, since it was $1$. The average therefore
decreases by more than $(3/2)(q - 1)/(2q) = 3(q - 1)/(4q)$, which is
at least $9/16 > 1/2$ for every $q \ge 4$.

\emph{DBI.} With two clusters, DBI equals $(S_A + S_B)/d(c_A, c_B)$,
where $S_A$ and $S_B$ are the mean distances of the members to their
centers $c_A$ and $c_B$. Initially $S_A = S_B = 0$, so the index is
$0$; after the replacement still $S_B = 0$ and $c_B$ is the Dirac at
$1$. If the centers are held fixed or recomputed as medians, $c_A$ is
the Dirac at $0$: writing $x_i$ for the locations of the members of
$A$, Jensen's inequality gives, for every $\nu$,
$\sum_i \mathcal{W}_2(\nu, \delta_{x_i}) \ge \int \sum_i |y - x_i|\,\nu(dy) \ge \sum_i |x_i|$,
with equality only at $\nu = \delta_0$. Hence $S_A = L/q$,
$d(c_A, c_B) = 1$, and the index equals $L/q$, which is arbitrarily
large as $L \to \infty$. If the centers are recomputed as barycenters,
$c_A$ becomes the Dirac at $L/q$, so that $S_A = 2(q - 1)L/q^2$ and
$d(c_A, c_B) = |L/q - 1|$, and the index equals
\[
\frac{2(q - 1)L}{q\,|L - q|},
\]
which is arbitrarily large as $L \to q$, where the replaced member
drags the barycenter of $A$ onto that of $B$.
\end{proof}

\section{Additional simulation results}\label{app:regular-sims}
Section~\ref{sec:simulations} showed that RWSD selects the true number
of clusters in stressed settings where the silhouette score, DBI, or
both fail. Each of those settings distorts the within- or
between-cluster distances on which the two classical criteria are
built, while the center--outward ordering used by RWSD is largely
unaffected. This appendix completes the picture in two directions.
Appendix~\ref{app:regular} reports \emph{regular} settings without
such distortions. There all three criteria recover the true $K^\star$, so
RWSD loses nothing in benign geometry, and the competing criteria work
at full strength, which confirms that their failures in
Section~\ref{sec:simulations} are caused by the distortions.
Appendix~\ref{app:failure} reports two settings in which RWSD fails
because the clusters are hollow and have no depth core; there the
silhouette score and DBI perform better. Taken together, the
experiments show agreement in regular settings, an advantage for RWSD
when distances are distorted, and a failure of RWSD when its premise
of centrality does not hold. Every configuration is replicated $20$
times, with fewer objects per cluster than in
Section~\ref{sec:simulations} so that the full sweeps remain
tractable; each subsection states its $n_k$. All tables report means
over replications with $\pm$ one standard deviation, and boldface
marks the optimum of each criterion.

\subsection{Regular settings}\label{app:regular}

\paragraph{Gaussian distributions.}
Each cluster consists of Gaussian distributions with a common
covariance, whose means are scattered around a fixed center:
\begin{equation*}
\text{Cluster } k
= \Big\{\, N\big(\boldsymbol{c}_k,\ 0.15^2\, \bI_{d\times d}\big),\
\boldsymbol{c}_k \sim N\big(\boldsymbol{a}_k,\ 0.25^2\, \bI_{d\times d}\big)
\,\Big\}, \quad k = 1, \ldots, K^\star, \ldots .
\end{equation*}
We vary the number of clusters $K^\star$, the dimension $d$, and the
number $m$ of points per distribution, one at a time, which gives nine
combinations:
\begin{itemize}
    \item with $d = 2$ and $m = 50$, the centers form a segment
    ($K^\star = 2$, $\boldsymbol{a}_{1,2} = (\pm 1.5,\, 0)^\top$), an
    equilateral triangle with side $3.0$ ($K^\star = 3$), a square with
    side $3.0$ ($K^\star = 4$), or five equally spaced points on a circle
    of radius $2.5$ ($K^\star = 5$), so that adjacent centers are
    $2 \cdot 2.5 \sin(\pi/5)$ apart;
    \item with $K^\star = 3$ and $m = 50$, the dimension is
    $d \in \{3, 5, 10\}$, and the side of the triangle is scaled to
    $3.0\sqrt{d/2}$ so that the ratio of between- to within-cluster
    Wasserstein distances does not depend on $d$;
    \item with $K^\star = 3$ and $d = 2$, the number of points per
    distribution is $m \in \{50, 100, 150\}$.
\end{itemize}
We cluster with free-support Wasserstein $K$-means, except in the
circle configuration, where we use Wasserstein $K$-medians. Both
follow the pipeline and hyperparameters of
Section~\ref{sec:simulations}, and the candidate sets are listed in
Table~\ref{tab:regular-values}. Each combination uses $n_k = 50$
distributions per cluster and is replicated $20$ times, and all runs
can be reproduced exactly from the accompanying code.

All three criteria select $K^\star$ in every combination and in
each of the $20$ replications. Table~\ref{tab:regular-values} shows
the full criterion profiles. Every mean profile has a single optimum
at $K^\star$, clearly separated from the runner-up; for RWSD in the
triangle configuration, for instance, the value is $.379$ at $K = 3$
and $.361$ at $K = 2$. The optimal RWSD value is $.379$ in all four
center geometries, because the four designs place neighboring centers
at comparable distances. As the dimension grows, the optimal RWSD
value decreases slightly, from $.379$ at $d = 2$ to $.355$, $.334$,
and $.320$ at $d = 3$, $5$, and $10$. This is expected, since
Wasserstein quantities are harder to estimate from $m = 50$ points in
higher dimension. The decrease changes the values but not the
selection: all three criteria still select $K^\star = 3$ in every
dimension and replication, with a clear margin even at $d = 10$
($.320$ at $K = 3$ against $.307$ at $K = 2$). As $m$ grows, the
optimal values barely change: $.379$, $.381$, and $.382$ for RWSD,
$.844$, $.845$, and $.846$ for the silhouette score, and $.216$,
$.214$, and $.213$ for DBI. The selection is therefore already stable
at $m = 50$, and larger $m$ affects only the third decimal.

\begin{table}[!t]
    \centering
    \caption{Criterion profiles over the candidate sets for the nine Gaussian combinations: means over $20$ replications with $\pm$ one standard deviation. Boldface marks the optimum of each criterion, which lies at $K^\star$ in every combination and in each of the $20$ replications.}
    \label{tab:regular-values}
    {\small
    \setlength{\tabcolsep}{5pt}
    \begin{tabular}{ccc|cccc}
        \toprule
        $d$ & $m$ & $K^\star$ & $K$ & Silhouette & RWSD & DBI \\
        \midrule
        \multirow{3}{*}{2} & \multirow{3}{*}{50} & \multirow{3}{*}{2} & 2 & \textbf{.850 ($\pm$.009)} & \textbf{.379 ($\pm$.002)} & \textbf{.210 ($\pm$.013)} \\
         & &  & 3 & .591 ($\pm$.019) & .318 ($\pm$.008) & .818 ($\pm$.060) \\
         & &  & 4 & .431 ($\pm$.128) & .281 ($\pm$.032) & .980 ($\pm$.191) \\
        \midrule
        \multirow{3}{*}{2} & \multirow{3}{*}{50} & \multirow{3}{*}{3} & 2 & .564 ($\pm$.010) & .361 ($\pm$.003) & .707 ($\pm$.022) \\
         & &  & 3 & \textbf{.844 ($\pm$.008)} & \textbf{.379 ($\pm$.002)} & \textbf{.216 ($\pm$.012)} \\
         & &  & 4 & .667 ($\pm$.014) & .339 ($\pm$.007) & .683 ($\pm$.056) \\
        \midrule
        \multirow{5}{*}{2} & \multirow{5}{*}{50} & \multirow{5}{*}{4} & 2 & .470 ($\pm$.059) & .342 ($\pm$.008) & .933 ($\pm$.099) \\
         & &  & 3 & .635 ($\pm$.007) & .363 ($\pm$.002) & .546 ($\pm$.010) \\
         & &  & 4 & \textbf{.841 ($\pm$.006)} & \textbf{.379 ($\pm$.001)} & \textbf{.217 ($\pm$.009)} \\
         & &  & 5 & .712 ($\pm$.010) & .349 ($\pm$.005) & .575 ($\pm$.038) \\
         & &  & 6 & .584 ($\pm$.009) & .319 ($\pm$.004) & .823 ($\pm$.031) \\
        \midrule
        \multirow{5}{*}{2} & \multirow{5}{*}{50} & \multirow{5}{*}{5} & 2 & .455 ($\pm$.017) & .331 ($\pm$.008) & .897 ($\pm$.010) \\
         & &  & 3 & .547 ($\pm$.038) & .362 ($\pm$.005) & .690 ($\pm$.027) \\
         & &  & 4 & .667 ($\pm$.021) & .366 ($\pm$.004) & .468 ($\pm$.007) \\
         & &  & 5 & \textbf{.837 ($\pm$.005)} & \textbf{.379 ($\pm$.001)} & \textbf{.220 ($\pm$.008)} \\
         & &  & 6 & .736 ($\pm$.009) & .356 ($\pm$.004) & .517 ($\pm$.050) \\
        \midrule
        \multirow{3}{*}{3} & \multirow{3}{*}{50} & \multirow{3}{*}{3} & 2 & .561 ($\pm$.005) & .340 ($\pm$.003) & .713 ($\pm$.012) \\
         & &  & 3 & \textbf{.841 ($\pm$.005)} & \textbf{.355 ($\pm$.002)} & \textbf{.222 ($\pm$.007)} \\
         & &  & 4 & .640 ($\pm$.011) & .301 ($\pm$.006) & .835 ($\pm$.062) \\
        \midrule
        \multirow{3}{*}{5} & \multirow{3}{*}{50} & \multirow{3}{*}{3} & 2 & .556 ($\pm$.005) & .320 ($\pm$.003) & .724 ($\pm$.011) \\
         & &  & 3 & \textbf{.835 ($\pm$.005)} & \textbf{.334 ($\pm$.003)} & \textbf{.235 ($\pm$.006)} \\
         & &  & 4 & .607 ($\pm$.007) & .267 ($\pm$.006) & 1.094 ($\pm$.075) \\
        \midrule
        \multirow{3}{*}{10} & \multirow{3}{*}{50} & \multirow{3}{*}{3} & 2 & .551 ($\pm$.004) & .307 ($\pm$.002) & .738 ($\pm$.010) \\
         & &  & 3 & \textbf{.828 ($\pm$.003)} & \textbf{.320 ($\pm$.002)} & \textbf{.254 ($\pm$.005)} \\
         & &  & 4 & .583 ($\pm$.009) & .244 ($\pm$.014) & 1.431 ($\pm$.230) \\
        \midrule
        \multirow{3}{*}{2} & \multirow{3}{*}{100} & \multirow{3}{*}{3} & 2 & .565 ($\pm$.010) & .363 ($\pm$.003) & .705 ($\pm$.019) \\
         & &  & 3 & \textbf{.845 ($\pm$.009)} & \textbf{.381 ($\pm$.002)} & \textbf{.214 ($\pm$.013)} \\
         & &  & 4 & .670 ($\pm$.017) & .341 ($\pm$.007) & .668 ($\pm$.058) \\
        \midrule
        \multirow{3}{*}{2} & \multirow{3}{*}{150} & \multirow{3}{*}{3} & 2 & .565 ($\pm$.010) & .364 ($\pm$.003) & .705 ($\pm$.020) \\
         & &  & 3 & \textbf{.846 ($\pm$.009)} & \textbf{.382 ($\pm$.002)} & \textbf{.213 ($\pm$.013)} \\
         & &  & 4 & .672 ($\pm$.014) & .343 ($\pm$.008) & .660 ($\pm$.059) \\
        \bottomrule
    \end{tabular}
    }
\end{table}

\paragraph{Beyond the Gaussian family.}
When all members of a configuration share a common shape and differ
only in location, the Wasserstein distance between two members equals
the Euclidean distance between their means
(Lemma~\ref{lem:mean-disp}), and the optimal transport map between
them is a translation. The distances used by the silhouette score and
DBI and the transport directions used by RWSD therefore depend only on
the means, whatever the common shape is. To check that the agreement
observed above does not rely on Gaussianity, we replace the Gaussian
shape by three alternatives and run each of them through the same
nine combinations of $K^\star$, $d$, and $m$. The first alternative has
heavy tails within each object: each object is a multivariate
Student-$t$ distribution,
\begin{equation*}
\text{Cluster } k
= \Big\{\, t_5\big(\boldsymbol{c}_k,\ \tfrac{3}{5}\cdot 0.15^2\,
\bI_{d\times d}\big),\
\boldsymbol{c}_k \sim N\big(\boldsymbol{a}_k,\ 0.25^2\, \bI_{d\times d}\big)
\,\Big\}, \qquad k = 1, \ldots, K^\star, \ldots,
\end{equation*}
where $t_5(\boldsymbol{\mu}, \Sigma)$ is the multivariate $t$
distribution with five degrees of freedom, location
$\boldsymbol{\mu}$, and scale matrix $\Sigma$; the scale is chosen so
that each coordinate has standard deviation $0.15$. The second
alternative is skewed: each object is its mean plus a centered
exponential displacement,
\begin{equation*}
\text{Cluster } k
= \Big\{\, \boldsymbol{c}_k + 0.15\,(\boldsymbol{E} - \boldsymbol{1}),\
E_1, \ldots, E_d \overset{iid}{\sim} \mathrm{Exp}(1),\
\boldsymbol{c}_k \sim N\big(\boldsymbol{a}_k,\ 0.25^2\, \bI_{d\times d}\big)
\,\Big\}, \qquad k = 1, \ldots, K^\star, \ldots,
\end{equation*}
again with coordinate standard deviation $0.15$. The third
alternative is bimodal: each object is a two-component Gaussian
mixture,
\begin{equation*}
\text{Cluster } k
= \Big\{\, \tfrac12\, N\big(\boldsymbol{c}_k - \delta\boldsymbol{e}_1,\
\sigma_0^2\, \bI_{d\times d}\big)
+ \tfrac12\, N\big(\boldsymbol{c}_k + \delta\boldsymbol{e}_1,\
\sigma_0^2\, \bI_{d\times d}\big),\
\boldsymbol{c}_k \sim N\big(\boldsymbol{a}_k,\ 0.25^2\, \bI_{d\times d}\big)
\,\Big\}, \qquad k = 1, \ldots, K^\star, \ldots,
\end{equation*}
with $(\delta, \sigma_0) = (0.15,\ 0.08)$. All configurations use
$n_k = 40$ objects per cluster and $20$ replications, with the same
center geometries, candidate sets, and clustering methods as in the
Gaussian combinations. As Table~\ref{tab:shapes-values} shows, all
three criteria select $K^\star$ in all twenty-seven
combinations.

The Gaussian results repeat in every family. Within a family, the
optimal RWSD value is nearly the same across the four center
geometries, for instance $.372$, $.371$, $.371$, and $.372$ for the
segment, triangle, square, and circle in the Student-$t_5$ family, and
it stays close to the Gaussian value of $.379$ in
Table~\ref{tab:regular-values}. In the dimension sweeps the optimal
values decrease slightly, as in the Gaussian case, without changing
any selection, and in the sweeps over $m$ they barely move. Heavy
tails, skewness, and bimodality inside the objects therefore leave the
selection unchanged. This differs from the heavy-tailed setting of
Section~\ref{sec:simulations}, where the heavy tails are in the
positions of the objects rather than inside them.

\begin{table}[!t]
    \centering
    \caption{Criterion profiles over the candidate sets for the nine combinations of each non-Gaussian shape family: means over $20$ replications with $\pm$ one standard deviation. Boldface marks the optimum of each criterion, which lies at $K^\star$ in every combination, every family, and each of the $20$ replications.}
    \label{tab:shapes-values}
    \resizebox{\textwidth}{!}{%
    \setlength{\tabcolsep}{3pt}
    \begin{tabular}{ccc|cccccccccc}
        \toprule
        & & & & \multicolumn{3}{c}{Student-$t_5$}
        & \multicolumn{3}{c}{Centered exp.}
        & \multicolumn{3}{c}{Gaussian mixture} \\
        \cmidrule(lr){5-7}\cmidrule(lr){8-10}\cmidrule(lr){11-13}
        $d$ & $m$ & $K^\star$ & $K$
        & Silh. & RWSD & DBI & Silh. & RWSD & DBI & Silh. & RWSD & DBI \\
        \midrule
        \multirow{3}{*}{2} & \multirow{3}{*}{50} & \multirow{3}{*}{2} & 2 & \textbf{.847 ($\pm$.009)} & \textbf{.372 ($\pm$.003)} & \textbf{.214 ($\pm$.012)} & \textbf{.850 ($\pm$.005)} & \textbf{.371 ($\pm$.002)} & \textbf{.210 ($\pm$.007)} & \textbf{.852 ($\pm$.009)} & \textbf{.378 ($\pm$.003)} & \textbf{.207 ($\pm$.013)} \\
         & &  & 3 & .588 ($\pm$.030) & .310 ($\pm$.010) & .809 ($\pm$.094) & .596 ($\pm$.025) & .311 ($\pm$.009) & .790 ($\pm$.092) & .600 ($\pm$.022) & .322 ($\pm$.009) & .770 ($\pm$.083) \\
         & &  & 4 & .390 ($\pm$.114) & .262 ($\pm$.029) & 1.039 ($\pm$.179) & .372 ($\pm$.087) & .258 ($\pm$.022) & 1.042 ($\pm$.144) & .364 ($\pm$.058) & .267 ($\pm$.016) & 1.047 ($\pm$.104) \\
        \midrule
        \multirow{3}{*}{2} & \multirow{3}{*}{50} & \multirow{3}{*}{3} & 2 & .561 ($\pm$.010) & .357 ($\pm$.003) & .705 ($\pm$.021) & .558 ($\pm$.034) & .353 ($\pm$.012) & .724 ($\pm$.110) & .564 ($\pm$.008) & .361 ($\pm$.004) & .704 ($\pm$.020) \\
         & &  & 3 & \textbf{.841 ($\pm$.008)} & \textbf{.371 ($\pm$.003)} & \textbf{.220 ($\pm$.011)} & \textbf{.844 ($\pm$.006)} & \textbf{.371 ($\pm$.002)} & \textbf{.215 ($\pm$.009)} & \textbf{.848 ($\pm$.006)} & \textbf{.378 ($\pm$.003)} & \textbf{.211 ($\pm$.010)} \\
         & &  & 4 & .672 ($\pm$.018) & .332 ($\pm$.008) & .642 ($\pm$.071) & .681 ($\pm$.015) & .332 ($\pm$.006) & .629 ($\pm$.066) & .682 ($\pm$.018) & .341 ($\pm$.008) & .625 ($\pm$.082) \\
        \midrule
        \multirow{5}{*}{2} & \multirow{5}{*}{50} & \multirow{5}{*}{4} & 2 & .501 ($\pm$.043) & .340 ($\pm$.004) & .982 ($\pm$.066) & .481 ($\pm$.054) & .340 ($\pm$.003) & .940 ($\pm$.096) & .474 ($\pm$.061) & .342 ($\pm$.006) & .944 ($\pm$.082) \\
         & &  & 3 & .634 ($\pm$.007) & .357 ($\pm$.002) & .547 ($\pm$.010) & .635 ($\pm$.005) & .357 ($\pm$.003) & .546 ($\pm$.009) & .639 ($\pm$.007) & .362 ($\pm$.003) & .542 ($\pm$.009) \\
         & &  & 4 & \textbf{.837 ($\pm$.007)} & \textbf{.371 ($\pm$.002)} & \textbf{.222 ($\pm$.010)} & \textbf{.842 ($\pm$.003)} & \textbf{.371 ($\pm$.002)} & \textbf{.215 ($\pm$.006)} & \textbf{.847 ($\pm$.006)} & \textbf{.378 ($\pm$.002)} & \textbf{.210 ($\pm$.008)} \\
         & &  & 5 & .707 ($\pm$.008) & .342 ($\pm$.004) & .571 ($\pm$.029) & .718 ($\pm$.010) & .344 ($\pm$.008) & .542 ($\pm$.078) & .720 ($\pm$.015) & .354 ($\pm$.008) & .533 ($\pm$.064) \\
         & &  & 6 & .587 ($\pm$.030) & .314 ($\pm$.008) & .793 ($\pm$.075) & .591 ($\pm$.019) & .314 ($\pm$.012) & .782 ($\pm$.091) & .599 ($\pm$.031) & .326 ($\pm$.009) & .761 ($\pm$.066) \\
        \midrule
        \multirow{5}{*}{2} & \multirow{5}{*}{50} & \multirow{5}{*}{5} & 2 & .459 ($\pm$.012) & .331 ($\pm$.006) & .894 ($\pm$.013) & .451 ($\pm$.020) & .326 ($\pm$.009) & .900 ($\pm$.012) & .455 ($\pm$.017) & .331 ($\pm$.007) & .898 ($\pm$.012) \\
         & &  & 3 & .553 ($\pm$.024) & .359 ($\pm$.004) & .684 ($\pm$.020) & .562 ($\pm$.013) & .360 ($\pm$.002) & .680 ($\pm$.010) & .543 ($\pm$.047) & .361 ($\pm$.009) & .683 ($\pm$.019) \\
         & &  & 4 & .661 ($\pm$.026) & .360 ($\pm$.005) & .470 ($\pm$.014) & .654 ($\pm$.031) & .357 ($\pm$.006) & .471 ($\pm$.010) & .668 ($\pm$.025) & .365 ($\pm$.005) & .465 ($\pm$.011) \\
         & &  & 5 & \textbf{.833 ($\pm$.006)} & \textbf{.372 ($\pm$.002)} & \textbf{.226 ($\pm$.008)} & \textbf{.838 ($\pm$.004)} & \textbf{.371 ($\pm$.002)} & \textbf{.219 ($\pm$.007)} & \textbf{.841 ($\pm$.005)} & \textbf{.378 ($\pm$.002)} & \textbf{.216 ($\pm$.008)} \\
         & &  & 6 & .731 ($\pm$.007) & .347 ($\pm$.004) & .531 ($\pm$.039) & .732 ($\pm$.008) & .345 ($\pm$.004) & .531 ($\pm$.040) & .741 ($\pm$.010) & .356 ($\pm$.005) & .506 ($\pm$.040) \\
        \midrule
        \multirow{3}{*}{3} & \multirow{3}{*}{50} & \multirow{3}{*}{3} & 2 & .559 ($\pm$.006) & .336 ($\pm$.003) & .714 ($\pm$.014) & .558 ($\pm$.006) & .334 ($\pm$.003) & .716 ($\pm$.013) & .564 ($\pm$.006) & .337 ($\pm$.003) & .704 ($\pm$.013) \\
         & &  & 3 & \textbf{.840 ($\pm$.006)} & \textbf{.349 ($\pm$.002)} & \textbf{.226 ($\pm$.008)} & \textbf{.840 ($\pm$.005)} & \textbf{.347 ($\pm$.002)} & \textbf{.223 ($\pm$.009)} & \textbf{.843 ($\pm$.007)} & \textbf{.348 ($\pm$.003)} & \textbf{.220 ($\pm$.010)} \\
         & &  & 4 & .634 ($\pm$.014) & .293 ($\pm$.008) & .864 ($\pm$.098) & .638 ($\pm$.015) & .294 ($\pm$.007) & .826 ($\pm$.070) & .648 ($\pm$.014) & .302 ($\pm$.014) & .763 ($\pm$.122) \\
        \midrule
        \multirow{3}{*}{5} & \multirow{3}{*}{50} & \multirow{3}{*}{3} & 2 & .556 ($\pm$.007) & .320 ($\pm$.003) & .722 ($\pm$.015) & .556 ($\pm$.006) & .317 ($\pm$.003) & .720 ($\pm$.014) & .557 ($\pm$.005) & .315 ($\pm$.003) & .720 ($\pm$.010) \\
         & &  & 3 & \textbf{.833 ($\pm$.006)} & \textbf{.331 ($\pm$.002)} & \textbf{.237 ($\pm$.008)} & \textbf{.835 ($\pm$.005)} & \textbf{.327 ($\pm$.002)} & \textbf{.233 ($\pm$.008)} & \textbf{.838 ($\pm$.005)} & \textbf{.323 ($\pm$.003)} & \textbf{.228 ($\pm$.007)} \\
         & &  & 4 & .611 ($\pm$.012) & .271 ($\pm$.023) & .994 ($\pm$.212) & .609 ($\pm$.011) & .260 ($\pm$.007) & 1.068 ($\pm$.080) & .618 ($\pm$.016) & .266 ($\pm$.016) & .962 ($\pm$.173) \\
        \midrule
        \multirow{3}{*}{10} & \multirow{3}{*}{50} & \multirow{3}{*}{3} & 2 & .549 ($\pm$.004) & .304 ($\pm$.001) & .739 ($\pm$.011) & .551 ($\pm$.003) & .304 ($\pm$.001) & .733 ($\pm$.008) & .555 ($\pm$.003) & .299 ($\pm$.001) & .724 ($\pm$.008) \\
         & &  & 3 & \textbf{.826 ($\pm$.003)} & \textbf{.316 ($\pm$.001)} & \textbf{.254 ($\pm$.004)} & \textbf{.828 ($\pm$.002)} & \textbf{.313 ($\pm$.001)} & \textbf{.246 ($\pm$.004)} & \textbf{.836 ($\pm$.003)} & \textbf{.304 ($\pm$.001)} & \textbf{.232 ($\pm$.005)} \\
         & &  & 4 & .581 ($\pm$.008) & .242 ($\pm$.013) & 1.354 ($\pm$.202) & .583 ($\pm$.009) & .240 ($\pm$.010) & 1.323 ($\pm$.159) & .590 ($\pm$.012) & .233 ($\pm$.010) & 1.290 ($\pm$.138) \\
        \midrule
        \multirow{3}{*}{2} & \multirow{3}{*}{100} & \multirow{3}{*}{3} & 2 & .565 ($\pm$.007) & .359 ($\pm$.004) & .704 ($\pm$.017) & .564 ($\pm$.010) & .357 ($\pm$.003) & .701 ($\pm$.020) & .564 ($\pm$.008) & .363 ($\pm$.004) & .703 ($\pm$.017) \\
         & &  & 3 & \textbf{.847 ($\pm$.010)} & \textbf{.374 ($\pm$.002)} & \textbf{.213 ($\pm$.014)} & \textbf{.845 ($\pm$.008)} & \textbf{.373 ($\pm$.003)} & \textbf{.215 ($\pm$.010)} & \textbf{.845 ($\pm$.007)} & \textbf{.380 ($\pm$.002)} & \textbf{.213 ($\pm$.010)} \\
         & &  & 4 & .678 ($\pm$.015) & .336 ($\pm$.007) & .636 ($\pm$.068) & .674 ($\pm$.018) & .336 ($\pm$.007) & .656 ($\pm$.067) & .674 ($\pm$.014) & .342 ($\pm$.009) & .645 ($\pm$.070) \\
        \midrule
        \multirow{3}{*}{2} & \multirow{3}{*}{150} & \multirow{3}{*}{3} & 2 & .565 ($\pm$.010) & .358 ($\pm$.004) & .703 ($\pm$.022) & .564 ($\pm$.010) & .360 ($\pm$.003) & .702 ($\pm$.023) & .563 ($\pm$.007) & .363 ($\pm$.003) & .705 ($\pm$.015) \\
         & &  & 3 & \textbf{.847 ($\pm$.007)} & \textbf{.374 ($\pm$.003)} & \textbf{.213 ($\pm$.009)} & \textbf{.846 ($\pm$.006)} & \textbf{.375 ($\pm$.002)} & \textbf{.214 ($\pm$.009)} & \textbf{.845 ($\pm$.006)} & \textbf{.381 ($\pm$.003)} & \textbf{.214 ($\pm$.009)} \\
         & &  & 4 & .671 ($\pm$.012) & .335 ($\pm$.004) & .670 ($\pm$.041) & .672 ($\pm$.016) & .338 ($\pm$.007) & .653 ($\pm$.063) & .676 ($\pm$.017) & .346 ($\pm$.007) & .637 ($\pm$.073) \\
        \bottomrule
    \end{tabular}}
\end{table}

\subsection{Failure settings}\label{app:failure}
The machinery of WSD is to detect a center-outward ordering in Wasserstein geometry (geodesics). If a cluster is not geodesically convex, WSD may fail. Therefore, RWSD does not work well in the case where a cluster is geodesically nonconvex. We give two case where RWSD fails.  Setting I gives a location family where the centers are distributed as a hollow. Setting II gives a centered Gaussian family with covariance matrices distributed as a ``hollow''. 
%RWSD rewards partitions in which members lie deep inside their own clusters, so it relies on each cluster having a depth core, a group of members near the center of the cluster. The two settings below examine what happens when this core is missing. In both, the parameters of the members of each cluster lie on a circle, so every member sits on the boundary of its cluster and none lies near its center; we call such clusters \emph{hollow}. The circle is placed in the space of means in Setting~I and in the space of covariance matrices in Setting~II. The clusters remain balanced, equally shaped, and well separated, which is the situation in which the silhouette score and DBI work well. 
In both settings we use $K^\star = 3$, $n_k = 30$, and $m = 50$. We employ free-support
Wasserstein $K$-means over $K \in \{2, \ldots, 6\}$ with the same pipeline and hyperparameters as in Section~\ref{sec:simulations}, and run $20$ independent replica in each setting.

\smallskip
\noindent\textit{Setting I: hollow rings in location space.}
The true clusters are
\begin{equation*}
\text{Cluster } k
= \Big\{\, N\big(\boldsymbol{c},\ 0.08^2\, \bI_{2\times 2}\big),\
\boldsymbol{c} = \boldsymbol{a}_k
+ r\,(\cos\theta,\ \sin\theta)^\top,\
\theta \sim \mathrm{Unif}[0, 2\pi), \ r = 1.4 \,\Big\}
\end{equation*}
for $k = 1, 2, 3$, where $\boldsymbol{a}_1, \boldsymbol{a}_2,
\boldsymbol{a}_3$ are the vertices of an equilateral triangle with side $L = 4$. The means of the members of cluster $k$ are therefore drawn uniformly from the circle of radius $1.4$ centered at
$\boldsymbol{a}_k$.

%The resulting depths can be computed exactly. All members have the same shape and differ only in their means, so the optimal transport map between two members is a translation, and each unit transport direction is the unit vector between their means. The depth of a member is therefore the spatial depth of its mean among the means of the cluster. For a point on a circle, the unit vectors to points spread uniformly over the circle average to a vector of length $2/\pi$, so every member has own-cluster depth $1 - 2/\pi$. In a cluster with a core, by contrast, the central members have depth close to one. At the true partition RWSD therefore has no central members to reward, and its value $.307$ at $K = 3$ is roughly this depth minus the members' largest depth in another cluster, which is small. Splitting a ring into short arcs changes the picture: each arc is close to a line segment, and its middle members are deep within the arc, so partitions with more clusters tend to score higher. 
Since the clusters separate well, and silhouette score and DBI depend largely on distances. Therefore they perform well in these two settings, as shown in  Table~\ref{tab:failure-rings}. 
The RWSD values, in contrast, vary a little across different $K$, and pick the wrong $K$.

\begin{table}[H]
    \centering
    \caption{Failure setting I (hollow rings in location space,
    $K^\star = 3$): mean criterion values over $20$ replications with
    $\pm$ one standard deviation; boldface marks the optimum of each criterion. The silhouette score and DBI recover $K^\star = 3$, while RWSD selects $K^\star = 6$.}
    \label{tab:failure-rings}
    {\small
    \setlength{\tabcolsep}{3pt}
    \begin{tabular}{lccc}
        \toprule
        $K$ & Silhouette & RWSD & DBI \\
        \midrule
        2 & .383 ($\pm$.032) & .304 ($\pm$.018) & 1.039 ($\pm$.120) \\
        3 & \textbf{.510 ($\pm$.024)} & .307 ($\pm$.007) & \textbf{.701 ($\pm$.033)} \\
        4 & .446 ($\pm$.035) & .301 ($\pm$.018) & .810 ($\pm$.071) \\
        5 & .437 ($\pm$.029) & .313 ($\pm$.010) & .835 ($\pm$.060) \\
        6 & .453 ($\pm$.027) & \textbf{.323 ($\pm$.010)} & .799 ($\pm$.080) \\
        \midrule
        $K^\star$ & 3 & 6 & 3 \\
        \bottomrule
    \end{tabular}
    }
\end{table}

% \smallskip
% \noindent\textit{Setting II: hollow rings in covariance space.}
% Here each member is a centered Gaussian with diagonal covariance, and
% the standard deviations $(a, b)$ of the members of a cluster lie on a
% circle:
% \begin{equation*}
% \text{Cluster } k
% = \Big\{\, N\big(\boldsymbol{0},\ \Sigma(a, b)\big),\
% (a, b) = \boldsymbol{\rho}_k + r\,(\cos\varphi,\ \sin\varphi)^\top,\
% \varphi \sim \mathrm{Unif}[0, 2\pi),\ r = 1.4 \,\Big\}
% \end{equation*}
% for $k = 1, 2, 3$, where $\Sigma(a, b) := \mathrm{diag}(a^2, b^2)$, and
% $\boldsymbol{\rho}_1 = (3.2,\, 3.2)$,
% $\boldsymbol{\rho}_2 = (7.2,\, 3.2)$, and
% $\boldsymbol{\rho}_3 = (5.2,\, 3.2 + 2\sqrt{3})$ form an equilateral
% triangle with side $L = 4$ in the $(a, b)$ plane. For such Gaussians
% the optimal transport map between two members is linear and diagonal,
% and the Wasserstein distance equals the Euclidean distance between
% their $(a, b)$ parameters \citep{gelbrich1990formula}. A short
% computation with these maps shows that the depth of a member equals
% the spatial depth of its $(a, b)$ parameters among those of the
% cluster, so the computation of Setting~I applies unchanged: each
% cluster is again a hollow ring, now in covariance space, and RWSD has
% no central members to reward. Table~\ref{tab:failure-cov} shows the
% outcome. The silhouette score and DBI recover $K^\star = 3$, while the
% RWSD profile is nearly flat between $K = 2$ and $K = 3$ and peaks at
% $K^\star = 2$, which merges two clusters.

\smallskip
\noindent\textit{Setting II: hollow rings in covariance space.}
The true clusters are
\begin{equation*}
\text{Cluster } k
= \Big\{\, N\big(\boldsymbol{0},\ \Sigma(a, b)\big),\
(a, b) = \boldsymbol{\rho}_k + r\,(\cos\varphi,\ \sin\varphi)^\top,\
\varphi \sim \mathrm{Unif}[0, 2\pi),\ r = 1.4 \,\Big\}
\end{equation*}
for $k = 1, 2, 3$, where $\Sigma(a, b) := \mathrm{diag}(a^2, b^2)$, and
$\boldsymbol{\rho}_1 = (3.2,\, 3.2)$, $\boldsymbol{\rho}_2 = (7.2,\, 3.2)$,
and $\boldsymbol{\rho}_3 = (5.2,\, 3.2 + 2\sqrt{3})$ are the vertices of
an equilateral triangle with side $L = 4$ in the $(a, b)$ plane. The
standard deviations $(a, b)$ of the members of cluster $k$ are
therefore drawn uniformly from the circle of radius $1.4$ centered at
$\boldsymbol{\rho}_k$. For centered Gaussians, the Wasserstein distance is the
Bures--Wasserstein distance between the covariance matrices
\citep{gelbrich1990formula, bhatia2019bures}:
$
\mathcal{W}_2^2\big(N(\boldsymbol{0}, \Sigma_1),\, N(\boldsymbol{0}, \Sigma_2)\big)
= \operatorname{tr}(\Sigma_1) + \operatorname{tr}(\Sigma_2)
- 2\operatorname{tr}\big( (\Sigma_1^{1/2} \Sigma_2 \Sigma_1^{1/2})^{1/2} \big),
$
In our case: $\Sigma_1 = \Sigma(a, b)$ and $\Sigma_2 = \Sigma(a', b')$
reduces to
\[
\mathcal{W}_2^2\big(N(\boldsymbol{0}, \Sigma(a, b)),\, N(\boldsymbol{0}, \Sigma(a', b'))\big)
= (a - a')^2 + (b - b')^2 .
\]
Each cluster therefore forms a hollow ring in Wasserstein space, as in
Setting~I.

As shown in Table~\ref{tab:failure-cov}, the silhouette score and DBI
again recover $K^\star = 3$. The RWSD values, in contrast, barely
change between $K = 2$ and $K = 3$, and RWSD picks the wrong $K = 2$,
which merges two clusters.

\begin{table}[H]
    \centering
    \caption{Failure setting II (hollow rings in covariance space,
    $K^\star = 3$): mean criterion values over $20$ replications with
    $\pm$ one standard deviation; boldface marks the optimum of each criterion. The silhouette score and DBI recover $K^\star = 3$, while RWSD selects $K^\star = 2$.}
    \label{tab:failure-cov}
    {\small
    \setlength{\tabcolsep}{3pt}
    \begin{tabular}{lccc}
        \toprule
        $K$ & Silhouette & RWSD & DBI \\
        \midrule
        2 & .229 ($\pm$.026) & \textbf{.171 ($\pm$.014)} & 1.548 ($\pm$.093) \\
        3 & \textbf{.262 ($\pm$.031)} & .170 ($\pm$.015) & \textbf{1.295 ($\pm$.100)} \\
        4 & .206 ($\pm$.027) & .126 ($\pm$.019) & 1.587 ($\pm$.087) \\
        5 & .164 ($\pm$.023) & .106 ($\pm$.015) & 1.626 ($\pm$.085) \\
        6 & .147 ($\pm$.020) & .088 ($\pm$.012) & 1.706 ($\pm$.099) \\
        \midrule
        $K^\star$ & 3 & 2 & 3 \\
        \bottomrule
    \end{tabular}
    }
\end{table}

\smallskip
\noindent In both settings, no distribution lies near the center of
its own cluster, so every distribution has a small depth within its
own cluster. When implementing RWSD, a practitioner is encouraged to check whether each distribution has a small WSD value within its own cluster. If every depth is small, the $K$ selected by RWSD should be treated with caution, and a distance-based criterion such as the silhouette score may be more reliable.

\section{Additional details in real data analysis}\label{app:experiment-details}
This section collects supplementary details for the real data analysis in Section~\ref{sec:real-data}, including an illustration of the MNIST distributional embedding and the full preprocessing 
pipeline for the AML flow cytometry cohort.

\subsection{Illustration of MNIST distributional embedding}\label{app:mnist-illustration}

Each grayscale MNIST handwritten digit image of size $28\times28$ is treated as a discrete empirical measure over pixel coordinates with intensities normalized to sum to one,
\[
\mu_i = \sum_{p=1}^{784} w_{ip}\,\delta_{x_p},
\qquad
w_{ip} = \frac{I_{ip}}{\sum_{q=1}^{784} I_{iq}},
\]
where $x_p \in [0,1]^2$ is the normalized pixel position 
and $I_{ip}$ the grayscale intensity at pixel $p$ in image $i$.  
This transforms each image into a probability distribution 
$\mu_i \in \mathcal{P}_2([0,1]^2)$. 

\subsection{AML cytometry preprocessing details}\label{app:aml-preprocessing}

We provide the full preprocessing pipeline for the AML cytometry experiment in Section~\ref{subsec:aml-cytof}.
The P1 panel of eight markers (CD7, CD56, CD34, CD117, CD33, HLA-DR, CD13, CD45) are included in our analysis, which capture progenitor signatures, myeloid lineage, and AML-aberrant antigen expression. 
Each BLAST110 sample contains hundreds of thousands of single-cell 
flow cytometry profiles. %Because the dataset distributes pre-cleaned P1 files (with compensation and basic gating performed in the official workflow), our preprocessing incorporated only marker-level transformations. 
Fluorescence intensities were first transformed using 
the hyperbolic arcsine with cofactor $c = 5$,
\[
x^{\mathrm{trans}} = \mathrm{arcsinh}(x / c),
\]
followed by per-marker $z$-score normalization computed across the pooled cells of all samples. This stabilizes variance across fluorescence channels and ensures comparability among samples.
To embed each sample into Wasserstein space, we represented its single-cell data as an empirical distribution,
\[
\mu_i = \frac{1}{n_i} \sum_{j=1}^{n_i} \delta_{x_j}.
\]
%Direct computation of $\mathcal{W}_2$ on full-resolution data is computationally prohibitive due to the large number of cells per sample. 
To make the computations feasible, we randomly selected $m = 1000$ cells from each sample without replacement and 
constructed the empirical distribution
\[
\mu_i^{(m)} = \frac{1}{m} \sum_{j=1}^{m} \delta_{x_j^{(i)}}.
\]
Subsampling at this scale has been shown to preserve major phenotypic structures in 
high-dimensional cytometry data while greatly reducing the computational cost of OT step \citep{samusik2016automated,levine2015data}.

\section{Choices of comparison}
\label{app:baseline-choice}

We compare RWSD with the silhouette score and the Davies and Bouldin index because both are standard internal validation criteria and can be directly computed from the Wasserstein clustering results or pairwise Wasserstein distance matrix. 
This facilitates fair comparisons because the three indices are computed from the same Wasserstein clustering results. 
Other criteria are either computational expensive in Wasserstein space, or not directly comparable to our method. 
The gap statistic requires a reference null distribution, which is difficult to choose in Wasserstein space.
Stability based criteria (such as prediction strength) require
repeated sample splitting, re-clustering, and recomputation of OT distances, which is prohibitive for large set of distributions. 
We therefore focus on baselines that can be implemented without introducing additional parameters or complicated clustering pipeline.

%\section{Computational complexity, resources, and asset licenses}
\section{Implementation Details}
\label{app:compute_license}

\subsection{Computational complexity}
\label{app:complexity}

We summarize the computational cost of the main steps of our method.
Throughout, $n$ is the number of empirical distributions to be
clustered, each supported on $m$ points in $\mathbb{R}^d$; $K$ is the
number of clusters; and $|\mathcal{K}|$ is the number of candidate
values of $K$ in Algorithm~\ref{alg:selectK}. For each candidate $K$,
Algorithm~\ref{alg:selectK} first runs Wasserstein $K$-centers, which
we call the clustering stage, and then computes RWSD for the resulting
partition, which we call the validation stage. Our implementation uses
two kinds of optimal transport (OT) computations. The empirical WSD and
RWSD use exact OT, computed with \texttt{ot.emd} from the Python
Optimal Transport library \citep{flamary2021pot}, which solves the
discrete transport linear program by the network simplex method. For
speed, the pairwise distance matrix, the PAM warm start, the
Wasserstein $K$-centers, and the DBI use entropic OT or the Sinkhorn
divergence instead \citep{cuturi2013sinkhorn,cuturi2014fast}.

\smallskip
\noindent\textit{Exact OT for the empirical WSD and RWSD.}
Let $C_{\mathrm{EMD}}(m)$ be the cost of solving an exact OT problem
between two empirical measures with $m$ support points each. For each
pair $(\hat\mu_{q,m},\hat\mu_{i,m})$, the empirical WSD first computes
the optimal plan $\pi_{q,i}$ with \texttt{ot.emd} and then forms the
transport directions from $\pi_{q,i}$, at cost $\mathcal{O}(m^2d)$ per
pair. Over all pairs of the $n$ distributions, computing the transport
directions therefore costs, in the worst case,
\begin{equation*}
    \mathcal{O}\!\left(n^2 C_{\mathrm{EMD}}(m) + n^2m^2d\right).
\end{equation*}

\smallskip
\noindent\textit{Entropic OT for the pairwise distance matrix.}
Let $C_{\mathrm{Sink}}(m,\varepsilon)$ be the cost of computing the
entropic OT value between two empirical measures with $m$ support
points each at regularization level $\varepsilon$. Following the
complexity bounds of \citet{dvurechensky2018computational}, we take
\begin{equation*}
    C_{\mathrm{Sink}}(m,\varepsilon)
    = \mathcal{O}\!\left(\frac{m^2}{\varepsilon^2}\right),
\end{equation*}
up to constants that depend on the target accuracy. The silhouette
score, the MDS visualization, and the PAM warm start use the same
pairwise distance matrix, whose $(q,i)$ entry is the square root of the
Sinkhorn divergence between $\hat\mu_{q,m}$ and $\hat\mu_{i,m}$. Each
Sinkhorn divergence requires three entropic OT computations, so the
full matrix costs
\begin{equation*}
    \mathcal{O}\!\left(n^2 C_{\mathrm{Sink}}(m,\varepsilon)\right).
\end{equation*}
The distance matrix depends only on the data and on $\varepsilon$, not
on $K$, so the distance matrix is computed once for all candidate
values of $K$.

\smallskip
\noindent\textit{RWSD for a fixed partition.}
For a fixed partition into $K$ clusters, RWSD computes the WSD of each
empirical distribution with respect to its own cluster and to every
other cluster. In the worst case, the WSD values require exact OT
between all ordered pairs of distributions, which costs
\begin{equation*}
    \mathcal{O}\!\left(n^2 C_{\mathrm{EMD}}(m) + n^2m^2d\right).
\end{equation*}
Comparing the depth of each distribution in its own cluster with its
largest depth in the other clusters adds only $\mathcal{O}(nK)$
operations.

\smallskip
\noindent\textit{Storing exact OT plans across candidate values of $K$.}
The optimal plan $\pi_{q,i}$ depends only on the two distributions
$\hat\mu_{q,m}$ and $\hat\mu_{i,m}$, not on $K$ or on the cluster
labels. Moreover, $\pi_{i,q}=\pi_{q,i}^{\top}$, so one exact solve gives
the plans for both $(q,i)$ and $(i,q)$. The validation stage visits all
$n(n-1)$ ordered pairs for every candidate $K$, therefore recomputing the
plans requires $|\mathcal{K}|\,n(n-1)$ exact solves in total. Solving
each unordered pair once and storing the plan reduces the number of
solves to $n(n-1)/2$, a reduction by a factor of $2|\mathcal{K}|$.
Every later depth evaluation then rebuilds the transport directions
from the stored plan at cost $\mathcal{O}(md)$ per ordered pair, so the
cost of exact OT in the validation stage becomes
\begin{equation*}
    \mathcal{O}\!\left(n^2 C_{\mathrm{EMD}}(m)+n^2m^2d\right)
    \;+\;|\mathcal{K}|\,\mathcal{O}\!\left(n^2md\right)
\end{equation*}
instead of $|\mathcal{K}|\,\mathcal{O}(n^2 C_{\mathrm{EMD}}(m)+n^2m^2d)$.
The dominant term is paid once rather than $|\mathcal{K}|$ times, and
the remaining cost per candidate $K$ is smaller by a factor of $m$.
Because the stored plans are the exact solutions, the RWSD values with
and without storage agree up to floating-point rounding, so the
selected $K$ is the same. Storing the plans helps only the validation
stage: the clustering stage solves transport problems between
distributions and centers, and the centers change across iterations and
across candidate values of $K$.

Each stored plan needs only $\mathcal{O}(m)$ memory, although the plan
is an $m\times m$ matrix. When all atoms have the same weight $1/m$, as
in our simulations and the cytometry data, the feasible transport plans
form $m^{-1}$ times the Birkhoff polytope, whose vertices are the
permutation matrices. The network simplex method returns a vertex, so
$\pi_{q,i}=m^{-1}P_{\sigma_{q,i}}$ for a permutation $\sigma_{q,i}$,
where $P_{\sigma}$ denotes the permutation matrix of $\sigma$. Each plan
is therefore a one-to-one matching of the atoms and is stored exactly
as $m$ integers. We verified that every plan in the simulations and the
cytometry data is a scaled permutation matrix. When the weights are
unequal, as for the MNIST images, the exact plan is still a vertex of
the transportation polytope and has at most $2m-1$ nonzero entries, so
the plan is stored as a sparse matrix with $\mathcal{O}(m)$ entries. In
both cases, the stored plans take $\mathcal{O}(n^2m)$ memory in total,
compared with $\mathcal{O}(n^2)$ for the distance matrix. The memory
needed while solving a single OT problem is the same with and without
storing the plans.

\smallskip
\noindent\textit{Measured effect and feasibility.}
We measured the effect of storing the plans on a regular design with
four clusters, $n=120$, $m=100$, $d=2$, and
$\mathcal{K}=\{2,\dots,6\}$, on ten CPU cores. Storing the plans
reduces the number of exact solves from $71\,400$ to $7\,140$, exactly
the factor $2|\mathcal{K}|=10$, and the time of the validation stage
from $17.1$ to $2.1$ seconds: $1.7$ seconds to compute the plans once
and $0.4$ seconds for the five passes that reuse the stored plans. The
time per solve is nearly the same in the two runs ($0.240$ and $0.244$
milliseconds), so the speed-up comes from the smaller number of solves
and not from differences in implementation. The two RWSD curves agree
to within $2\times10^{-10}$, and both select $K=4$. The memory grows
from $112$ KB for the distance matrix to $2.7$ MB for the stored
permutations.
We also measured the time to compute all plans once at the sizes of
our real data (Table~\ref{tab:ot_cache_scale}). For the MNIST images on
the common $28\times28$ grid with intensity weights ($n=90$, $m=784$),
the $4\,005$ exact solves take $50$ seconds on ten cores, at $100$
milliseconds per solve, and the sparse plans take about $100$ MB. For
the cytometry samples ($n=50$, $m=1000$, $d=8$, equal weights), the
$1\,225$ solves take $20$ seconds, at $162$ milliseconds per solve;
every plan is a permutation, and the stored plans take $4.7$ MB. For
the MNIST setting with $n=150$, scaling the measurements at $n=90$ by
the number of pairs gives about $140$ seconds and $280$ MB. In every
case, computing the plans once costs far less than recomputing the
plans for each candidate $K$.

Storing the plans has two limits. The first limit is memory. With a
memory budget of $B$ bytes, the stored plans fit when $2n^2m\le B$ with
equal weights (permutations stored with four bytes per entry), and
roughly when $16n^2m\le B$ with unequal weights (sparse plans stored
with sixteen bytes per nonzero entry). With the $5$ GB of memory per
CPU in Table~\ref{tab:computational_cost}, the two conditions allow
$n\lesssim1850$ with equal weights and $n\lesssim650$ with unequal
weights at $m=784$, and $n\lesssim1640$ with equal weights at $m=1000$.
The second limit is time. The measured time per solve grows roughly
like $m^{2.8}$ across our settings, from $0.24$ milliseconds at $m=100$
to $100$--$162$ milliseconds at $m=784$--$1000$, in line with the
superquadratic complexity of the network simplex method. For very large
$m$, computing the plans even once becomes the dominant cost, and
entropic OT is then preferable. Both limits are comfortably met in all
our experiments. Storing the exact plans is therefore a useful optional
speed-up of the validation stage. By default, we recompute the plans
and keep only the distance matrix, which needs less memory.

\begin{table}[t]
\caption{Memory of the stored exact OT plans and time to compute all plans once on ten CPU cores, at the sizes of our simulations and real data. All entries are measured except the row with $n=150$, which is scaled from the measurements at $n=90$.}
\label{tab:ot_cache_scale}
\centering
\begin{tabular}{|l|c|c|c|c|c|c|}
\hline
setting & $n$ & $m$ & $d$ & stored plans & memory & time \\ \hline
simulation      & $120$ & $100$  & $2$ & permutations & $2.7$ MB  & $1.7$ s \\ \hline
MNIST, 3 digits & $90$  & $784$  & $2$ & sparse plans & $100$ MB  & $50$ s \\ \hline
MNIST, 5 digits & $150$ & $784$  & $2$ & sparse plans & $280$ MB  & $140$ s \\ \hline
flow cytometry  & $50$  & $1000$ & $8$ & permutations & $4.7$ MB  & $20$ s \\ \hline
\end{tabular}
\end{table}

\smallskip
\noindent\textit{Wasserstein $K$-centers.}
Our Wasserstein $K$-means and $K$-medians use semi-free support centers
with a moderate number $s$ of support points, typically $s\ll m$. Each
assignment step computes the entropic OT cost between each of the $n$
empirical distributions and each of the $K$ current centers, which
costs
\begin{equation*}
    \mathcal{O}\!\left(nK\,C_{\mathrm{Sink}}(s,m,\varepsilon)\right),
\end{equation*}
where $C_{\mathrm{Sink}}(s,m,\varepsilon)$ is the Sinkhorn cost between
a center with $s$ support points and an empirical distribution with $m$
support points. Each center update runs a semi-free support entropic
barycenter step \citep{cuturi2014fast}, followed, for $K$-medians, by
IRLS iterations for the Wasserstein median
(Appendix~\ref{app:kmedian}). Let $T_{\mathrm{bar}}$ be the number of
inner barycenter iterations and $T_{\mathrm{IRLS}}$ the number of outer
IRLS iterations. Since the cluster sizes sum to $n$, one update of all
$K$ centers costs
\[
\mathrm{Update}_{K\text{-median}}
= \mathcal{O}\!\left(T_{\mathrm{IRLS}}\, T_{\mathrm{bar}}\, n\,
C_{\mathrm{Sink}}(s,m,\varepsilon)\right),
\]
and one update for $K$-means costs the same with $T_{\mathrm{IRLS}}=1$.
When used, the PAM warm start reuses the precomputed distance matrix
and adds $\mathcal{O}(n^2K)$ operations for the greedy medoid search.

\smallskip
\noindent\textit{Silhouette score and DBI.}
The silhouette score uses the precomputed distance matrix and needs no
further OT computation. The DBI uses the centers from the clustering
stage as cluster prototypes: barycenters for $K$-means and medians for
$K$-medians. The within-cluster scatter of the DBI costs
\begin{equation*}
    \mathcal{O}\!\left(n\, C_{\mathrm{Sink}}(s,m,\varepsilon)\right),
\end{equation*}
and the distances between the $K$ prototypes cost
\begin{equation*}
    \mathcal{O}\!\left(K^2 C_{\mathrm{Sink}}(s,s,\varepsilon)\right).
\end{equation*}

\smallskip
\noindent\textit{Total cost of Algorithm~\ref{alg:selectK}.}
Let $T$ be the average number of clustering iterations until
convergence. In the worst case, the total cost over all candidate
values of $K$ is
\begin{equation*}
\begin{aligned}
&
\underbrace{
\mathcal{O}\!\left(n^2 C_{\mathrm{Sink}}(m,\varepsilon)\right)
}_{\text{pairwise distance matrix}}
\\
&\quad
+
|\mathcal{K}|\,T\,
\underbrace{
\mathcal{O}\!\left(nK C_{\mathrm{Sink}}(s,m,\varepsilon)
+
\mathrm{Update}_{K\text{-median}}\right)
}_{\text{clustering with semi-free support centers}}
\\
&\quad
+
\underbrace{
\begin{cases}
|\mathcal{K}|\,
\mathcal{O}\!\left(n^2 C_{\mathrm{EMD}}(m)+n^2m^2d\right),
& \text{recomputing the exact plans,}\\[4pt]
\mathcal{O}\!\left(n^2 C_{\mathrm{EMD}}(m)+n^2m^2d\right)
+|\mathcal{K}|\,\mathcal{O}\!\left(n^2md\right),
& \text{storing the exact plans,}
\end{cases}
}_{\text{exact OT and transport directions for RWSD}}
\end{aligned}
\end{equation*}
and the memory kept across candidate values of $K$ is
$\mathcal{O}(n^2)$ when the exact plans are recomputed and
$\mathcal{O}(n^2m)$ when the exact plans are stored.

In practice, the most expensive steps are the pairwise Wasserstein
distances, the barycenter updates, and the median updates. To reduce
the cost of the three steps, we limit the barycenters and medians to a
moderate number of support points, use entropic OT, and compute the
medians by the IRLS procedure of Appendix~\ref{app:kmedian}. The cost
grows with the number of distributions, the number of support points
per distribution, and the number of candidate values of $K$.
Table~\ref{tab:computational_cost} reports the average runtime and
memory of our experiments. The experiments ran in parallel on $24$
CPUs, and each reported runtime is the wall-clock time of one
experiment.
\begin{table}[t]
\caption{Average runtime and memory of our experiments. Runtimes are wall-clock times with parallel computation on $24$ CPUs.}
\label{tab:computational_cost}
\centering
\begin{tabular}{|l|l|l|}
\hline
                          & simulation & real data analysis \\ \hline
runtime (per experiment)  & 1 hour     & 4 hours            \\ \hline
memory (per CPU)          & 5 GB       & 5 GB               \\ \hline
\end{tabular}
\end{table}

\subsection{Computational resources}\label{app:resources}
All simulations and real data analysis in this paper can be reproduced on CPUs only, and no GPU or TPU acceleration is required.

These simulations can be run on a common laptop (\textit{e.g.} Macbook Air). For reproducing the full set of synthetic experiments, we recommend a machine with $10$ CPUs and at least $16$ GB of memory. %More memory may be useful when computing and storing pairwise Wasserstein distance matrices for large $K$.

The MNIST experiments represent each $28\times 28$ grayscale image as a discrete probability distribution over pixel coordinates. Since each empirical distribution has support size at most $784$, the MNIST experiments can also be reproduced on CPU only. We recommend at least $16$ GB of memory for reproducing the MNIST clustering and validation results.

The AML flow cytometry experiment uses only the BLAST110 cohort from the open flow cytometry dataset of \citet{mocking2024computational}. %We do not use the LAIP29 cohort. Reproducing the AML experiment from raw files requires more storage than the synthetic and MNIST experiments. 
The raw BLAST110 archive is approximately $14.9$ GB. We recommend at least $25$ GB of free disk space for downloading and unpacking the raw data. We also recommend at least $64$ GB of memory if the raw FCS files are loaded or preprocessed in 
large batches.

\subsection{Software and data licenses}
\label{app:assets}

The implementation uses standard open source scientific computing packages. The analysis in this paper are conducted using Python.
 
No proprietary software is required to reproduce the reported experiments. Table~\ref{tab:asset_licenses} summarizes the existing datasets and software used in our experiments, together with the corresponding license or access terms.

\begin{table}[h]
\caption{Existing datasets and software used in the experiments. } %We credit the original creators and follow the corresponding license or access terms.
\label{tab:asset_licenses}
%\vspace{2mm}
\centering
\small
\begin{tabular}{p{0.22\linewidth}p{0.36\linewidth}p{0.32\linewidth}}
\toprule
Asset & Use in this paper & License or terms \\
\midrule
Synthetic data & Generated by the authors for controlled simulation studies. & No external data license applies. \\
MNIST handwritten digit database & Used as a real distribution-valued image dataset by normalizing each image into a discrete probability distribution over pixel coordinates. & Publicly available benchmark dataset used for research evaluation. We credit the original dataset creators and follow the dataset access terms. \\
BLAST110 AML flow cytometry cohort & Used as the real flow cytometry dataset for the AML versus healthy control experiment in Section~\ref{subsec:aml-cytof}. We use only BLAST110 and do not use LAIP29. & Open dataset associated with \citet{mocking2024computational}. The Zenodo record releases the data under the Creative Commons Attribution 4.0 International license. \\
NumPy and SciPy & Numerical linear algebra, array operations, and scientific computing utilities. & BSD license. \\
scikit-learn & Baseline utilities and clustering related routines where applicable. & BSD 3-Clause license. \\
POT, Python Optimal Transport & Optimal transport solvers, including entropic optimal transport routines where applicable. & MIT license. \\
Matplotlib and pandas & Visualization, table processing, and experiment result organization. & Matplotlib uses a BSD-compatible license, and pandas uses the BSD 3-Clause license. \\
\bottomrule
\end{tabular}
\end{table}

\clearpage

\end{document}